\documentclass{LMCS}

\def\doi{9(2:10)2013}
\lmcsheading%
{\doi}
{1--44}
{}
{}
{Nov.~\phantom01, 2012}
{Jun.~25, 2013}
{}

\usepackage{ifpdf}
\ifpdf
\usepackage{pdfsync}
\fi
\usepackage{amssymb}
\usepackage{amsfonts}
\usepackage{enumerate}
\usepackage{latexsym}
\usepackage{xspace}

\usepackage{tikz}
\usetikzlibrary{arrows,shapes.geometric,calc}

\usepackage{hyperref}

\hypersetup{
  pdfauthor={Martin Huschenbett, Alexander Kartzow, Jiamou Liu, Markus Lohrey},
  pdftitle={Tree-Automatic Well-Founded
    Trees},
  pdfkeywords={hyperarithmetical hierarchy, isomorphism problem,
    ordinal rank, tree-automatic structures, 
    well-founded trees}%
}

\newcommand{\A}{\mathcal A}
\newcommand{\Af}{\mathfrak A}
\newcommand{\ar}{\mathsf{ar}}
\newcommand{\bin}{\mathsf{bin}}
\newcommand{\B}{\mathcal B}
\newcommand{\Bf}{\mathfrak B}
\newcommand{\C}{\mathcal C}
\newcommand{\Cf}{\mathfrak C}
\newcommand{\cl}{\mathsf{cl}} 
\newcommand{\fr}{\mathsf{fr}}
\newcommand{\Con}[3]{\ensuremath{{#2}\equiv_{#1}{#3}}}
\newcommand{\ConBoxNull}[0]{\equiv}
\newcommand{\ConBox}[2]{\ensuremath{{#1}\ConBoxNull{#2}}}

\newcommand{\dom}{\mathsf{dom}}
\newcommand{\E}{\mathcal E}
\newcommand{\F}{\mathcal F}
\newcommand{\Ff}{\mathfrak F}
\newcommand{\fin}{\mathsf{fin}}
\newcommand{\forest}{\mathrm{forest}}
\newcommand{\FO}{\mathsf{FO}}
\newcommand{\FOext}[1][\empty]%
	{\FO\ifthenelse{\equal{#1}{\empty}}{}{[#1]}
	+\exists^\infty+\exists^{\mathsf{chain}}}

\newcommand{\Gf}{\mathfrak G}

\newcommand{\Hf}{\mathfrak H}
\newcommand{\Iso}{\mathrm{Iso}}
\newcommand{\iso}{\mathsf{iso}}
\newcommand{\lef}{\mathsf{lef}}
\newcommand{\leaves}{\mathsf{leaves}}

\newcommand{\N}{\mathbb N}

\newcommand{\K}{\mathcal K}

\renewcommand{\P}{{\mathbb P}}
\newcommand{\ps}[1]{\ensuremath{2^{#1}}}

\newcommand{\Qf}{\mathfrak Q}
\newcommand{\Pf}{\mathfrak P}
\newcommand{\pref}{\mathsf{pref}}

\newcommand{\Rel}{\mathcal{R}}
\newcommand{\Run}{\mathsf{Run}}
\newcommand{\rank}{\mathsf{rank}}

\newcommand{\Sf}{\mathfrak S}
\newcommand{\T}{\mathcal T}
\newcommand{\Tf}{\mathfrak T}
\newcommand{\tree}{\mathsf{tree}}
\newcommand{\trees}[1]{\T^{\fin}_{2,#1}}
\newcommand{\tp}{\mathsf{tp}}
\newcommand{\Uf}{\mathfrak U}
\newcommand{\unlabel}{\mathsf{unlabel}}
\newcommand{\Vf}{\mathfrak V}
\newcommand{\Wf}{\mathfrak W}
\newcommand{\word}{\mathsf{word}}
\newcommand{\wulpo}{wulpo\xspace}
\newcommand{\Xf}{\mathfrak{X}}
\newcommand{\Yf}{\mathfrak{Y}}

\newcommand{\Zf}{\mathfrak{Z}}

\begin{document}
\title[Tree-Automatic Well-Founded Trees]{Tree-Automatic Well-Founded
  Trees\rsuper*}

\author[M.~ Huschenbett]{Martin Huschenbett\rsuper a}
\address{{\lsuper a}Technische Universit{\"a}t Ilmenau, Germany}
\email{martin.huschenbett@tu-ilmenau.de}

\author[A.~Kartzow]{Alexander Kartzow\rsuper b}
\address{{\lsuper{b,d}}Universit{\"a}t Leipzig, Germany}
\email{\{kartzow, lohrey\}@informatik.uni-leipzig.de}

\author[J.~Liu]{Jiamou Liu\rsuper c}
\address{{\lsuper c}Auckland University of Technology, New Zealand}
\email{jiamou.liu@aut.ac.nz}

\author[M.~Lohrey]{Markus Lohrey\rsuper d} 

\thanks{{\lsuper{b,d}}The second and fourth author are supported by the DFG  
  research project GELO}

\subjclass{???}

\ACMCCS{[{\bf Theory of computation}]: Logic; Formal languages and
  automata theory---Tree languages}

\titlecomment{{\lsuper*}An extended abstract of this paper appeared at CiE 2012 
  \cite{KartzowLL12}
}  

\begin{abstract}
  We investigate tree-automatic well-founded trees.  Using
  Delhomm{\'{e}}'s decomposition technique for tree-automatic
  structures, we show that the (ordinal) $\rank$ of a tree-automatic
  well-founded tree is strictly below $\omega^\omega$.  Moreover, we
  make a step towards proving that the ranks of tree-automatic
  well-founded partial orders are bounded by $\omega^{\omega^\omega}$:
  we prove this bound for what we call upwards linear partial orders.
  As an application of our result, we show that the isomorphism
  problem for tree-automatic well-founded trees is complete for level
  $\Delta^0_{\omega^\omega}$ of the hyperarithmetical hierarchy with
  respect to Turing-reductions.
\end{abstract}

\maketitle

\section{Introduction}
Various classes of infinite but finitely presented structures received
a lot of attention in algorithmic model theory \cite{BaGrRu10}. Among
the most important such classes of structures is the class of {\em
  string-automatic structures} \cite{KhoN95}. A (relational) structure
is string-automatic if its universe is a regular set of words and all
relations can be recognized by synchronous multi-tape automata. During
the past 15 years a theory of string-automatic structures has emerged.
This theory was developed along two interrelated branches.
The first is a structural branch, which leads to (partial)
characterizations of particular classes of string-automatic structures
\cite{Delhomme04,KhMi09,KhoNRS07,KhoRS05,OlTh05}. The second is an
algorithmic branch, which leads to numerous decidability and
undecidability, as well as complexity results for important
algorithmic problems for string-automatic structures
\cite{BG04,KhoNRS07,KusLo11jsl}. One of the most fundamental results
for string-automatic structures states that their first-order theories
are uniformly decidable \cite{KhoN95}.

By replacing strings and string automata by trees and tree automata,
respectively, Blumensath \cite{Blu99} generalized string-automatic
structures to {\em tree-automatic structures} and proved that their
first-order theories are still uniformly decidable.  However compared
to string-automatic structures, the theory of tree-automatic
structures is less developed.  The only non-trivial characterization
of a class of tree-automatic structures we are aware of concerns
ordinals. In \cite{Delhomme04} Delhomm{\'{e}} proved that an ordinal
is tree-automatic if and only if it is strictly below
$\omega^{\omega^\omega}$.  Some complexity results for first-order
theories of tree-automatic structures are shown in
\cite{KusLo11jsl}. Recently, Huschenbett proved that it is decidable
whether a given tree-automatic scattered linear order is
string-automatic \cite{Huschenbett12}.

In this paper, we study tree-automatic well-founded
trees. We stress here that in this paper the term \emph{tree} always
refers to order trees $\Tf=(T, \leq)$, i.e., partial orders that have
the shape of a tree,  as opposed to successor trees where the edge
relation is not transitive (except for trivial cases). 
Using Delhomm{\'{e}}'s decomposition technique for tree-automatic
structures \cite{Delhomme04},
we show that the rank of a tree-automatic well-founded tree
is strictly below $\omega^\omega$. 
As a generalization of trees, we introduce the class of
upwards linear partial orders. These are orders where the elements
above any fixed element form a linear suborder. If such an order is
well-founded and tree-automatic, then its rank is strictly below
$\omega^{\omega^\omega}$. Unfortunately, we do not know if this bound
holds for all tree-automatic well-founded partial orders. 
We will show in Example 
\ref{exa:BoxDestroysRanksofWfPO} that there is no hope to extend
Delhomm\'{e}'s technique in order to obtain bounds on the ranks of all
tree-automatic well-founded partial orders. Thus, characterizing the ranks 
of well-founded partial orders requires the development of completely
new techniques.

We apply this result to the isomorphism problem for
tree-automatic well-founded trees. In \cite{KuLiLo11}, it was shown
that the isomorphism problem for string-automatic well-founded trees
is complete for level $\Delta^0_\omega$ of the  hyperarithmetical
hierarchy. In other words, the isomorphism problem for
string-automatic well-founded trees
is recursively equivalent to true arithmetic.
We show that the rank of well-founded computable trees determines
the complexity of the isomorphism problem in the following sense:
The isomorphism problem for well-founded computable trees of rank at
most $\lambda$ (where $\lambda$ is a
computable limit ordinal)
belongs to level $\Delta^0_{\lambda}$ of the hyperarithmetical
hierarchy. Since we know that the rank of a tree-automatic
well-founded tree is strictly below $\omega^\omega$,
this implies 
that the isomorphism problem for
tree-automatic well-founded trees belongs to level
$\Delta^0_{\omega^\omega}$
of the hyperarithmetical hierarchy. We also provide a corresponding
lower bound. We prove that
the isomorphism problem for tree-automatic well-founded trees is
$\Delta^0_{\omega^\omega}$-complete under Turing-reductions.

Let us remark that for non-well-founded order trees, the isomorphism
problem is complete for $\Sigma^1_1$ (the first existential level
of the analytical hierarchy) already in the string-automatic case
\cite{KuLiLo11}, and this complexity is in a certain sense maximal,
since the isomorphism problem for the class of all computable
structures is $\Sigma^1_1$-complete as well \cite{CaKni06,GonKn02}.
Let us also emphasize that our proof techniques do not work for
successor trees.

\section{Preliminaries}

We write $\N_{>0}$ for $\N\setminus\{0\}$.
For $M,N$ sets, $\ps{M}$ denotes the \emph{powerset} of $M$ and $N^M$
denotes the \emph{set of functions} from $M$ to $N$.
We denote by $M\sqcup N$ 
(by $\bigsqcup_{i\in N} M_i$, respectively) the \emph{disjoint union} of the
sets $M$ and $N$ ($(M_i)_{i\in N}$, respectively).
For any (partial) function $f$ we write $\dom(f)$ for the domain of $f$.

A \emph{(finite and relational) signature} $\tau=(\Rel,\ar)$ is
a finite set $\Rel$ of \emph{relation symbols}
together with a map $\ar:\Rel\to\N_{>0}$ assigning to each each $R\in\Rel$
its arity $\ar(R)$.
A \emph{$\tau$-structure} $\Sf=(S,(R^\Sf)_{R\in\Rel})$ consists
of a \emph{domain} $S$ and
an $\ar(R)$\nobreakdash-ary \emph{relation} $R^\Sf$ on $S$
for each $R\in\Rel$.
For $T \subseteq S$,
we denote the \emph{substructure} of $\Sf$ induced by the restriction of $\Sf$
to the set $T$ as $\Sf{\restriction}_T$.
In this paper we consider only structures with countable domains.

A partial order $\Af=(A, \leq)$ is regarded as a structure over
a signature consisting of a single binary relation symbol $\leq$.
Every substructure of $\Af$ is again a partial order.
We call a subset $B\subseteq A$ an
\emph{antichain} if for all distinct $a,b\in B$, neither
$a\leq b$ nor $b\leq a$.

Let $A$ be a (not necessarily finite) set. We use $\preceq$ to
denote the prefix order on finite words in $A^*$, i.e., for $u,v\in
A^*$, $u\preceq v$ if $v=uw$ for some $w\in A^*$. For a language
$L\subseteq A^*$, let $\pref(L) = \{w\in A^*\mid \exists u\in L:
w\preceq u\}$ be its {\em prefix-closure}.

\subsection{Trees, forests and upwards linear partial orders}
\label{sec:wulpos}

Let $\Pf=(P, \leq)$ be a partial order.
We say that $\Pf$ is {\em connected} if the undirected graph
$(P, \{ (a,b) \mid a \leq b \text{ or } b \leq a\})$ is connected.
For an element $p \in P$ let 
$\uparrow p = \{ a \in P \mid p < a \}$ be its  (strict) \emph{upwards
  closure} and  
$\downarrow p = \{ a \in P \mid a < p \}$ be its  
(strict) \emph{downwards closure}. Two elements $a,b \in P$ are
{\em comparable} if $a \leq b$ or $b \leq a$.
We say that $\Pf$ is \emph{upwards linear} if
for every $p \in P$, $\uparrow p$ induces a linear order
$\Pf{\restriction}_{\uparrow p}$. 
A {\em forest} is an upwards linear  partial order $\Ff =
(F,\leq)$ such that for every $p\in P$ the set
$\uparrow p$ induces a finite linear order.
A {\em tree} is a forest
which has a greatest element, which is called the {\em root} of the
tree. Note that a forest is a disjoint union of (an arbitrary number
of) trees and all its substructures are also forests.
For a given forest $\Ff$, we denote as 
$\langle \Ff \rangle$ the tree that results from adding a new root
(i.e., a new greatest element) to $\Ff$. If $F$ is the domain of $\Ff$
we write $\langle F\rangle$ for the domain of $\langle \Ff \rangle$.
For a node $u$ in $\Ff$, $\Ff(u)$ denotes the subtree of $\Ff$ at $u$,
i.e., $\Ff(u) = \Ff\restriction_{\{v \in F \mid v \leq u\}}$.
We define the {\em successor relation} of $\Ff$ as 
\begin{align*}
E_{\Ff}=\{(x,y)\in F^2\mid x>y, \neg \exists z: x>z>y\}.  
\end{align*}
For $x\in F$ the set of
{\em children} of $x$ in $\Ff$ is $E_\Ff(x)=\{y\in F\mid (x,y)\in
E_\Ff\}$. The set of leaves of $\Ff$ is
$\leaves(\Ff) = \{ x \in F \mid E_\Ff(x)=\emptyset\}$.
A partial order $\Pf=(P,\leq)$ is {\em well-founded}, if it does
not contain an infinite descending chain
$a_1 > a_2 >a_3 > \cdots$.
The {\em sum} of two partial orders $\Pf_1 =(P_1, \leq_1)$ and
$\Pf_2 =(P_2, \leq_2)$ with $P_1 \cap P_2 = \emptyset$ is the partial
order $\Pf_1 + \Pf_2 = (P_1 \cup P_2, \leq)$, where
$a \leq b$ if either ($a,b \in P_1$  and $a \leq_1 b$), or
($a,b \in P_2$  and $a \leq_2 b$), or ($a \in P_1$ and $b \in P_2$).

We write \emph{\wulpo} as an abbreviation for well-founded upwards
linear partial order. Note that the class of \wulpo's contains all
ordinals and all well-founded forests and is closed under taking substructures.

Let $\Pf=(P, \leq)$ be some partial order. 
In the following, we set $p < \infty$ for all $p \in P$.
In particular, $\infty$ is comparable to every $p \in P \cup \{\infty\}$ and
$\max(p,\infty)=\infty$ for every $p \in P \cup \{\infty\}$.
For $p_1 \in P$ and $p_2  \in P \cup \{\infty\}$ with $p_1 \leq p_2$,
we denote by  $[p_1, p_2)$ the \emph{interval}
$\{p\in P \mid p_1\leq p < p_2\}$. 
\label{def:branching-free}
We call this interval \emph{branching free} if for all $p \in
[p_1,p_2)$ we have  
$$
\downarrow p = \downarrow p_1 + [p_1,p) .
$$
In other words: $[p_1,p_2)$ is branching free, if for every
$p \in [p_1,p_2)$ and every $x < p$, the points
$x$ and $p_1$ are comparable.

\begin{lem} \label{very-simple-lemma}
Let $p_1 \in P$ and $p_2  \in P \cup \{\infty\}$ with
$p_1 \leq p_2$.
If $[p_1,p_2)$ is branching free, then for every $p \in [p_1,p_2)$
also $[p_1,p)$ is branching free.
\end{lem}

\begin{proof} 
Assume that $x < p' \in [p_1,p)$.
We have to show that $x$ and $p_1$ are comparable.
But this clear, since also  $x < p' \in [p_1,p_2)$
and the latter interval is branching free.
\end{proof} 

\begin{lem} \label{lemma-comparable}
Let $\Pf=(P, \leq)$ be an upwards linear partial order
and let $a,c \in P$ and $b,d \in P \cup \{\infty\}$ with $a \leq b$ and $c
\leq d$. If $[a,b)$ and $[c,d)$ are branching free and not disjoint,
then $a$ and $c$ are comparable, and $b$ and $d$ are comparable. 
\end{lem}

\begin{proof}
Assume that $[a,b)$ and $[c,d)$ are not disjoint and let
$p \in [a,b) \cap [c,d)$. 
Since $\Pf$ is upwards linear, the assumptions $p < b$ and $p < d$
imply that $b$ and $d$ are comparable. 
Moreover, since $c \leq p \in [a,b)$ and $[a,b)$ is branching free,
the elements $c$ and $a$ are comparable too.
\end{proof}

We call a branching free interval {\em maximal branching free} if it is
maximal with respect to set inclusion.
We use this interval notation mostly for upwards linear partial
orders. In this case every interval induces a linear suborder. 

\begin{lem} \label{lemma-disjoint-or-equal}
  Let $\Pf=(P, \leq)$ be an upwards linear partial order and let
  $a,c \in P$ and $b,d \in P \cup \{\infty\}$ with $a \leq b$ and $c
  \leq d$.
  If $[a,b)$ and $[c,d)$ are maximal branching
  free, then they are either disjoint or equal. 
\end{lem}

\begin{proof} 
  If $[a,b)$ and $[c,d)$ are not disjoint,
  Lemma~\ref{lemma-comparable} implies that  $\max(b,d)$ and
  $\min(a,c)$ are  defined. 
  We claim that $[\min(a,c),\max(b,d))$ is also branching free. 
  Since we assumed $[a,b)$ and $[c,d)$ to be maximal branching free,
  this implies $[a,b) = [\min(a,c),\max(b,d)) = [c,d)$.

  Let us now prove the claim. 
  If $a=c$ or $b=d$ or ($a<c$, $d<b$) or ($c<a$, $b<d$)
  then $[\min(a,c),\max(b,d))$ is $[a,b)$ or
  $[c,d)$ and hence branching free. The remaining two cases
  are symmetric. Hence we can assume that $a < c$ and $b < d$ and we have
  to show that $[a,d)$ is branching free. 
  Since $a \leq b$ and $a \leq c$, we have either $b \leq c$ or $c
  < b$. If $b \leq c$, then the intervals $[a,b)$ and $[c,d)$
  would be disjoint. Hence we have
  $$
  a < c < b < d .
  $$
  To show that $[a,d)$ is branching free, let $x \in [a,d)$ and $y <
  x$. We have to show that $a$ and $y$ are comparable.
  
  If $x \in [a,b)$ then $a$ and $y$ are comparable since 
  $[a,b)$ is branching free. So, assume that $x < b$ does not hold. 
  Since $\Pf$ is upwards linear and $a \leq b$, $a \leq x$, it follows
  $b \leq x$. In particular $x \in [c,d)$. Since $[c,d)$ is branching
  free, $c$ and $y$ must be comparable. If $c \leq y$ then also $a
  \leq y$. Hence $y < c$. Since $c \in [a,b)$ and $[a,b)$ is branching
  free, $y$ and $a$ are comparable.
\end{proof}

\begin{lem}\label{lem:WulpoPartitionMaxBranchingFreeIntervals}
  Let $\Pf$ be a \wulpo. Every node of $\Pf$ is contained in a unique
  maximal branching free interval. 
\end{lem}

\begin{proof}
  Fix a node $p\in P$. Firstly, we define a node $p_{+}>p$, secondly
  we define a node $p_{-}\leq p$, and finally, we prove that
  $[p_{-}, p_{+})$ is maximal branching free. 
  If for all $p' > p$ and all $p''<p'$,  $p''$ and $p$ are
  comparable, set $p_{+}:=\infty$. Note that $[p,p_{+})$ is
  branching free.
  Otherwise, by well-foundedness and upwards linearity, there is a
  minimal element $p_{+}>p$ for which there exists $p'<p_{+}$
  incomparable to $p$. Thus, $[p,p_{+})$ is branching free. 
  
  Lemma~\ref{lemma-comparable} implies that, for all $p_1\leq p, p_2\leq p$, if
  $[p_1,p_{+})$ and $[p_2,p_{+})$ are branching free, then $p_1$
  and $p_2$ are comparable. 
  By well-foundedness, there is a minimal $p_{-}\leq p$ such that
  $I:=[p_{-}, p_{+})$ is branching free. 

  In order to prove maximality, assume that for some $p_1\leq p <
  p_2$, $[p_1,p_2)$ is branching free. 
  First we show that $p_{+} \not\in [p_1,p_2)$: If $p_{+} \in
  [p_1,p_2)$ then, since $[p_1,p_2)$ is branching free, for every
  $p' < p_+$, $p'$ and $p_1$ are comparable. By upwards linearity,
  if $p'$ and $p_1$ are comparable then also $p'$ and $p$ are
  comparable. Hence, every $p' < p_+$ is comparable to $p$. But this
  contradicts the choice of $p_+$. Hence, we have $p_{+} \not\in [p_1,p_2)$.
  
  Since $p < p_2$ and $p < p_+$, $p_2$ and $p_+$ are comparable.
  We cannot have $p_+ < p_2$. Hence, $\max(p_2,p_{+}) = p_+$.
  By Lemma~\ref{lemma-disjoint-or-equal}, the interval
  \begin{equation*}
    [\min(p_1,p_{-}), p_{+}) = [\min(p_1,p_{-}), \max(p_2,p_{+}))    
  \end{equation*}
  is branching free. 
  But by definition of $p_{-}$, $I$ is the minimal branching free
  interval of the form $[p,p_{+})$. Thus, $p_{-}\leq p_1$ and
  $[p_1,p_2)$ is contained in $I$. This shows that $I$ is maximal
  branching free. 
\end{proof}

Lemma~\ref{lemma-disjoint-or-equal} and
\ref{lem:WulpoPartitionMaxBranchingFreeIntervals} 
imply that every \wulpo is partitioned into its maximal branching free intervals.

\begin{lem}\label{lem:SeparationMBFreeInt}
  Let $\Pf$ be a \wulpo and  $I\neq J$ two different maximal branching
  free intervals of $\Pf$. For $p_1\in I$ and $p_2\in J$, if $p_1<p_2$
  then there is a $p\in \Pf$ such that $p < p_2$ but $p$ and $p_1$ are
  incomparable, i.e., $\downarrow p_2 \neq \downarrow p_1 \cup
  [p_1,p_2)$. 
\end{lem}
\begin{proof}
  Heading for a contradiction, assume that 
  \begin{equation}
    \label{eq:SeparationMBFreeInt}
  \downarrow p_2 = \downarrow p_1 \cup [p_1,p_2).    
  \end{equation}
  Note that there is a $p_m\in \Pf\cup\{\infty\}$ such that
  $K := [p_1,p_m)=\{p\in\Pf\mid p_1 \leq p \leq p_2\}$.
  Due to \eqref{eq:SeparationMBFreeInt}, $K$ is branching free.
  Using Lemma \ref{lemma-disjoint-or-equal} we conclude that
  $I\cup K \cup J$ is branching free contradicting maximality of $I$
  and $J$.
\end{proof}
We finally discuss the replacement of maximal branching free intervals in \wulpo's. 
Let $\Pf=(P, \leq)$ be a \wulpo and $I=[p_0,p_1)$ a maximal branching free
interval. For each $i\in I$ and each $p\in P\setminus I$ we have
\begin{align*}
  &\text{ $p \leq i$ if and only if $p\leq i'$ for all $i'\in I$ and}\\
  &\text{ $p \geq i$ if and only if $p\geq i'$ for all $i'\in I$}
\end{align*}
Let $J$ be some ordinal. Then we can replace $I$ by $J$ and obtain the
\wulpo
$\Pf'=((P\setminus I) \sqcup J, \leq')$ where
$p \leq' p'$ if
\begin{enumerate}[(1)]
\item $p,p'\in P\setminus I$ and $p\leq p'$, 
\item $p,p'\in J$ and $p\leq p'$ in $J$, 
\item $p\in P\setminus I, p'\in J$ and $p\leq p_0$, or
\item $p\in J$, $p'\in P\setminus I$ and $p_0\leq p'$.
\end{enumerate}
Note that the maximal branching free intervals of $\Pf'$ are $J$ and
the maximal branching free intervals of 
$\Pf$ except for $I$. 
Analogously, one can also replace simultaneously all maximal branching free intervals
of a \wulpo by new ordinals. This kind of replacement will be used
in Lemma \ref{lem:wulpo-box-respects-branching}.

\subsection{Ordinal arithmetic and the ordinal rank}
\label{sec:Rank}

We use standard terminology concerning ordinals; see e.g. \cite{Ros82}.
We briefly recall the operations of natural sum and natural product
on ordinals. 
The natural sum of ordinals $\alpha$ and
$\beta$, denoted by
$\alpha \oplus \beta$, is the largest (well-founded) linearization of 
$\alpha\sqcup \beta$ (cf.~\cite{Carruth42}, where the following
statements can be found).
It is commutative, associative and strictly monotone, i.e., $\alpha \oplus \beta <
\alpha \oplus \gamma$ if and only if $\beta < \gamma$. Furthermore,
\begin{align} \label{Eqn-Oplus-Limit}
  \alpha \oplus \beta = \omega^\gamma\text{ implies }\alpha =
  \omega^\gamma\text{ or }\beta = \omega^\gamma.
\end{align}
The natural product (cf.~\cite{Carruth42}) of $\alpha$ and $\beta$,
denoted as $\alpha\otimes \beta$, is
the largest  (well-founded) linearization of the direct (Cartesian)
product of $\alpha$ 
and $\beta$ with the componentwise order.
It is commutative, associative, and strictly monotone, i.e., for
$\alpha\neq 0$ we have $\alpha
\otimes \beta < 
\alpha \otimes \gamma$ if and only if $\beta < \gamma$. Furthermore,
\begin{align} \label{Eqn-Otimes-Limit}
  \alpha \otimes \beta = \omega^{\omega^\gamma}\text{ implies }\alpha =
  \omega^{\omega^\gamma}\text{ or }\beta = \omega^{\omega^\gamma}.
\end{align}

\begin{lem}\label{lem:subdistributivity-plus-natprod}
  Let $\alpha, \beta, \gamma$ be ordinals. Then
  \begin{align*}
    (\alpha \otimes \beta) + (\alpha \otimes \gamma) \leq
    \alpha \otimes (\beta + \gamma)
  \end{align*}
  where $+$ is the usual sum of ordinals. 
\end{lem}
\begin{proof}
  Observe that 
  $(\alpha \otimes \beta) + (\alpha \otimes \gamma)$ is a well-order
  with domain $(\alpha\times \beta) \sqcup (\alpha\times \gamma) = 
  \alpha \times (\beta \sqcup \gamma)$ such
  that for all $a<a'\in \alpha$, $b<b'\in \beta$ and $c<c'\in \gamma$
  we have 
  $(a,b) < (a',c)$, 
  $(a,b) < (a',b)$, 
  $(a,c) < (a',c)$, 
  $(a,b) < (a,b')$, and
  $(a,c) < (a,c')$. 
  Thus, we conclude that it is a linearization of 
  the product of the orders $\alpha$ and $\beta+\gamma$. 
  Since $\alpha \otimes (\beta + \gamma)$ is the largest linearization
  of these two orders, the claim follows immediately.
\end{proof}
Let us recall the  (ordinal) rank of a
well-founded partial order.
For a set of ordinals $M$, we denote by $\sup(M)$ the
supremum of $M$, where $\sup(\emptyset) = 0$.

\begin{defi}
  Let $\Pf$ be a well-founded partial order.
  The \emph{rank}  of an element
  $p\in \Pf$ is
  inductively defined by
  \begin{align*}
    \rank(p, \Pf)=\sup\{\rank(p',\Pf)+1\mid p'<p\in P\}.
  \end{align*}
  The rank of $\Pf$ (also called the \emph{ordinal height} of $\Pf$) is
  \begin{align*}
    \rank(\Pf)=\sup\{\rank(p,\Pf)+1 \mid p\in P\}.
  \end{align*}
\end{defi}
If the partial order $\Pf$ is clear from the context, we will write
$\rank(p)$ 
instead of $\rank(p, \Pf)$. Note that $\rank(p, \Pf) =
\rank(\Pf{\restriction}_{\downarrow p})$. Also note that the rank
of the ordinal $\alpha$ is exactly $\alpha$.
Next, we collect some useful facts about the rank. 

\begin{lem} \label{lemma-rank-substructure}
  Let $\Pf=(P,\leq)$ be  a well-founded partial order and
  $S\subseteq P$ a subset. The induced suborder
  $\Pf{\restriction}_S$ has rank at most
  $\rank(\Pf)$. 
\end{lem}
\begin{proof}
  By induction on the rank of an element $s\in S$ (with respect to
  $\Pf$) one easily shows
  that the rank of $s$ in $\Pf{\restriction}_S$ is bounded by the rank
  of $s$ in $\Pf$. 
\end{proof}

We will also use a  result proved by Khoussainov and Minnes on
decomposition of partial orders into suborders. 

\begin{lem}[\cite{KhMi09}, Lemma 3.3]
  \label{lem:rankNaturalSum}
  Let $\Pf=(P, \leq)$ be a well-founded partial order, and let $P =
  P_1 \sqcup P_2$ be 
  a partition of the domain of $\Pf$. Setting
  $\Pf_i = \Pf\restriction_{P_i}$ for $i \in \{1,2\}$,
  we have $\rank(\Pf)\leq \rank(\Pf_1) \oplus \rank(\Pf_2)$.
\end{lem}

\begin{lem} \label{lem:IntervalRankShift}
  Let $\Pf=(P, \leq)$ be a well-founded partial order and $p_1\in P$, 
  $p_2 \in P \cup\{\infty\}$ with $p_1 \leq p_2$. If
  $[p_1,p_2)$ is branching free, then 
  $\rank(p)=\rank(p_1) + \rank([p_1,p))$
  for each $p\in [p_1,p_2)$. 
  Moreover, if $p_2 \in P$, then
  $\rank(p_2) \geq \rank(p_1) + \rank([p_1,p_2))$. 
\end{lem}

\begin{proof}
  Let $p\in [p_1,p_2)$. We prove the first statement of the lemma
  by induction on $\alpha := \rank( [p_1,p) )$. We distinguish three
  cases:

  \medskip
  \noindent 
  {\em Case 1.} $\alpha=0$. We must have  $[p_1,p)=\emptyset$, i.e.,
  $p_1=p_2$. Hence $\rank(p_2) = \rank(p_1) + 0$. 
  
  \medskip
  \noindent 
  {\em Case 2.} $\alpha$ is a successor ordinal, i.e., $\alpha=
  \beta+1$. Then by   definition of the rank there is
    an element $p' \in [p_1,p)$ such that
    \begin{equation*}
    	\rank(p', [p_1,p)) = \rank([p_1,p'))=\beta.
    \end{equation*}
    Moreover, $\rank(p) \geq \rank(p') + 1$ since $p' < p$.
    By induction
    hypothesis, we  conclude that 
    \begin{equation*}
    	\rank(p) \geq \rank(p') + 1 = \rank(p_1) + \beta + 1 =
    	\rank(p_1)+\alpha.
    \end{equation*}
    Moreover, for every $p' \in [p_1,p)$ we must have $\rank([p_1,p')) \leq \beta$.
    Since $[p_1,p_2)$ is branching free, we have $\downarrow p' =
    \downarrow p_1 + [p_1,p')$ for every $p' < p$. 
    Therefore 
    $$
    \rank(p) = \sup \{ \rank(p')+1 \mid p' \in [p_1,p) \} .
    $$
    By induction, 
    $$
    \rank(p') = \rank(p_1)+\rank([p_1,p')) \leq \rank(p_1)+\beta .
    $$
    Hence, $\rank(p) \leq \rank(p_1)+\beta+1 = \rank(p_1)+\alpha$.

    \medskip
    \noindent 
    {\em Case 3.} $\alpha$ is a limit ordinal. 
    Since $[p_1,p_2)$ is branching free, for each $p'<p$ one of the
    following holds. 
    \begin{enumerate}[(1)]
    \item $p'\in [p_1,p)$ and there is some
      $\beta<\alpha$ such that $\rank([p_1,p'))=\beta$.
      By induction hypothesis, $\rank(p') = \rank(p_1)+\beta$.
    \item     $p'<p_1$ whence $\rank(p') \leq \rank(p_1)$. 
    \end{enumerate}
    Thus, we have
    \begin{align*}
    	\rank(p) &\leq   \sup\{\rank(p_1)+\beta+1 \mid \beta<\alpha\} \\
		&= \sup\{\rank(p_1)+\beta \mid \beta < \alpha\} \\
                &= \rank(p_1) + \alpha.
    \end{align*}
    On the other hand, $\rank([p_1,p))=\alpha$ implies that for each
    $\beta<\alpha$ there is a node ${p_\beta\in [p_1,p)}$
    such that $\beta \leq \rank([p_1,p_\beta)< \alpha$. 
    By induction, we get
    $\rank(p_\beta) \geq \rank(p_1)+\beta$. 
    Thus,
    \begin{align*}
    	\rank(p) &\geq  \sup\{\rank(p_\beta) +1 \mid \beta<\alpha\} \\
		&\geq  \sup\{ \rank(p_1)+\beta +1  \mid \beta<\alpha\} \\
		&= \sup\{ \rank(p_1)+\beta   \mid \beta<\alpha\} \\
                &= \rank(p_1) + \alpha.
    \end{align*}
    Finally, the second statement of the lemma follows easily:
    \begin{align*}
    	\rank(p_2) & \geq  \sup\{ \rank(p)+1 \mid p \in [p_1,p_2) \} \\
        & =  \sup\{ \rank(p_1)+ \rank([p_1,p))+1 \mid p \in [p_1,p_2) \}\\
        & =  \rank(p_1) + \rank([p_1,p_2))
    \end{align*}
    This concludes the proof of the lemma.
 \end{proof}

\begin{cor}\label{cor:IntervalRankShift}
  Let $\Pf$ be a partial order and $p_1 \leq p_2\in \Pf$ such that
  $[p_1,p_2)$ is branching free. For every $p\in [p_1,p_2)$,  
  $\rank(p) < \rank(p_1) + \rank([p_1,p_2))$.   
\end{cor}

\subsection{Finitely labeled trees} \label{sec:tree}

We identify a non-empty, finite, prefix-closed subset
$T\subseteq\{0,1\}^*$ with the tree
${(T,\succeq)}$ and call $T$ a 
\emph{finite binary tree}. Note that $\varepsilon$ is the 
largest element of the inverse prefix relation $\succeq$ and hence the root of the tree.
We denote the set of all finite binary trees
as $\T^\fin_2$.
For $T \in \T^{\fin}_2$ let
\begin{align*}
	\fr(T) = T\{0,1\}\setminus T \quad \text{and} \quad
	\cl(T) = T \cup \fr(T)
\end{align*}
be the \emph{frontier} and \emph{closure} of $T$.
Notice that $\cl(T) \in \T^{\fin}_2$.

A finite {\em $\Sigma$-labeled} binary tree is a
pair $(T,\lambda)$, where $T \in \T^\fin_2$
and $\lambda: T\rightarrow \Sigma$ is a labeling
function. We denote the set of all
finite $\Sigma$-labeled binary trees by  $\T^\fin_{2,\Sigma}$; elements of
$\T^\fin_{2,\Sigma}$ will be denoted by lower case
letters ($s,t,\ldots$).
When $\Sigma$ is the singleton set $\{\sharp\}$,
we will simply consider a tree $t \in \T^\fin_{2,\Sigma}$ as
unlabeled, i.e., $t \in \T^\fin_2$.
The set of leaves of $t = (T,\lambda)$ is $\leaves(t)=\leaves(T,\succeq)$.
We use the following operations on labeled trees
$t=(T,\lambda) \in \T^\fin_{2,\Sigma}$.
For $d \in T$, the {\em subtree rooted at} $d$ is the tree
$t(d) =(U,\mu) \in \T^\fin_{2,\Sigma}$ with
$U = \{ u \in \{0,1\}^*\mid du \in T\}$ and
$\mu(u) = \lambda(du)$.
For trees $t_1,\ldots,t_n \in \T^\fin_{2,\Sigma}$
and nodes $d_1,\ldots,d_n \in \cl(T)$ forming an antichain
(i.e., $d_i$ is not a prefix of $d_j$ for $i\neq j$)
we consider the tree
$t[d_1/t_1,\dotsc,d_n/t_n]\in \T^\fin_{2,\Sigma}$
which is obtained from $t$ by simultaneously replacing for each $i$ the subtree rooted at $d_i$ by $t_i$. Formally,
$t[d_1/t_1,\ldots,d_n/t_n]=(V,\nu)$ is defined by
\begin{align*}
	V &= \left(T \setminus \left(\{ d_1,\ldots,d_n \}\{0,1\}^* \right)\right) \cup
		\bigcup_{1\leq i\leq n}Â d_iT_i   \text{ and}\\
	\nu(v) &= \begin{cases}
		\lambda_i(u) & \text{if $v=d_iu$ for some (unique)
			$i \in \{1,\ldots,n\}$,} \\
		\lambda(v) & \text{otherwise,}
	\end{cases}
\end{align*}
where $t_i=(T_i,\lambda_i)$.

We fix a new symbol $\diamond \not\in \Sigma$. Let
$\Sigma_\diamond$ denote the set $\Sigma \sqcup \{\diamond\}$ and
$\Sigma_\diamond^n$ the $n$-fold product 
$\Sigma_\diamond\times \dots \times \Sigma_\diamond$.
Let ${\bar t=(t_1,\ldots,t_n) \in (\T^{\fin}_{2,\Sigma})^n}$  be a
tuple of trees with $t_i=(T_i,\lambda_i)$. 
The \emph{convolution} of $\bar t$ is the $\Sigma_\diamond^n$-labeled
binary tree  
$\otimes\bar t = (T,\lambda) \in
\T^{\fin}_{2,\Sigma_\diamond^n}$
defined by
\begin{align*}
	T &= \bigcup_{1\leq i\leq n} T_i \text{ and}\\
	\lambda(u) &= (\lambda_1'(u),\ldots,\lambda_n'(u))\,, \text{
          where}\\
	\lambda_i'(u) &= \begin{cases}
		\lambda_i(u) & \text{if $u\in T_i$,} \\
		\diamond & \text{if $u\in T\setminus T_i$.}
	\end{cases}
\end{align*}
Instead of $\otimes (t_1, t_2, \dots, t_n)$ we sometimes also write 
$t_1\otimes\cdots\otimes t_n$.
Finally, for any $n$-ary relation
$R\subseteq(\T^{\fin}_{2,\Sigma})^n$
we define
\begin{equation*}
	\otimes R
	= \{ \otimes \bar t \mid \bar t \in R \}
	\subseteq \T^{\fin}_{2,\Sigma_\diamond^n}\,.
\end{equation*}

\subsection{Tree automata and tree-automatic structures}

Let $\Sigma$ be a finite alphabet.
A (top-down) \emph{tree automaton}
over $\Sigma$ is a tuple $\A = (Q,\Delta, I,F)$, where
$Q$ is the finite set of states, $I\subseteq Q$ is the set of initial
states, $F \subseteq Q$ is the set of final states, and
\mbox{$\Delta\subseteq (Q\setminus F)\times\Sigma \times Q \times Q$}
is the transition relation.\footnote{In contrast to the usual
  definition, we disallow transitions starting in final
  states. Obviously this restricted version of tree-automata is
  equivalent to the usual one. Our definition simplifies some
  constructions because runs on a tree $t$ only assume final states on
  the frontier of $t$.}  Given a finite $\Sigma$-labeled binary tree
$t=(T,\lambda) \in \T^\fin_{2,\Sigma}$, a
{\em successful run} of $\A$ on $t$ is a mapping
$\rho:\cl(T) \rightarrow Q$ such that
\begin{enumerate}[(1)]
\item  $\rho(\varepsilon) \in I$,
\item  $\rho(\fr(T)) \subseteq F$, and
\item $(\rho(d),\lambda(d),\rho(d0), \rho(d1))\in\Delta$
for every $d\in T$.
\end{enumerate}
Let  $L(\A)$  be the set of all $t \in \T^\fin_{2,\Sigma}$
which admit a successful run of~$\A$. A set $L\subseteq \T^\fin_{2,\Sigma}$
is called {\em regular} if there is a tree automaton $\A$ over
$\Sigma$ with $L=L(\A)$.

A successful run $\rho$ of $\A$ on  $t=(T,\lambda)\in
\T^\fin_{2,\Sigma}$ naturally defines a
$\left((\Sigma\times (Q\setminus F)) \cup F\right)$-labeled binary
tree $\tree(\rho) = (\cl(T),\mu)$ such that
$\mu(d) = (\lambda(d),\rho(d))$ for every $d\in T$ and
${\mu(d) =  \rho(d)}$ for every $d \in \fr(T)$.
$\Run(\A,t)$ denotes the set of
all trees $\tree(\rho)$ where $\rho$ is a successful run of $\A$ on
$t$.  Let 
$\Run(\A) = \bigcup_{t\in L(\A)}\Run(\A,t)$. This is also a regular set:
a tree automaton for $\Run(\A)$ can be obtained by replacing every
transition $(p,a,p_1,p_2)\in \Delta$ by $(p,(a,p),p_1,p_2)$ and adding
transitions $(f,f,\bot,\bot)$ for each final state $f\in F$ where
$\bot$ is a new state which is the only final state of the new
automaton. For notational simplicity we refer to $\tree(\rho)$
simply as $\rho$.

A tree automaton $\A=(Q,\Delta,I,F)$ is called 
\emph{bottom-up deterministic} if 
\begin{enumerate}[(1)]
\item  for all $p,q\in Q$ and $a\in\Sigma$ there
  is a unique $r\in Q$ with $(r,a,p,q)\in\Delta$ and
\item $F$ is a
  singleton set.
\end{enumerate}
In this situation, for $F=\{q_0\}$, a tree $t=(T,\lambda) \in
\T^{\fin}_{2,\Sigma}$, a node $u\in\cl(T)$, a set 
$U\subseteq \fr(T)$
and a map $\zeta:U\to Q$ we recursively define
a state $\A(t,u,\zeta) \in Q$  by
\begin{equation*}
	\A(t,u,\zeta) = \begin{cases}
		r & \text{if $u\in T$ and
			$(r,\lambda(u),\A(t,u0,\zeta),\A(t,u1,\zeta))\in\Delta$,} \\
		\zeta(u) & \text{if $u\in U$,} \\
		q_0 & \text{if $u\in\fr(T)\setminus U$.}
	\end{cases}
\end{equation*}
We omit the second parameter if $u=\varepsilon$ and
the third one if $U=\emptyset$.
Notice that $\A(t,u)=\A(t(u))$ for all $u\in T$.
The tree $t$ admits a successful run $\rho$ of~$\A$
if and only if $\A(t)\in I$.
It is well known that for each tree automaton $\A$
one can compute a bottom-up deterministic tree automaton $\A'$
with $L(\A)=L(\A')$.

An $n$-ary relation $R\subseteq (\T^\fin_{2,\Sigma})^n$ is called
{\em tree-automatic} if there is a tree automaton $\A_R$ over
$\Sigma_\diamond^n$
such that $L(\A_R) = \otimes R$.
A relational structure $\Sf$ is called {\em tree-automatic} over $\Sigma$
if its domain $S$ is a regular subset of $\T^\fin_{2,\Sigma}$ and each of
its atomic relations $R^\Sf$ is tree-automatic; any tuple $\P$ of automata
that accept the domain and the relations of $\Sf$ is called a
{\em tree-automatic presentation} of $\Sf$. In this case, we write
$\Sf(\P)$ for $\Sf$. If a tree-automatic structure $\Sf$ is isomorphic to
a structure $\Sf'$, then $\Sf$ is called a {\em tree-automatic copy} of
$\Sf'$ and $\Sf'$ is {\em tree-automatically presentable}. In this paper
we sometimes abuse the terminology referring to $\Sf'$ as simply
tree-automatic and calling a tree-automatic presentation of $\Sf$ also
a tree-automatic presentations of $\Sf'$. We also simplify our statements by
saying ``given/compute a tree-automatic structure $\Sf$'' for
``given/compute a tree-automatic presentation $\P$ of a structure
$\Sf(\P)$''.
The structures $(\N, +)$ and $(\N, \cdot)$ are examples
of tree-automatic structures.  

A tree-automatic structure over a singleton alphabet
(i.e., the domain of the structure is a subset of $\T^\fin_2$) is called
{\em unary tree-automatic}. Moreover, let
$$
\T_{\bin} = \{ t \in \T^\fin_2 \mid \forall u \in t : u0 \in t
\Leftrightarrow u1 \in t \}
$$
be the set of all finite (unlabeled) full binary trees.
We will make use of the following simple lemma.

\begin{lem}\label{lemma-unary-alphabet}
For every tree-automatic structure $\Sf$  there is an isomorphic unary
tree-automatic structure $\Sf'$
whose domain is a subset of $\T_{\bin}$. Moreover, there is a
computable isomorphism from $\Sf$ to $\Sf'$.
\end{lem}

\begin{proof}
Let $\Sigma$ be some finite alphabet; w.l.o.g. assume that
$\Sigma = \{1,2,\ldots,n\}$. For $1 \leq i \leq n$
let $a_i = \pref( \{ 0, 10, 110, \ldots, 1^{i-1}0, 1^i\}) \in \T_{\bin}$.
We inductively define an injective mapping
${\unlabel : \T^{\fin}_{2,\Sigma} \to \T_{\bin}}$ as
follows:
Let $t = (T,\lambda) \in \T^{\fin}_{2,\Sigma}$ and for $i \in \{0,1\}$
let $t_i = t(i)$ be the subtree of $t$ rooted at node $i$,
where we set $t_i = \emptyset$ if $i \notin T$. Then
\begin{equation*}
\unlabel(t) = \{\varepsilon,0\} \cup 00\unlabel(t_0)
  \cup 01\unlabel(t_1) \cup 1 a_{\lambda(\varepsilon)},
\end{equation*}
where we set $\unlabel(\emptyset) = \emptyset$.
By induction over the size of $t$ it easily follows that the mapping
$\unlabel$ is indeed injective. We show that for
every tree-automatic relation $R\subseteq (\T^{\fin}_{2,\Sigma})^k$,
the relation
\[
\unlabel(R) = \{(\unlabel(t_1),\ldots, \unlabel(t_k))\mid
(t_1,\ldots,t_k)\in R\}
\]
is also tree-automatic. Suppose $\A=(Q, \Delta,I,F)$
is a tree automaton recognizing the relation $R$. We construct a
(top-down) tree automaton
$\A'$ as follows:
The state set of $\A'$ contains the set
$$
Q \;\cup\; (Q \times Q)
\;\cup\; \{\diamond, 1,2,\ldots, n\}^n \cup \{\bot\}
$$
(in addition $\A'$ contains some auxiliary states that we do not
specify),
where $\bot$ is final (states from $F$ are no longer final).
For a state $q \in Q$, $\A'$ contains the following
transitions (we omit here the unique node label, which formally should
be the second component of every transition):
$$
(q,  (p,r), (x_1, \ldots, x_k)) \text{ if } (q, (x_1, \ldots, x_k), p,
r) \in \Delta \quad \text{ and } (q,\bot,\bot) \text{ for } q \in F .
$$
For a state $(p,r) \in Q \times Q$, $\Delta'$ contains the following
transitions:
$$
((p,r), p, r).
$$
Finally, the tree automaton $\A'$ contains additional states and transitions
such that from a state $(x_1, \ldots, x_k) \in \{\diamond, 1,2,\ldots,n\}^n$
only the tree $\pref(s_1 \otimes s_2 \otimes \cdots \otimes s_n)$ with
\begin{equation*}
    s_i=\begin{cases}
      a_{x_i} \ \ &\text{if $x_i\in \{1,\ldots,n\}$,}\\
      \{\varepsilon\} \ \ &\text{if $x_i=\diamond$}.
    \end{cases}
\end{equation*}
is accepted.  Now let $t_1,\ldots,t_k\in \T^\fin_{2,\Sigma}$ and $q \in Q$.
A straightforward induction on the size of trees shows that
$\A$ accepts the convolution $t_1 \otimes \cdots \otimes t_k$ via a run $\rho$
with $\rho(\varepsilon) = q$ if and only if
$\A'$ accepts ${\unlabel(t_1) \otimes \cdots \otimes \unlabel(t_k)}$ via
a run $\rho'$
with $\rho'(\varepsilon) = q$.

The above considerations shows that for every
tree-automatic structure $\Sf$  there exists an isomorphic unary
tree-automatic structure $\Sf'$,
whose domain is a subset of $\T_{\bin}$. An isomorphism between
$\Sf$ and $\Sf'$ is given by the computable mapping $\unlabel$.
This proves the lemma.
\end{proof}
Consider $\FOext[\tau]$,
the first-order
logic over the signature $\tau$ extended by the quantifiers
$\exists^\infty$ (``there are infinitely many'')
and the chain-quantifier $\exists^{\mathsf{chain}}$
(if $\varphi(x,y)$ is some formula, then
$\exists^{\mathsf{chain}} \varphi(x,y)$ asserts that $\varphi$ is a
partial order and there is an infinite increasing $\varphi$-chain).
Let $\varphi(x_1,\dotsc,x_m,y_1,\dotsc,y_n)$
be a formula of this logic in the signature of some structure $\Sf$.
For all $\bar b\in S^n$ we define the $m$-ary relation
\begin{equation*}
	\varphi^{\Sf}(\cdot,\bar b)
	= \{ \bar a\in S^m \mid \Sf \models \varphi(\bar a,\bar b) \}
\end{equation*}
on $\Sf$. In case of $n=0$ we simply write $\varphi^{\Sf}$ instead of
$\varphi^{\Sf}(\cdot)$.
The following theorem from \cite{Blu99,KartzowPHD2011,ToL08} lays out
the main motivation for investigating tree-automatic structures.

\begin{thm}\label{thm:tree-automatic}
  From a tree-automatic presentation $\P$ of a structure $\Sf(\P)$
  over the signature $\tau$ and
  an $\FOext[\tau]$-formula
  $\varphi(\overline{x})$ one can compute a tree automaton
  accepting $\otimes\varphi^{\Sf(\P)}$.
  In particular, the
  $\FOext[\tau]$-theory
  of any tree-automatic
  structure $\Sf$ is (uniformly) decidable.
\end{thm}
Note that the property of being a tree is expressible in
$\FO[\leq]+\exists^\infty$. Moreover, the chain-quantifier allows to
define well-foundedness of a tree. Hence, we obtain:

\begin{cor}
  The class of tree-automatic well-founded trees is decidable.
\end{cor}
Let $\K$ be a class of tree-automatic presentations. The
{\em isomorphism problem} $\Iso(\K)$ is the set of pairs
$(\P_1,\P_2)\in \K\times \K$ of tree-automatic presentations with
$\Sf(\P_1)\cong \Sf(\P_2)$. If $\K$ is the class of
tree-automatic presentations
for a class $\C$ of relational structures (e.g. trees), then we 
briefly speak of the isomorphism problem for (tree-automatic members
of) $\C$. The isomorphism problem for the class of all tree-automatic
structures is complete for $\Sigma^1_1$, the first level of the
analytical hierarchy; this holds already for string-automatic trees
\cite{KhoNRS07,KuLiLo11}.

\section{Bounding the rank of tree-automatic \wulpo's and well-founded trees}
\label{sec:Delhomme}
In this section we show that the $\rank$ of a tree-automatic \wulpo is
strictly below $\omega^{\omega^\omega}$. As a corollary, we 
obtain that the rank of every tree-automatic
well-founded tree is strictly below $\omega^\omega$. The first part of
our proof is a refinement of Delhomm{\'e}'s decomposition theorem for
tree-automatic structures \cite{Delhomme04}. This result
shows that for a given tree-automatic structure $\Af$ and
a first-order formula $\varphi(x,\bar y)$ all substructures induced
by $\varphi$ (via different tuples $\bar s$ of parameters)
are composed from a finite set $\mathcal{S}_\varphi^{\Af}$ of structures
using the operations of box-augmentation and sum-augmentation.
Roughly speaking, a structure $\Af$ is a sum-augmentation of structures
$\Bf_1,\ldots,\Bf_n$ if there is a finite partition of the domain of
$\Af$ which induces the substructures $\Bf_1,\ldots,\Bf_n$ (see Definition \ref{def:sum-augmentation}).
The structure $\Af$ is a  box-augmentation of $\Bf_1,\ldots,\Bf_n$
if its domain is a finite product of copies of the domains of
$\Bf_1,\ldots,\Bf_n$ such that fixing all but the $i^{\mathrm{th}}$
component of this product arbitrarily results in a structure
isomorphic to $\Bf_i$ (see Definition \ref{def:box-augmentation}). 

Let $\nu$ be a function assigning ordinals to structures.
An ordinal $\alpha$ is $\nu$-in\-de\-com\-pos\-able if for all structures
$\Bf_1,\ldots,\Bf_n$ and all
sum- or box-augmentations $\Af$ of them,
${\nu(\Af)=\alpha}$ implies $\nu(\Bf_i)=\alpha$
for some $i\in \{1,\ldots,n\}$
(see Definition \ref{def:indecomposable value}).
Delhomm{\'e}'s result implies the following. The substructures of a
tree-automatic structure $\Af$ (induced by a fixed first-order formula
$\varphi$) only take finitely many $\nu$-indecomposable values.
Unfortunately, Delhomm{\'{e}} never published a proof of this result.
Since our further arguments rely on this proof,
we reprove his result in Sections~\ref{sec:DelhommesDecomposition}
and~\ref{sec:Indecomposability}. 
More precisely, we strengthen his result in the following way.
\begin{iteMize}{$\bullet$}
\item All box-augmentations occurring 
  in the decomposition of tree-automatic substructures provided by
  Delhomm{\'e}'s result are tamely colorable.
  Roughly speaking, $\Af$ is a tamely colorable
  box-augmentation of $\Bf_1,\ldots,\Bf_n$ if
  there are finite colorings of the $\Bf_i$ and a simple rule how to
  completely reconstruct $\Af$ from the $\Bf_i$ and these colorings.
\item We introduce a notion of $\nu$-in\-de\-com\-posa\-bil\-ity
  for intervals of ordinals:
  we say $[\alpha_1,\alpha_2]$ is $\nu$-indecomposable if
  \begin{enumerate}[(1)]
  \item for all structures
    $\Bf_1,\ldots,\Bf_n$
    and all sum-augmentations $\Af$ of them,
    ${\nu(\Af)=\alpha_2}$ implies $\nu(\Bf_i)=\alpha_2$
    for some $i\in \{1,\ldots,n\}$ and
  \item for all structures
    $\Bf_1,\ldots,\Bf_n$
    and all tamely colorable  box-augmentations $\Af$ of them,
    $\nu(\Af)=\alpha_2$ implies $\nu(\Bf_i)\in [\alpha_1,\alpha_2]$
    for some $i\in \{1,\ldots,n\}$.
  \end{enumerate}
\end{iteMize}
It follows directly that for a sequence
$\alpha_0 < \alpha_1 < \alpha_2 < \cdots$ of ordinals such that
each interval $[\alpha_i,\alpha_{i+1}]$ is $\nu$\nobreakdash-indecomposable, 
the  substructures of a tree-automatic structure
$\Af$ (induced by a first-order formula $\varphi$)
take only finitely many values among the $\alpha_i$ under $\nu$ 
(Proposition~\ref{Prop:Indecomposability}). 

In Section \ref{sec:rank-box-ind}, we prove that the $\rank$ of any
tree-automatic \wulpo
is bounded by $\omega^{\omega^\omega}$ and in the case of a
tree-automatic well-founded forest its $\rank$ is bounded by
$\omega^\omega$. 
We prove this result using
$\rank$-indecomposability:
we can prove that all intervals of the form
$[\omega^{\omega^\alpha},\omega^{\omega^{\alpha+1}}]$ are
$\rank$-indecomposable (if we restrict the domain of $\rank$ to the
class of \wulpo's). Thus, there are 
infinitely many $\rank$-indecomposable intervals below
$\omega^{\omega^\omega}$. 
A simple transfinite induction shows that for every well-founded partial
order $\Pf$ of rank $\alpha$ and for all $\beta<\alpha$, there is a
node $p\in\Pf$ such that $\Pf{\restriction}_{\{p'\mid p'<p\}}$ is of rank
  $\beta$. 
Thus, for a fixed tree-automatic well-founded partial order $\Pf$, the
automaton corresponding to the formula $x<y$  
induces substructures of all ranks $\beta<\rank(\Pf)$. 
For a tree-automatic \wulpo $\Pf$ we have seen that these
substructures may only take finitely many values of the form
$\omega^{\omega^i}$. Thus, one concludes immediately that 
$\rank(\Pf)<\omega^{\omega^\omega}$. 
The bound on the ranks of tree-automatic forests follows from an
analogous proof: writing $\rank_{\F}$ for the restriction of $\rank$
to forests, we can prove that the interval $[\omega^i, \omega^{i+1}]$
is $\rank_\F$-indecomposable. 

Note that Delhomm{\'e}'s original result is too weak for our proof.
Using the definition of  $\rank$-indecomposable values with respect
to all box-augmentations, the ordinals $0$ and $1$ are the only
$\rank$-indecomposable
values.  The problem is that any forest containing an infinite
antichain is
the box-augmentation of two infinite antichains. Note that an infinite
antichain has $\rank$ $1$. Hence, ordinals above $1$
are not $\rank_{\F}$-indecomposable with respect to Delhomm{\'e}'s original
definition. The crucial point is that such a box-augmentation of two
infinite antichains in general is not tamely colorable.

\subsection{Delhomm{\'e}'s decomposition theorem for
  tree-automatic structures}
\label{sec:DelhommesDecomposition}

In this section we reprove Delhomm{\'e}'s decomposition
theorem~\cite{Delhomme04}.
Beforehand, we give the precise definitions of
sum- and box-augmentations.
After providing a proof of Delhomm{\'e}'s original result,
we introduce tamely colorable box-augmentations and
improve this result.
\emph{For the rest of this section, we fix a relational signature
$\tau=(\Rel,\ar)$.}

\begin{defi}
	\label{def:sum-augmentation}
	\label{def:box-augmentation}
	Let $\Af$ and $\Bf_1,\ldots,\Bf_n$ be $\tau$-structures.
	\begin{enumerate}[(1)]
	\item We say that $\Af$ is a
 	 	\emph{sum-augmentation} of
  		$(\Bf_1,\ldots,\Bf_n)$
		if there is a partition
		\begin{equation*}
			A = \bigsqcup_{1\leq i \leq n} A_i
		\end{equation*}
		of the domain of $\Af$ such that
	  	$\Af{\restriction}_{A_i} \cong \Bf_i$
		for all $1\leq i\leq n$.
	\item We say that $\Af$ is a
  		\emph{box-augmentation} of
  		$(\Bf_1, \ldots, \Bf_n)$
		if there is a bijection
		\begin{equation*}
			\eta:\prod_{1\leq i \leq n} B_i \to A
		\end{equation*}
		such that for each $1\leq k \leq n$ and all
		$\bar d=(d_1,\ldots,d_{k-1},d_{k+1},\ldots,d_n)\in
		\prod_{1\leq i\leq n,i\not=k} B_i$
                \begin{align*}
                  &\eta^{\bar d}_k : B_k \to A \\
                  &\eta^{\bar d}_k(e) =
                  \eta(d_1, \ldots d_{k-1}, e, d_{k+1}, \ldots d_n)
		\end{align*}
   		is an (isomorphic) embedding of $\Bf_k$ into $\Af$.
	\end{enumerate}
\end{defi}

\begin{rem}
Whenever $\Af$ is a sum- or box-augmentation of $(\Bf_1,\ldots,\Bf_n)$,
the $\Bf_i$ are substructures of $\Af$.
In particular,
if $\Af$ is a forest or a \wulpo, the $\Bf_i$ are also
forests or \wulpo's, respectively.
\end{rem}

To simplify notation, for a $\tau$-structure $\Af$,
an $\FOext[\tau]$-formula $\varphi(x,y_1,\ldots,y_n)$, and
a tuple $\bar s\in A^n$
we use $\Af{\restriction}_{\varphi,\bar s}$ as an abbreviation
for $\Af{\restriction}_{\varphi^{\Af}(\cdot,\bar s)}$. 
Delhomm{\'e}'s decomposition theorem is the following.

\begin{thm}
  \label{thm:delhomme}
  From a given tree-automatic $\tau$-structure $\Af$
  and an
  $\FOext[\tau]$-formula
  \mbox{$\varphi(x,y_1,\dotsc,y_n)$,}
  one can compute a finite set
  $\mathcal{S}^{\Af}_\varphi$
  of tree-automatic structures
  such that for all $\bar s\in A^n$ the substructure
  $\Af{\restriction}_{\varphi,\bar s}$
  is a sum-augmentation of box-augmentations of elements
  from~$\mathcal{S}^{\Af}_\varphi$.
\end{thm}

\begin{proof}
Let $\A_R$, where $R \in \Rel$, (resp. $\A_\varphi$) be a
bottom-up deterministic tree automaton accepting
$\otimes R^{\Af}$ (resp. $\otimes\varphi^{\Af}$) and
$Q_R$ (resp. $Q_\varphi$) be its state set.
For each $t=(T,\lambda) \in \trees{\Sigma}$ and all $R\in\Rel$ of
arity $r$ we define
$\mathop{\otimes}\limits_{R} t = \otimes(t,\ldots,t) \in\trees{\Sigma_\diamond^{r}}$
where the convolution is made up of $r$ many copies of~$t$.
We further consider the tree
$t\mathop{\otimes}\limits_n \emptyset=(T,\lambda') \in
\trees{\Sigma_\diamond^{1+n}}$ 
with $\lambda'(u)=(\lambda(u),\diamond,\ldots,\diamond)$,
where the number of $\diamond$-symbols is~$n$.

As a first step towards proving the claim,
we construct the set $\mathcal{S}_\varphi^{\Af}$.
Before we construct the set, let us give some intuition on the
structures we are interested in. $\mathcal{S}_\varphi^{\Af}$ consists
of structures $\Bf_\gamma$ that are substructures of 
the structures $\Af{\restriction}_{\varphi,\bar s}$ which are obtained
as follows: we fix an element $t$ of $\Af{\restriction}_{\varphi,\bar s}$ and
we fix a node $d$ in $t$ that is outside of the domain of $\bar
s$. Now the domain of $\Bf_\gamma$ consists of all elements of 
$\Af{\restriction}_{\varphi,\bar s}$ that are obtained from $t$ by
replacing the subtree rooted at $d$ by some other tree. Thus, the
structures $\Bf_\gamma$ are the substructures of all
$\Af{\restriction}_{\varphi,\bar s}$ whose domain is obtained by only
local changes to a fixed tree.

Let
	$\Gamma = \prod_{R\in\Rel} Q_R \times Q_\varphi
		\times \prod_{R\in\Rel} 2^{Q_R}$.
For each
$\gamma=((q_R)_{R\in\Rel},q_\varphi,(P_R)_{R\in\Rel})\in\Gamma$
we define a structure $\Bf_\gamma$ with domain
\begin{equation} \label{eq:domain_B_gamma}
	B_\gamma = \{ t\in\trees\Sigma \mid
		\text{$\A_R(\mathop{\otimes}\limits_{R} t)=q_R$ for
                  all $R\in\Rel$  and 
		$\A_\varphi(t\mathop{\otimes}\limits_n \emptyset)=q_\varphi$} \}
\end{equation}
and relation $R\in\Rel$ (with arity $r=\ar(R)$) interpreted by
\begin{equation}
	\label{eq:relations_B_gamma}
	R^{\Bf_\gamma} = \{ \bar t\in (B_\gamma)^r \mid
		\A_R(\otimes\bar t)\in P_R \}.
\end{equation}
Clearly, $\Bf_\gamma$ is tree-automatic.
Finally, we let $ \mathcal{S}_\varphi^{\Af}$ be the finite set
\begin{equation*}
	\mathcal{S}_\varphi^{\Af} = \{ \Bf_\gamma \mid \gamma\in\Gamma \}.
\end{equation*}
It remains to show that for each tuple $\bar s=(s_1,\dotsc,s_n) \in \Af^n$
the substructure $\Af{\restriction}_{\varphi,\bar s}$
is a sum-augmentation of box-augmentations of
structures from $\{\Bf_\gamma \mid \gamma\in\Gamma\}$.

For this purpose, we fix such a tuple $\bar s$ and
consider the finite binary tree
$D=\bigcup_{1\leq i\leq n} S_i$,
where $s_i=(S_i,\mu_i)$.
The \emph{$\bar s$\nobreakdash-type} of a tree $t=(T,\lambda)\in\trees\Sigma$
is the tuple
\begin{equation*}
	\tp_{\bar s}(t)=
	(t{\restriction}_D,U,(\zeta_R)_{R\in\Rel},\zeta_\varphi),
\end{equation*}
where
\begin{enumerate}[(1)]
	\item $t{\restriction}_D=(T\cap D,\lambda{\restriction}_{(T\cap D)})
		\in \trees{\Sigma}$
		is the restriction of $t$ to the domain $T\cap D$,
	\item $U=T\cap\fr(D)\subseteq\fr(T\cap D)$,
	\item $\zeta_R:U\to Q_R$ is a map with
		$\zeta_R(u)=\A_R(\mathop{\otimes}\limits_{R} t(u))$,
		and
	\item $\zeta_\varphi:U\to Q_\varphi$ is a map with
		$\zeta_\varphi(u)=\A_\varphi(t(u)\mathop{\otimes}\limits_n
                \emptyset)$. 
\end{enumerate}
\begin{figure}[t]
	\centering
	\begin{tikzpicture}[knoten/.style={inner sep=0,minimum size=5mm,fill,circle}]
	
	\draw (0,0.5) node {$t$};
	\begin{scope}[thick,every node/.style={inner sep=0,minimum size=2mm,fill,circle}
	             ,level/.style={sibling distance=6cm/#1,level distance=1.5cm}
	             ,label distance=1mm
	             ]
		\node (u) at (0,0) {}
			child { node (u0) {}
				child { node[label=right:$u_1$] (u00) {}
					child { node (u000) {} }
					child { node (u001) {} }
				}
				child { node[label=left:$u_2$] (u01) {}
					child { node (u010) {} }
					child { node (u011) {} }
				}
			}
			child { node (u1) {}
				child { node (u10) {} }
				child { node (u11) {} }
			};
	\end{scope}
	
	\begin{scope}
		\coordinate (c) at ($(u000)!0.5!(u011)$);
		\draw ($(c)!2!(0,-4.5)$) node {$D$};
		\draw ($(u)+(63.4:0.25)$)
			-- ($(u1)+(63.4:0.25)$) arc (63.4:45:0.25)
			-- ($(u11)+(45:0.25)$) arc (45:26.6:0.25)
			-- ($(u000)!2!(0,-4.5)+(26.6:0.25)$) arc (26.6:-90:0.25)
			-- ($(u011)!2!(0,-4.5)+(-90:0.25)$) arc (-90:-206.6:0.25)
			-- ($(u10)+(153.4:0.25)$) arc (153.4:135:0.25)
			-- ($(u1)+(-0.25,-0.25)+1.141*(-0.25,0)$)
			-- ($(u0)+(-90:0.25)$) arc (-90:-243.4:0.25)
			-- ($(u)+(116.6:0.25)$) arc (116.6:63.4:0.25);
	\end{scope}

	\begin{scope}
		\draw ($(u)+(0,-1)$) node {$t{\restriction}_D$};
		\draw[semithick,densely dashed] ($(u)+(63.4:0.15)$)
			-- ($(u1)+(63.4:0.15)$) arc (63.4:45:0.15)
			-- ($(u11)+(45:0.15)$) arc (45:-90:0.15)
			-- ($(u10)+(-90:0.15)$) arc (-90:-225:0.15)
			-- ($(u1)+(-0.15,-0.15)+1.141*(-0.15,0)$)
			-- ($(u0)+(-90:0.15)$) arc (-90:-243.4:0.15)
			-- ($(u)+(116.6:0.15)$) arc (116.6:63.4:0.15);
	\end{scope}
	
	\begin{scope}
		\draw ($(u00)+(0,-1)$) node {$t(u_1)$};
		\draw ($(u00)+(26.6:0.15)$)
			-- ($(u001)+(26.6:0.15)$) arc (26.6:-90:0.15)
			-- ($(u000)+(-90:0.15)$) arc (-90:-206.6:0.15)
			-- ($(u00)+(153.4:0.15)$) arc (153.4:26.6:0.15);
	\end{scope}

	\begin{scope}
		\draw ($(u01)+(0,-1)$) node {$t(u_2)$};
		\draw ($(u01)+(26.6:0.15)$)
			-- ($(u011)+(26.6:0.15)$) arc (26.6:-90:0.15)
			-- ($(u010)+(-90:0.15)$) arc (-90:-206.6:0.15)
			-- ($(u01)+(153.4:0.15)$) arc (153.4:26.6:0.15);
	\end{scope}
	
	\begin{scope}
		\draw ($(u00)!0.5!(u01)$) node {$U$};
		\draw[semithick,densely dotted] ($(u00)+(90:0.3)$)
			-- ($(u01)+(90:0.3)$) arc (90:-90:0.3)
			-- ($(u00)+(-90:0.3)$) arc (-90:-270:0.3);
	\end{scope}
	\end{tikzpicture}
	\caption{The $\bar s$-type of some example tree $t$. The dashed line surrounds $t{\restriction}_D$ and $U$ consists of all nodes within the dotted line.}
	\label{fig:s-type}
\end{figure}
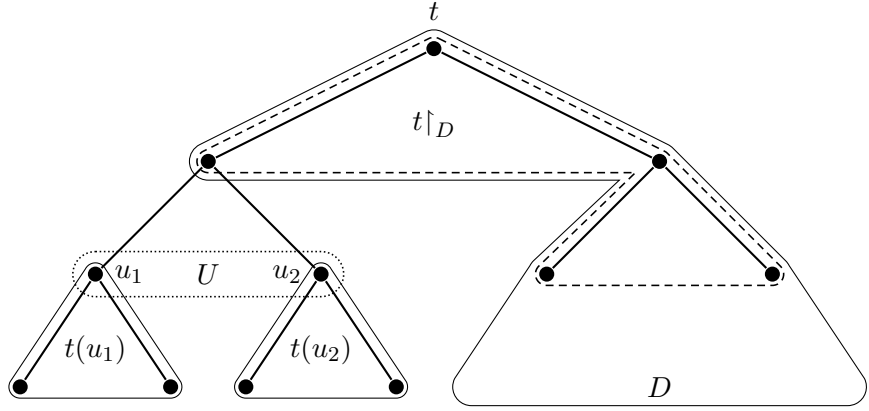
Figure~\ref{fig:s-type} illustrates the choice of $t{\restriction}_D$ and $U$. 
Observe that for $t_u:= t(u)\mathop{\otimes}\limits_n\emptyset$ the
 equation 
\begin{equation*}
	\otimes(t,\bar s) =
		\otimes(t{\restriction}_D,\bar
                s)[(u/ t_u)_{u\in U}]
\end{equation*}
holds, whence
\begin{equation}
	\label{eq:tp_saturates_phi}
	\A_\varphi(\otimes(t,\bar s)) =
		\A_\varphi(\otimes(t{\restriction}_D,\bar s),\zeta_\varphi).
\end{equation}
Therefore, $\tp_{\bar s}(t)$ determines whether $t\in\varphi^\Af(\cdot,\bar s)$.
Since $D$ and $\Rel$ are finite,
there are only finitely many distinct $\bar s$\nobreakdash-types. Consequently, the equivalence relation $\sim_{\bar s}$
on $\trees\Sigma$ defined by
$t\sim_{\bar s} t'$ if and only if $\tp_{\bar s}(t)=\tp_{\bar s}(t')$
has finite index.
Due to \eqref{eq:tp_saturates_phi},
$\sim_{\bar s}$ saturates the set $\varphi^{\Af}(\cdot,\bar s)$, i.e.,
$\varphi^{\Af}(\cdot,\bar s)$ is a disjoint union of some
$\sim_{\bar s}$\nobreakdash-classes.
Assume that $Z_1,\dotsc,Z_\ell\subseteq\varphi^{\Af}(\cdot,\bar s)$
are these 
$\sim_{\bar s}$-classes.
Then $\Af{\restriction}_{\varphi,\bar s}$
is a sum-augmentation of
$\Af{\restriction}_{Z_1},\ldots,\Af{\restriction}_{Z_\ell}$.

As the next step, we fix a single $\sim_{\bar s}$\nobreakdash-class
$Z\subseteq\varphi^{\Af}(\cdot,\bar s)$.
We show that $\Zf=\Af{\restriction}_Z$ is a box-augmentation of  a tuple
of structures from $\mathcal{S}_\varphi^{\Af}$.
Let $(t,U,(\zeta_R)_{R\in\Rel},\zeta_\varphi)$
be the $\bar s$-type corresponding to $Z$.
Note that the domain of $t$ is contained in the set $D$.
For $u\in U$ we define
\begin{equation*}
	\gamma(u) = 
	((\zeta_R(u))_{R\in\Rel},\zeta_\varphi(u),
		(P_R(u))_{R\in\Rel})
	\in\Gamma,
\end{equation*}
where for each relation $R$
\begin{equation} \label{P_R(u)}
	P_R(u) =
		\{ q\in Q_R \mid \A_R(\mathop{\otimes}\limits_R
                t,\zeta_R[u\mapsto 
                q])\in I_R \}. 
\end{equation}
Here $I_R\subseteq Q_R$ is the set of initial states of $\A_R$, and
$\zeta_R[u\mapsto q]$ is the map $\zeta_R$
with the value at position $u$ changed to $q$.
Finally, we put $\Xf_u = \Bf_{\gamma(u)}$ and
denote the domain of $\Xf_u$ by $X_u$.
By definition of the the domain $B_{\gamma(u)}$ of
$\Xf_u$ (see \eqref{eq:domain_B_gamma}), $y \in X_u$ is equivalent to
\begin{align} 
\A_R(\mathop{\otimes}\limits_R y) & =  \zeta_R(u) \text{ for all } R \in\Rel \label{eq-in-X_u(1)}\\
\A_\varphi(y \mathop{\otimes}\limits_n \emptyset) & = 
\zeta_\varphi(u) \label{eq-in-X_u(2)} .
\end{align}
It remains to prove that $\Zf$ is a box-augmentation of
$(\Xf_u)_{u\in U}$.
First, observe that the map
\begin{equation*}
	\eta: \prod_{u\in U} X_u\to\trees\Sigma
\end{equation*}
defined by
\begin{equation*}
	\eta((x_u)_{u\in U}) = t[(u/x_u)_{u\in U}]
\end{equation*}
is injective.
Moreover, a tree $x=(X,\lambda)\in\trees\Sigma$ is contained in the
image of $\eta$ 
if and only if
\begin{enumerate}[(1)]
	\item $x{\restriction}_D=t$,
	\item $X\cap\fr(D)=U$, and
	\item $x(u)\in X_u$ for each $u\in U$.
\end{enumerate}
Note that by \eqref{eq-in-X_u(1)} and  \eqref{eq-in-X_u(2)}, 
condition (3) is equivalent to 
$\A_R(\mathop{\otimes}\limits_R x(u)) = \zeta_R(u)$ for all 
$u \in U$ and $R \in\Rel$ and $\A_\varphi(x(u) \mathop{\otimes}\limits_n \emptyset) = 
\zeta_\varphi(u)$ for all $u \in U$.
Hence, the conjunction of (1), (2), and (3) is equivalent to 
$\tp_{\bar s}(x)=(t,U,(\zeta_R)_{R\in\Rel},\zeta_\varphi)$ which again is
equivalent to
$x\in Z$.
Thus, $\eta$ is a bijection onto its image $Z$.

Finally, we show that $\eta$ is a componentwise embedding of the
$(\Xf_u)_{u\in U}$ into $\Zf$. 
For $u_0\in U$ and
$\bar x=(x_u)_{u\in U\setminus\{u_0\}}\in
\prod\limits_{u\in U\setminus\{u_0\}} X_u$
define the map $\eta_{u_0}^{\bar x}: X_{u_0} \to Z$ by
\begin{equation*}
	\eta_{u_0}^{\bar x}(x_{u_0}) = \eta((x_u)_{u\in U}).
\end{equation*}
For all $R\in\Rel$ of arity $r$ and $(t_1, \dots, t_r) \in (X_{u_0})^r$ we have
\begin{equation*}
	\otimes \left(\eta_{u_0}^{\bar x}(t_1), \dots, \eta_{u_0}^{\bar
          x}(t_r)\right) = 
	\left( \mathop{\otimes}\limits_{R} t
	\left[(u /  x_u)_{u\in U\setminus\{u_0\}}\right] \right)
      \left [u_0/\otimes
      (t_1, \dots, t_r)\right]. 
\end{equation*}
Since $\A_R( \mathop{\otimes}\limits_{R} {x_u})=\zeta_R(u)$ for all
$u\in U\setminus\{u_0\}$, 
this implies
\begin{equation} \label{eq-componentwise}
	\A_R(\otimes (\eta_{u_0}^{\bar x}(t_1),\dots,\eta_{u_0}^{\bar
          x}(t_r))) =
		\A_R(\mathop{\otimes}\limits_{R}t,\zeta_R[u_0\mapsto
                \A_R(\otimes(t_1, \dots, t_r))])\,.
\end{equation}
We obtain the following chain of equivalences
\begin{align*}
	(\eta_{u_0}^{\bar x}(t_1), \dots, \eta_{u_0}^{\bar
          x}(t_r)) \in R^{\Zf} 
	&\quad\Longleftrightarrow\quad
	(\eta_{u_0}^{\bar x}(t_1), \dots, \eta_{u_0}^{\bar
          x}(t_r)) \in R^{\Af}\\
	&\quad\Longleftrightarrow\quad
        \A_R(\otimes 
        (\eta_{u_0}^{\bar x}(t_1), \dots, \eta_{u_0}^{\bar
          x}(t_r)))\in I_R \\
        &\quad\stackrel{\text{\eqref{eq-componentwise}}}{\Longleftrightarrow}\quad
        \A_R(\mathop{\otimes}\limits_{R} t,
        \zeta_R[u_0\mapsto \A_R(\otimes(t_1, \dots,
        t_r))])\in I_R \\ 
	&\quad\stackrel{\text{\eqref{P_R(u)}}}{\Longleftrightarrow}\quad
		\A_R(\otimes(t_1, \dots, t_r))\in P_R(u_0) \\
	&\quad\stackrel{\text{\eqref{eq:relations_B_gamma}}}{\Longleftrightarrow}\quad
		(t_1, \dots, t_r) \in R^{\Xf_{u_0}}
\end{align*}
showing that $\eta_{u_0}^{\bar x}$ is
an embedding of $\Xf_{u_0}$ into $\Zf$.
Consequently, $\Zf$ is a box-augmentation of $(\Xf_u)_{u\in U}$.
\end{proof}

After all, let us take a closer look at the properties of the
box-augmentation $\Zf$ of $(\Xf_u)_{u\in U}$ in the previous proof.
Consider some $R\in\Rel$ and let $r=\ar(R)$.
For all
$\bar x_1,\ldots,\bar x_r\in\prod_{u\in U} X_u$ with
$\bar x_i=(x_{i,u})_{u\in U}$
we have
\begin{equation*}
	\otimes(\eta(\bar x_1),\ldots,\eta(\bar x_r)) =
		(\mathop{\otimes}\limits_{R}
                t)[(u/\otimes(x_{1,u},\ldots,x_{r,u}))_{u\in U}] 
\end{equation*}
and hence
\begin{equation}
	\label{eq:induced_coloring}
	\A_R(\otimes(\eta(\bar x_1),\ldots,\eta(\bar x_r))) =
		\A_R(\mathop{\otimes}\limits_{R} t,\zeta),
\end{equation}
where $\zeta:U\to Q_R$ is defined by
$\zeta(u)=\A_R(\otimes(x_{1,u},\ldots,x_{r,u}))$.
Thus, the tuple $(\zeta(u))_{u\in U}$ 
determines  whether 
$(\eta(\bar x_1),\ldots,\eta(\bar x_r))\in R^\Af$.

This observation leads to an improved version of
Theorem~\ref{thm:delhomme},
namely Corollary~\ref{cor:tame-delhomme} below.
As an abstraction of our observation, 
we introduce the notion of
tamely colorable box-augmentations
in the following definition.
Therein for $R\in\Rel$
a \emph{finite $R$-coloring} of a $\tau$-structure $\Af$
is a map $\sigma:A^{\ar(R)}\to C$ into a finite set $C$ such that
for all $\bar x,\bar y\in A^{\ar(R)}$ with
$\sigma(\bar x)=\sigma(\bar y)$ we have $\bar x\in R^\Af$
if and only if $\bar y\in R^\Af$.

\begin{defi}
	\label{defi:tame-box-augmentation}
	Let $\A$ and $\Bf_1,\ldots,\Bf_n$ be $\tau$-structures.
	We say that $\Af$ is a
	\emph{tamely colorable box-augmentation}
	of $(\B_1,\ldots,\Bf_n)$ if
	\begin{enumerate}[(1)]
		\item there exists a bijection
			$\eta:\prod_{1\leq i \leq n} B_i \to A$
			witnessing  that $\Af$ is a box-augmentation of
			$(\Bf_1,\ldots,\Bf_n)$ and
		\item for every $R\in\Rel$ of arity $r=\ar(R)$
                  there are finite $R$-colorings $\sigma_i:B_i^r\to C_i$
                  of each $\Bf_i$
                  such that the map
                  $\sigma: A^r\to C_1\times\cdots\times C_n$
                  mapping $(a_1, \dots, a_r)\in A^r$ with
                  $a_j=\eta(b_{j,1},\ldots,b_{j,n})$ for $1 \leq j
                  \leq r$  to
                  \begin{equation*}
                    \sigma(a_1,\ldots,a_r)
                    = \left(\sigma_1(b_{1,1},\ldots,b_{r,1}),\ldots,
                      \sigma_n(b_{1,n},\ldots,b_{r,n})\right),
                  \end{equation*}
                  is a finite $R$-coloring of $\Af$.
	\end{enumerate}
\end{defi}

\noindent Roughly speaking, the last condition says that
the colors assigned to the tuples of the
``componentwise preimages'' of the $a_i$ under $\eta$
already determine whether $(a_1,\dotsc,a_r)\in R^\Af$.

\begin{rem}
In the situation of Definition~\ref{defi:tame-box-augmentation},
assume that $\Af$ and the $\Bf_i$ are directed graphs and
all $C_i$ are the same set, say $\{1,\ldots,m\}$.
For $1\leq i\leq n$ the structure
$\Xf_i=\bigl(B_i;R_1^{\Xf_i},\dotsc,R_m^{\Xf_i}\bigr)$
with ${R_c^{\Xf_i}=\sigma_i^{-1}(c)}$
can be regarded as a coloring of the edges of
the complete directed graph on $B_i$ with $m$ colors.
Since this coloring is compatible with the edge relation of $\Bf_i$,
the graph $\Af$ is a generalized
product---in the sense of Feferman and Vaught---of
the structures $\Xf_1,\ldots,\Xf_n$
using only atomic formulas.
\end{rem}

We conclude by proving the ``tamely-colorable'' version of
Theorem~\ref{thm:delhomme}.

\begin{cor}
  \label{cor:tame-delhomme}
  Given a tree-automatic $\tau$-structure $\Af$ and
  an $\FOext[\tau]$-formula $\varphi(x,\bar y)$,
  one can compute a finite set
  $\mathcal{S}^{\Af}_\varphi$
  of tree-automatic $\tau$-structures
  such that for all $\bar s\in A^n$ the substructure
  $\Af{\restriction}_{\varphi,\bar s}$
  is a sum-augmentation of tamely colorable box-augmentations
  of elements from~$\mathcal{S}^{\Af}_\varphi$.
\end{cor}

\begin{proof}
We show that the box-augmentation $\Zf$ of $(\Xf_u)_{u\in U}$
constructed in the proof of Theorem~\ref{thm:delhomme} is tamely colorable.
Therefore, consider the bijection $\eta:\prod_{u\in U} X_u\to Z$
witnessing this box-augmentation and
fix some $R\in\Rel$ with arity $r=\ar(R)$.
Due to the definition of $R^{\Xf_u}$ in
\eqref{eq:relations_B_gamma},
for each $u\in U$ the map $\sigma_u:X_u^r\to Q_R$ with
$\sigma_u(\bar x)=\A_R(\otimes\bar x)$ is a finite $R$-coloring of $\Xf_u$.
Finally, \eqref{eq:induced_coloring} shows that
condition (2) from Definition~\ref{defi:tame-box-augmentation} holds.
\end{proof}

\subsection{Sum- and box-indecomposability}
\label{sec:Indecomposability}

Suppose that $\C$ is a class of $\tau$-structures and
$\nu$ is a function assigning 
an ordinal $\nu(\Af)$ to each $\Af\in\C$.\footnote{All our
  applications use isomorphism invariant functions, which means that
  $\nu(\Af)=\nu(\Bf)$ if $\Af\cong\Bf$. In fact, the results of this
  section apply to arbitrary class functions but are only useful if
  the function is isomorphism invariant for the important part of the
  class $\C$ because for other functions there are not enough
  indecomposable values. Thus, upon first reading of the following
  part, we recommend the reader to think of $\nu$ as an isomorphism
  invariant function.}
In this situation, we say that $\C$ is \emph{ranked} by $\nu$.

\begin{defi}\label{def:indecomposable value}
	Let $\C$ be a class of $\tau$-structures ranked by $\nu$.
  \begin{enumerate}[(1)]
  \item An ordinal $\alpha$ is called \emph{$\nu$-sum-indecomposable} if
  	for every $\Af\in\C$ with $\nu(\Af)=\alpha$ and
	all $\tau$-structures $\Bf_1,\ldots,\Bf_m$ such that
    $\Af$ is a sum-augmentation of  $(\Bf_1,\ldots,\Bf_m)$,
    there is $1 \leq i\leq m$ such that $\Bf_i\in\C$ and $\nu(\Bf_i)=\alpha$.
  \item
    An interval $[\alpha_1,\alpha_2]$ of ordinals is called
    \emph{$\nu$-tamely-colorable-box-indecomposable} if
    for every $\Af\in\C$ with $\nu(\Af)=\alpha_2$ and
    all $\tau$-structures $\Bf_1,\ldots,\Bf_m$ such that
    $\Af$ is a tamely colorable box-aug\-men\-tation of
    $(\Bf_1,\ldots,\Bf_m)$,
    there is $1 \leq i\leq m$ such that
    $\Bf_i\in\C$ and ${\nu(\Bf_i)\in [\alpha_1,\alpha_2]}$.
  \end{enumerate}
\end{defi}

\begin{rem}
  For classes $\C$ which are closed under taking substructures,
  like the classes of forests and \wulpo's,
  the requirement $\Bf_i\in\C$ is always satisfied. Hence in this case
  explicitely requiring $\Bf_i\in C$ is not necessary.
\end{rem}

The decomposition results from the previous section imply that $\nu$
may only take finitely many $\nu$-sum-indecomposable and
$\nu$-tamely-colorable-box-indecomposable
values among the substructures of the form
$\Af{\restriction}_{\varphi,\bar s}$
(defined just before Theorem~\ref{thm:delhomme})
for a fixed $\FOext[\tau]$-formula $\varphi(x,\bar y)$.

\begin{prop} \label{Prop:Indecomposability}
  Let $\C$ be a class of $\tau$-structures ranked by $\nu$ and
  $\alpha_0<\alpha_1<\alpha_2<\dotsc$ an infinite sequence of
  $\nu$-sum-indecomposable ordinals such that
  $[\alpha_i,\alpha_{i+1}]$ is $\nu$-tamely-colorable-box-indecomposable
  for all $i\in\N$.
  Moreover, let $\Af$ be a tree-automatic $\tau$-structure and
  $\varphi(x,y_1,\dotsc,y_r)$ an $\FOext[\tau]$-formula.
  Then there are only finitely many $i\in\N$ which
  admit a tuple $\bar s\in A^r$ with
  $\Af{\restriction}_{\varphi,\bar s}\in\C$ and
  $\nu(\Af{\restriction}_{\varphi,\bar s})=\alpha_i$.
\end{prop}

\begin{proof}
	Let $\mathcal{S}_\varphi^{\Af}$ be the finite set of structures
	which exists by Corollary~\ref{cor:tame-delhomme}.	
	Consider an $i\in\N_{>0}$ satisfying the condition above,
	witnessed by $\bar s\in A^r$, i.e.,
	$\Af{\restriction}_{\varphi,\bar s}\in\C$ and
	$\nu(\Af{\restriction}_{\varphi,\bar s})=\alpha_i$.	
	There are structures $\Xf_1,\ldots,\Xf_k$ such that
	each of them is a tamely colorable box-augmentation of elements
	from~$\mathcal{S}_\varphi^{\Af}$ and
	$\Af{\restriction}_{\varphi,\bar s}$
	is a sum-augmentation of $(\Xf_1,\ldots,\Xf_k)$.
	Due to the definition of $\nu$-sum-indecomposability,
	there is a $1\leq j\leq k$ such that $\Xf_j\in\C$ and
        $\nu(\Xf_j)=\alpha_i$. 
	There are structures
	$\Yf_1,\ldots,\Yf_\ell\in\mathcal{S}_\varphi^{\Af}$
	such that $\Xf_j$ is a tamely colorable box-augmentation of
	$(\Yf_1,\ldots,\Yf_\ell)$.
	By the definition of $\nu$-tamely-colorable-box-indecomposability,
	there is an $1\leq h\leq\ell$ such that
	$\Yf_h\in\C$ and $\nu(\Yf_h)\in[\alpha_{i-1},\alpha_i]$.
	Thus, $i$ belongs to the set
	\begin{equation*}
		\{ i\in\N_{>0} \mid
			\exists \Bf\in\mathcal{S}_\varphi^\Af\cap\C:
			\nu(\Bf)\in[\alpha_{i-1},\alpha_i] \},
	\end{equation*}
	which is finite since each $\Bf\in\mathcal{S}_\varphi^\Af$
	satisfies the condition above for at most two distinct $i$'s.
\end{proof}

\subsection{Rank-tamely-colorable-box- and
  rank-sum-indecomposability} 
\label{sec:rank-sum-ind}
\label{sec:rank-box-ind}

Let us briefly prove that the ordinals $\omega^\alpha$ are
$\rank$-sum-indecomposable for the class of all well-founded
partial orders. Afterwards, we characterize
$\rank$-tamely-colorable-box-indecomposable intervals for the classes
of \wulpo's and well-founded forests.
For this purpose let $\rank_{ul}$ be the function obtained from
$\rank$ by restriction of the domain to the class of \wulpo's and
analogously let $\rank_F$ be $\rank$ restricted to the class of
well-founded forests. 

\begin{prop} \label{prop:rank-sum-indecomposable}
  The ordinals of the form $\omega^\alpha$ are
  $\rank$-sum-indecomposable for the class of well-founded
  partial orders.
\end{prop}
\begin{proof}
  Let $\Pf=(P, \leq)$ be a well-founded
  partial order and assume that $\Pf$ is a sum-augmentation 
  of $(\Pf_1, \dots, \Pf_n)$. 
  If 
  \mbox{$\rank(\Pf\restriction_{P \setminus P_1})<\omega^\alpha$}
  and 
  $\rank(\Pf_1)<\omega^\alpha$, then
  \begin{align*}
    \rank(\Pf)
    \stackrel{\text{Lemma~\ref{lem:rankNaturalSum}}}{\leq}
    \rank(\Pf_1) \oplus 
    \rank(\Pf\restriction_{P \setminus P_1}) < \omega^\alpha,
  \end{align*}
  where the last inequality follows from Property
  \eqref{Eqn-Oplus-Limit} of $\oplus$ (see page \pageref{Eqn-Oplus-Limit}).
  Thus,
  ${\rank(\Pf)=\omega^\alpha}$ implies
  $\rank(\Pf_1)=\omega^\alpha$ or
  $\rank(\Pf\restriction_{P \setminus P_1})=\omega^\alpha$.
  The claim follows by induction on $n$.
\end{proof}
It follows trivially that the ordinals of the form $\omega^\alpha$ are
$\rank_{ul}$-sum-indecomposable and $\rank_F$-sum-indecomposable.

We want to show
that 
$[\omega^{\omega^\alpha}, \omega^{\omega^{\alpha+1}}]$ is a
$\rank_{ul}$-tamely-colorable-box-indecomposable interval for each ordinal
$\alpha$
and 
$[\omega^{\alpha}, \omega^{\alpha+1}]$ is a 
$\rank_{F}$-tamely-colorable-box-indecomposable interval for each ordinal
$\alpha$. We start with the observation
that every box-decomposition of a \wulpo only contains at most one
proper \wulpo in the sense that if a \wulpo is a box-augmentation of
$(\Pf_1, \Pf_2, \dots, \Pf_n)$ then all but one of the $\Pf_i$
are disjoint unions of ordinals.
In order to prove this fact, we introduce the following notation.
Let $\Pf=(P, \leq)$ be some partial order. We call $a\in P$ a \emph{branching
  node (of $\Pf$)}, if there
are $b, c\in P$ such that $b<a$, $c<a$ and neither $b\leq c$ nor
$c\leq b$ (i.e., $b$ and $c$ are incomparable).

\begin{lem} \label{lemma-box-aug-branching}
  Let $\Pf$, $\Pf_1$, and $\Pf_2$ be \wulpo's.
  If $\Pf$ is a box-augmentation of $(\Pf_1, \Pf_2)$
  then
  $\Pf_1$ or $\Pf_2$ does not contain a branching node.
\end{lem}
\begin{proof}
  Let $\Pf=(P, \leq)$ and $\Pf_i=(P_i,\leq_i)$. 
  Heading for a contradiction assume that
  \mbox{$a_i, b_i, c_i \in P_i$} for $i\in\{1,2\}$ are nodes such that
  $b_i <_i a_i, c_i<_i a_i$ and 
  neither $b_i\leq_i c_i$ nor $c_i\leq_i b_i$.

  Let $\eta: P_1 \times P_2 \to P$ be the bijection that witnesses
  that  $\Pf$ is a box-augmentation of $(\Pf_1, \Pf_2)$.
  Then $\Pf$ contains the chains
  \begin{align*}
    &\eta(b_1, b_2) < \eta(b_1, a_2)\text{ and}\\
    &\eta(b_1, b_2) < \eta(a_1, b_2).
  \end{align*}
  Since $\Pf$
  is a \wulpo, the elements above $\eta(b_1,b_2)$ are linearly ordered and
  we may assume that $ \eta(b_1, a_2) < \eta(a_1, b_2)$ without
  loss of generality.
  Thus, we obtain
  \begin{align*}
    \eta(b_1,c_2) < \eta(b_1, a_2) < \eta(a_1, b_2).
  \end{align*}
  Furthermore, we have
  \begin{align*}
    \eta(b_1, c_2) < \eta(a_1, c_2).
  \end{align*}
  Again, the elements above $\eta(b_1,c_2)$ are linearly ordered
  and we obtain that $\eta(a_1, c_2)$ and $\eta(a_1, b_2)$ are
  comparable in $\Pf$. 
  By definition of a box-augmentation, we obtain that $c_2$
  and $b_2$ are comparable in $\Pf_2$ as well, which contradicts our
  assumptions. 
  Thus,  $\Pf_1$ or $\Pf_2$ does not contain a
  branching  node. 
\end{proof}

\begin{cor} \label{coro-at-most-1-branching}
  Let $\Pf$ and $\Pf_1, \dots, \Pf_n$ be \wulpo's such that $\Pf$ is a
  box-augmentation of $(\Pf_1, \dots, \Pf_n)$. There is at most one
  $i\in\{1, \dots, n\}$ such that $\Pf_i$ contains a branching node,
  i.e., there is an $i\in\{1, \dots, n\}$ such that
  $\Pf_k$ is a disjoint union of well-orders for all $k\neq i$.
\end{cor}

\begin{proof}
  Let $\eta:\prod_{i=1}^{n} \Pf_i \to \Pf$ be the bijection of the
  box-augmentation.
  Choose numbers ${1\leq j < k \leq n}$ and  a tuple
  $\bar b={ (b_1, \dots, b_{j-1},b_{j+1}, \dots, 
    b_{k-1}, b_{k+1}, \dots,  b_n)
    \in\prod_{i\in\{1, \dots n\}\setminus\{j,k\}} P_i}$ 
  arbitrarily  but fixed.
  Then $\eta^{\bar b}_{j,k}: P_j\times P_k \to P$ with
  $\eta^{\bar b}_{j,k}(b_j,b_k)=\eta(b_1, \dots, b_n)$ induces a
  box-augmentation of some sub\wulpo 
  $\Pf'\leq \Pf$. Application of
  Lemma~\ref{lemma-box-aug-branching} yields the 
  claim.
\end{proof}

\begin{rem} \label{rem-difference-forest-wulpos}
  In the following, our proofs for the case of \wulpo's and the case
  of well-founded forests proceed completely analogous. 
  The difference in the results stems from an observation
  concerning  Corollary~\ref{coro-at-most-1-branching}: if a
  well-founded forest
  $\Ff$ is a  box-augmentation of partial orders $\Pf_1, \dots, \Pf_n$
  then the 
  $\Pf_i$ occur as substructures of $\Ff$. Hence, these are also
  well-founded forests. But if a well-founded forest is a disjoint union of ordinals, all these
  ordinals must be finite. Thus, each connected component of such a
  disjoint union has finite rank and the whole structure has rank at
  most $\omega$. In contrast, if a disjoint union of ordinals is a
  tree-automatic 
  \wulpo each connected component is a tree-automatic ordinal whence
  its rank is bounded by $\omega^{\omega^\omega}$ and all smaller
  ordinals can be reached. This difference
  causes the different results with respect to box-indecomposability. 
\end{rem}

\begin{lem} \label{lemma-box-of-chains}
  Let $\Pf$ and $\Pf_1, \dots, \Pf_m$ be \wulpo's such  that $\Pf$ is a
  box-augmentation of
  $(\Pf_1, \dots, \Pf_m)$ via $\eta:\prod_{i=1}^{m} \Pf_i \to \Pf$.
  Let $I_i\subseteq P_i$ be a well-order.
  The substructure $\Pf' \subseteq \Pf$  induced by
  $\eta(\prod_{i=1}^{m} I_i)$ is a well-order.
\end{lem}

\begin{proof}
  Note that $\Pf'$ is a \wulpo because it is a substructure of $\Pf$. 
  Thus, it suffices to show that $\Pf'$ is linear. 
  Let $a^j_i\in I_i$ for each $1\leq i \leq m$ and $j\in\{1,2\}$. 
  Set $m_i=\min(a^1_i,a^2_i)$. 
  For $j\in\{1,2\}$,
  \begin{align*}
    \eta(m_1, m_2,\dots, m_n) \leq \eta(a^j_1, m_2, \dots, m_n) \leq
    \cdots \leq \eta(a^j_1, a^j_2, \dots, a^j_n). 
  \end{align*}
  Since the elements above $\eta(m_1, \dots, m_n)$ are 
  pairwise comparable,  the elements $\eta(a^1_1, \dots, a^1_n)$ and
  $\eta(a^2_1, \dots, a^2_n)$ are comparable. 
  Since the $a^j_i$ have been chosen arbitrarily any two elements of
  $\Pf'$ are comparable, i.e., $\Pf'$ is linear. 
\end{proof}

In the following lemma, the term ``replacement'' refers to the
replacement operation introduced at the end of
Section~\ref{sec:wulpos}.

\begin{lem}\label{lem:wulpo-box-respects-branching}
 If a \wulpo $\Pf$ is the box-augmentation of ordinals
 $\Cf_1, \dots, \Cf_n$ 
 and a \wulpo $\Qf$ via the bijection $\eta$, then $\Pf$ results from $\Qf$ by
 replacing every maximal branching free interval $I$ of $\Qf$ by 
 $\eta\left( (\prod_{i=1}^n \Cf_i) \times I\right)$ (which is a well-order by
 Lemma~\ref{lemma-box-of-chains}).
\end{lem}

\begin{proof}
  We first show the following two claims.
  \begin{enumerate}[(1)]
  \item  
    For
    all $q_0,q_1\in \Qf$  and for all $b_i,c_i\in \Cf_i$
    the nodes
    $\eta(b_1, \dots, b_n, q_0)$ and 
    $\eta(c_1, \dots, c_n, q_1)$ are incomparable if and only if
    $q_0$ and $q_1$ are incomparable.
  \item For all 
    $b_i,c_i\in \Cf_i$ and
    $q_0,q_1,q\in \Qf$ such that $q_0<q$ and $q_1<q$ but $q_0$
    and $q_1$ are incomparable, 
    $\eta(b_1, \dots, b_n, q_1) < \eta(c_1, \dots, c_n, q)$.
  \end{enumerate}
  The first statement is easy to show: If
  $\eta(b_1, \dots, b_n, q_0) \leq \eta(c_1, \dots, c_n, q_1)$
  then the elements
  $\eta(\max(b_1,c_1), \dots, \max(b_n,c_n), q_0)$ and
  $\eta(\max(b_1,c_1), \dots, \max(b_n,c_n), q_1)$ are both above 
  the element $\eta(b_1, \dots, b_n, q_0)$ and therefore comparable. By the 
  definition of a box-augmentation, also $q_0$ and $q_1$ are
  comparable. 
  For the other direction, if $q_0 \leq q_1$, then 
  \begin{equation*}
    \eta(\min(b_1,c_1), \dots, \min(b_n,c_n), q_0) \leq
    \eta(\min(b_1,c_1), \dots, \min(b_n,c_n), q_1).    
  \end{equation*}
  Thus, $\eta(b_1, \dots, b_n, q_0)$ and 
  $\eta(c_1, \dots, c_n, q_1)$ are both above 
  $\eta(\min(b_1,c_1), \dots, \min(b_n,c_n), q_0)$ and therefore
  comparable. 
  
  Let us now prove the second claim. 
  Due to the first claim, $q_1 < q$ implies that
  $\eta(b_1, \dots, b_n, q_1)$ and $\eta(c_1, \dots, c_n, q)$ are comparable.
  Moreover, 
  \begin{equation}
    \label{eq:incomparabilityPreserverd}
    \text{for all $a_i\in \Cf_i, \ \eta(a_1, \dots, a_n,q_0)$ and $\eta
    (a_1, \dots, a_n,q_1)$ are 
    incomparable,}
  \end{equation}
  because $q_0$ and $q_1$ are incomparable and $\Pf$ is a
  box-augmentation. 
  Similarly, it is clear that
  \begin{align*}
    &\eta(\min(b_1,c_1), \dots,
    \min(b_n,c_n), q_i) \leq \eta(\min(b_1,c_1), \dots,
    \min(b_n,c_n), q) \leq \eta(c_1, \dots, c_n, q) \text{ and}\\
    &\eta(\min(b_1,c_1), \dots,
    \min(b_n,c_n), q_i) \leq \eta(b_1, \dots, b_n, q_i)
  \end{align*}
  for  $i\in\{0,1\}$. Since $\Pf$ is a \wulpo, the
  nodes above 
  $\eta(\min(b_1,c_1), \dots, \min(b_n,c_n), q_i)$ are linearly
  ordered. Hence
  $\eta(c_1, \dots, c_n, q)$ and
  $\eta(b_1, \dots, b_n, q_i)$ are comparable for $i\in\{0,1\}$.
  Note that 
  \begin{equation*}
  	\eta(b_1, \dots, b_n, q_i) \leq \eta(c_1, \dots, c_n,q) \leq
  	\eta(b_1, \dots,  b_n, q_{1-i})
  \end{equation*}
  for $i \in \{0,1\}$
  would contradict \eqref{eq:incomparabilityPreserverd}.
  Analogously, if 
  \begin{align*}
  	&\eta(c_1, \dots, c_n,q) \leq \eta(b_1, \dots, b_n, q_1)  \text{
          and}\\
        &\eta(c_1, \dots, c_n,q)  \leq \eta(b_1, \dots,  b_n, q_0), 
  \end{align*}
  then
  upwards linearity would lead to a contradiction with 
  \eqref{eq:incomparabilityPreserverd}.
  Thus, we conclude that
  \begin{equation*}
    \eta(b_1, \dots, b_n, q_1) \leq \eta(c_1, \dots, c_n,q),
  \end{equation*}
  which shows the second claim.

  Using these two claims, we now show that
  for every pair $I_1,I_2$ of distinct maximal branching free
  intervals of $\Qf$ and for all $q_1\in I_1,q_2\in I_2$, and for all
  $b_i,c_i\in \Cf_i$ we have
  \begin{equation} \label{eq-same-branching-structure}
    q_1 < q_2 \quad \Longleftrightarrow \quad \eta(b_1,\dots, b_n, q_1) <
    \eta(c_1, \dots, c_n, q_2).
  \end{equation}
  First assume that $q_1< q_2$. Since $q_1$ and $q_2$ come from
  different maximal branching free intervals, Lemma
  \ref{lem:SeparationMBFreeInt} implies that there is a $q_3 < q_2$
  such that $q_1$ and $q_3$ are incomparable. 
  Due to the second claim, this immediately shows that
  $\eta(b_1,\dots, b_n, q_1) <
  \eta(c_1, \dots, c_n, q_2)$. 
  By symmetry, we conclude that 
  $q_2 < q_1$ implies $\eta(b_1,\dots, b_n, q_2) <
  \eta(c_1, \dots, c_n, q_1)$.
  Moreover, because of the first claim, if $q_1$ and $q_2$ are
  incomparable then   
  $\eta(b_1,\dots, b_n, q_1)$ and $\eta(c_1, \dots, c_n, q_2)$ are
  incomparable as well.
 
  By \eqref{eq-same-branching-structure}, $\Pf$ is obtained from $\Qf$
  by replacing every maximal branching free interval $I$ of $\Qf$
  by a box-augmentation of $(\Cf_1, \dots, \Cf_n, I)$, which is a well
  order by Lemma~\ref{lemma-box-of-chains}.
  This concludes the proof of the lemma.
\end{proof}

\begin{lem} \label{lemma:2-claims}
  Let $\Pf$ be  a \wulpo that is a box-augmentation of 
  ordinals $\Cf_1, \dots, \Cf_n$ and
  a  \wulpo $\Qf$ via the bijection $\eta$. 
  Let $q \in \Qf$ and $q' \in \Qf\cup\{\infty\}$
  such that $q < q'$ and
  $[q,q')$ is a maximal branching free interval 
  in $\Qf$. Let $\alpha := \rank([q,q'))>0$ and
  $0_i$ the minimal element of $\Cf_i$ for each $1\leq i \leq
  n$. 
  Then
  \begin{equation} \label{eq:1st-claim}
    \rank(\eta(0_1, 0_2, \dots, 0_n, q), \Pf) \leq 
    \rank(\Cf_1) \otimes \rank(\Cf_2) \otimes\dots\otimes
    \rank(\Cf_n) \otimes \rank(q, \Qf)
  \end{equation}
  and for each $\hat q\in [q, q')$ and $c_i\in\Cf_i$ we have
  \begin{equation} \label{eq:2nd-claim}
    \rank(\eta(c_1, c_2, \dots, c_n, \hat q), \Pf) <
    \rank(\Cf_1) \otimes \rank(\Cf_2) \otimes\dots\otimes
    \rank(\Cf_n) \otimes 
        (\rank(q, \Qf) + \alpha ).
  \end{equation}
\end{lem}

\begin{proof}
  We prove both claims simultaneously by induction. 
  If $\rank(q,\Qf)=0$, then $q$ is minimal in $\Qf$. By Lemma~\ref{lem:wulpo-box-respects-branching},
  $\eta(0_1, 0_2, \dots, 0_n, q)$ is minimal in $\Pf$.
  Hence, its rank is also $0$ as desired.

  Now assume that \eqref{eq:1st-claim} is true for some $q\in\Qf$ inducing
  a maximal branching free interval $[q,q')\subseteq\Qf$ and let $c_i\in\Cf_i$ and $\hat q\in [q,q')$.
  Due to Lemma~\ref{lem:wulpo-box-respects-branching}, 
  $I := \eta(\prod_{i=1}^n  \Cf_i \times [q,q'))$ is a maximal branching
  free interval of
  $\Pf$ with minimal element $\eta(0_1, 0_2, \dots, 0_n, q)$. The rank
  of $I$ can be at most  $\rank(\Cf_1) \otimes \dots \otimes \rank(\Cf_n) 
    \otimes \alpha$. The interval
    \begin{equation*}
      J := [\eta(0_1, \dots, 0_n, q), \eta(c_1, \dots, c_n, \hat q)) \subseteq
      \Pf      
    \end{equation*}
    is strictly contained in $I$ whence
  \begin{equation} \label{eq:rank-J}
  \rank(J) < \rank(\Cf_1) \otimes \dots \otimes \rank(\Cf_n) 
    \otimes \alpha .
  \end{equation}
  Thus, we get
  \begin{align*}
    &\rank(\eta(c_1, \dots, c_n, \hat q), \Pf)\\
    \overset{\text{Lem.}}{\underset{\ref{lem:IntervalRankShift}}{=}}  
    &\rank(\eta(0_1, \dots, 0_n, q), \Pf) +
    \rank(J)\\
     \overset{\text{\eqref{eq:rank-J}}}{<}  
    &\rank(\eta(0_1, \dots, 0_n, q), \Pf) + 
    \rank(\Cf_1) \otimes \dots \otimes \rank(\Cf_n) 
    \otimes \alpha \\
    \overset{\text{\eqref{eq:1st-claim}}}{\leq}
    &\left(\rank(\Cf_1) \otimes\dots\otimes
    \rank(\Cf_n) \otimes \rank(q, \Qf)\right)
    + \left( \rank(\Cf_1) \otimes \dots \otimes \rank(\Cf_n) 
    \otimes \alpha       \right)
    \\
    \overset{\text{Lem.}}
    {\underset{\ref{lem:subdistributivity-plus-natprod}}{\leq}}    
    &\rank(\Cf_1) \otimes\dots\otimes
    \rank(\Cf_n) \otimes (\rank(q, \Qf) + \alpha).
  \end{align*}
  Finally, let $q,q'\in\Qf$ be nodes such that $[q,q')$ is maximal
  branching free and all maximal branching free intervals below $q$
  satisfy the claims. 
  Lemma~\ref{lem:WulpoPartitionMaxBranchingFreeIntervals} shows that for
  each $q_1 < q$ there is a maximal branching free interval 
  $I:=[q_0, q_2)$ with $q_2\leq q$ such that $q_1\in I$. Then for all
  $c_i\in \Cf_i$  
  \begin{align*}
    \rank(\eta(c_1, c_2, \dots, c_n, q_1),\Pf) 
    \overset{\text{IH}}{<} &
    \rank(\Cf_1) \otimes \dots \otimes \rank(\Cf_n)\otimes 
    (\rank(q_0,\Qf) + \rank(I)) \\
    \overset{\text{Lem.}}{\underset{~\ref{lem:IntervalRankShift}}{\leq}} &
    \rank(\Cf_1) \otimes \dots \otimes \rank(\Cf_n)\otimes 
    (\rank(q_2,\Qf) \\
    \leq &
    \rank(\Cf_1) \otimes \dots \otimes \rank(\Cf_n)\otimes 
    (\rank(q,\Qf) .
  \end{align*} 
  By Lemma~\ref{lem:wulpo-box-respects-branching} all elements below 
  $\eta(0_1, 0_2, \dots, 0_n, q)$ are of the form $\eta(c_1, c_2,
  \dots, c_n, q_1)$ with $q_1 < q$. Thus, 
  \begin{equation*}
  	\rank(\eta(0_1,0_2, \dots, 0_n, q) \leq
    \rank(\Cf_1) \otimes \dots \otimes \rank(\Cf_n)\otimes 
    (\rank(q,\Qf),
  \end{equation*}
which proves the lemma.
\end{proof}

\begin{cor}\label{Cor:ProdConnectedWulpos}
  Assume that $\Pf$ is a \wulpo that is a box-augmentation of
  connected \wulpo's $\Qf_1, \dots, \Qf_n$.
  Then 
  \begin{align*}
    \rank(\Pf) \leq \rank(\Qf_1) \otimes\dots\otimes
    \rank(\Qf_n). 
  \end{align*}
\end{cor}

\begin{proof}
  By Corollary~\ref{coro-at-most-1-branching} we can assume that
  $\Qf_2, \dots, \Qf_n$ are ordinals with
  minimal elements $0_2, \dots, 0_n$. 
  For each node $p\in\Pf$ there is a maximal branching free interval
  $[p_0,p_1)\subseteq\Pf$ such that $p\in[p_0,p_1)$ (where
  $p_1=\infty$ is possible).  
  There is a $q_0\in \Qf_1$ such that $p_0=\eta(q_0, 0_2, \dots, 0_n)$
  and there is a $q_1\in \Qf_1\cup\{\infty\}$ such that
  $I:=[q_0,q_1)$ is maximal branching free.
  By Lemma~\ref{lemma:2-claims}, we have
  \begin{align*}
  \rank(p, \Pf) &< \rank(\Qf_2) \otimes\dots\otimes \rank(\Qf_n)
  \otimes (\rank(q_0, \Qf_1)+ \rank(I)) \\
  &\leq
  \rank(\Qf_2) \otimes\dots\otimes \rank(\Qf_n)
  \otimes \rank(\Qf_1),
  \end{align*}
  where in case that $p=p_0$ the strict inequality follows from
  $\rank(I)\geq 1$.
\end{proof}
Having studied connected \wulpo's, we now have to deal with
disconnected ones. 
For this purpose, we  have to
restrict our attention to tamely colorable box-augmentations.
As a first step, we analyze boxes of antichains.
If a \wulpo $\Pf$ is a tamely colorable box-augmentation of
$n$ antichains then the rank of
$\Pf$ is bounded by some constant that only depends on the tame
colorings of the antichains and on $n$.
In order to prove this fact, we first introduce a notion of \emph{same
  factor equivalence} on the elements of a box-augmentation. Elements
are equivalent with respect to this equivalence if and only if 
their preimages in each of the factors of the box are contained in the
same connected component. 

\begin{defi}
  Let $\Pf$ be a \wulpo that is a box-augmentation of 
  (\wulpo's) $(\Pf_1, \dots, \Pf_n)$.  
  Let $\Pf_i=(P_i, \leq_i)$. 
  For $a_i, b_i\in P_i$ 
  we write $\Con{i}{a_i}{b_i}$ if 
  $a_i$ and $b_i$ are in the same connected component of $\Pf_i$.
  Due to upwards linearity, this is equivalent to saying that
  $\Con{i}{a_i}{b_i}$ if there exists some $c_i\in P_i$ such that $a_i\leq_i
  c_i$ and $b_i \leq_i c_i$. 

  Let $\eta$ be the bijection witnessing that $\Pf$ is
  box-augmentation of $(\Pf_1, \dots, \Pf_n)$. 
  For $a,b\in P$ we define the \emph{same factors  equivalence} by
  $\ConBox{a}{b}$ if $\Con{i}{a_i}{b_i}$ for all $1\leq i \leq n$
  where $a=\eta(a_1,\dots,a_n)$ and $b=\eta(b_1,\dots,b_n)$.
  For $p\in P$ we write $[p]_{\ConBoxNull}$ for the equivalence class
  $\{p'\mid \ConBox{p'}{p}\}$. 
\end{defi}

\begin{lem} \label{lemma:at-most-c-classes-on-chain}
  Let $\Pf=(P, \leq)$ be a \wulpo that is a box-augmentation of 
  $(\Pf_1, \dots, \Pf_n)$ via the bijection $\eta$. If $C\subseteq P$ is
  a chain that contains 
  elements of $k$ distinct $\ConBoxNull$-classes, then $\Pf$ contains
  a chain $C'$ of length $k$ such that 
  $C'\subseteq \eta(A_1 \times\dots\times A_n)$ where $A_i$ is an
  antichain in $\Pf_i$. 
\end{lem}

\begin{proof}
  Let $a^1<a^2<\dots<a^k$ be a chain of pairwise
  $\ConBoxNull$-inequivalent elements of $\Pf$. 
  There are $a^j_i\in \Pf_i$ such that
  $a^j=\eta(a^j_1, \dots, a^j_n)$. 
  Let
  \begin{align*}
    A^j_i:=\{ a^{j'}_i\mid \Con{i}{a^j_i}{a^{j'}_i}\}
  \end{align*}
  be the connected component of $a^j_i$ restricted to those 
  elements that appear as factors of the $a^{j'}$ with $1\leq j,j' \leq k$. 
  Due to upwards linearity, for all $1 \leq i \leq n$ and $1 \leq j
  \leq k$ there is a minimal 
  $m^j_i\in P_i$ such that for all $x\in A^j_i$, $x\leq_i m^j_i$. 
  By definition, we have
  \begin{align}  \label{eqn-Minimalelements}
    m^j_i=m^{j'}_i \Leftrightarrow \Con{i}{a^j_i}{a^{j'}_i}
    \Leftrightarrow
    A^j_i = A^{j'}_i.
  \end{align}
  By definition of a box-augmentation we have
  \begin{align*}
    a^j=\eta(a^j_1, \dots, a^j_n) \leq \eta(m^j_1,a^j_2, \dots, a^j_n)
    \leq \dots \leq \eta(m^j_1, \dots, m^j_n). 
  \end{align*}
  Thus, $a^1 \leq a^j \leq \eta(m^j_1, \dots, m^j_n)$ for all $1\leq j
  \leq k$. Due to upwards linearity, the set
  \begin{equation*}
  	C' := \{\eta(m^j_1, \dots, m^j_n)\mid 1\leq j \leq k\}
  \end{equation*}
  forms a chain. 
  Note that $\eta(m^j_1,\dots, m^j_n) = \eta(m^{j'}_1, \dots,
  m^{j'}_n)$ would imply $m^j_i=m^{j'}_i$, i.e., $\Con{i}{a^j_i}{a^{j'}_i}$
  (by \eqref{eqn-Minimalelements}) for all $1 \leq i \leq n$.
  But this would lead to the contradiction 
  $\ConBox{a^j}{a^{j'}}$. 
  Thus, the chain $C'$ consists of $k$ elements. 
  We conclude by proving that $A_i:=\{m^j_i \mid 1\leq j \leq k\}$ is
  an antichain in $\Pf_i$. Heading for a contradiction assume that 
  there are $j\neq j'$ such that $m^j_i < m^{j'}_i$. 
  Then clearly $\Con{i}{m^j_i}{m^{j'}_i}$ holds and hence also
  $\Con{i}{a^j_i}{a^{j'}_i}$ holds. With \eqref{eqn-Minimalelements}
  we conclude that $m^j_i =  m^{j'}_i$ contradicting our assumption 
  $m^j_i < m^{j'}_i$. 
\end{proof}
The previous result can be used to bound the length of ordinals
occurring in the image of antichains in a tamely colorable
box-augmentation.

\begin{lem}\label{lem:bound-length-chain-wulpo}
  Let $\Pf=(P,\leq)$ be a \wulpo
  which is a tamely colorable box-augmentation of
  $(\Pf_1, \Pf_2, \dots, \Pf_n)$ via
  $\eta$.
  There is a constant $c\in\N$ such
  that the following holds:
  For all choices of antichains $A_i\subseteq P_i$
  (for each $1\leq i \leq n$),  the substructure of $\Pf$ induced by
  $\eta(\prod_{i=1}^{n} A_i)$ does not contain a chain of length $c$.
\end{lem}

\begin{proof}
  For each $1\leq i \leq n$ we fix a finite coloring 
  $\sigma_i:P_i\times  P_i\to C_i$ of $\Pf_i$
  such that the map $\sigma : P\times P \to C$ with
  $C = \prod_{1\leq i\leq n} C_i$ and
  \begin{equation*}
  	\sigma(\eta(p_1,\ldots,p_n),\eta(q_1,\ldots,q_n)) =
		(\sigma_1(p_1,q_1),\ldots,\sigma_n(p_n,q_n)).
  \end{equation*}
  is a finite coloring of $\Pf$.
  
  By Ramsey's theorem \cite{Ramsey30}
  there exists a constant $c\in\N$
  such that every complete simple graph
  with at least $c$ nodes
  whose edges are colored by $|C|$ colors
  contains a monochromatic triangle.

  For the sake of a contradiction,
  assume that for each $1\leq i\leq n$
  there exists an antichain $A_i$ in $\Pf_i$
  such that
  $\eta(\prod_{i=1}^{n} A_i)$
  contains a chain of length $c$.
  Due to the choice of $c$ there exist
  three elements $p^1>p^2>p^3$ in this chain and
  a color $\bar c=(c_1,\ldots,c_n)\in C$
  such that
  \begin{equation*}
  	\sigma(p^1,p^2) = \sigma(p^1,p^3) = \sigma(p^2,p^3) = \bar c.
  \end{equation*}
  Let $p^j=\eta(p^j_1,\ldots,p^j_n)$, where $p^j_i \in A_i$.
  For each $1\leq i\leq n$ we obtain
  \begin{equation*}
  	\sigma_i(p^1_i,p^2_i) = \sigma_i(p^1_i,p^3_i) =
		\sigma_i(p^2_i,p^3_i) = c_i .
  \end{equation*}
  For $p'=\eta(p^2_1,p^1_2,\ldots,p^1_n)$ we get
  $\sigma(p^1,p^3)=\sigma(p',p^3)$.
  Since $p^1>p^3$, this implies $p'>p^3$.
  As the elements above $p^3$ are linearly ordered,
  we conclude that $p'$ and $p^1$ are comparable.
  Recall that $\eta(\cdot, p^1_2, \ldots p^1_n)$
  embeds $\Pf_1$ into $\Pf$.
  We conclude that $p^1_1$ and $p^2_1$ are comparable in $\Pf_1$.
  Since $A_1$ is an antichain and $p^1_1,p^2_1\in A_1$,
  this implies $p^1_1=p^2_1$.
  
  Analogous arguments for the other coordinates show that
  $p^1_i=p^2_i$ for each $1\leq i \leq n$, i.e., $p^1=p^2$.
  However, this contradicts $p^1>p^2$.
\end{proof}
We now head for the following result. Given a \wulpo that is a tamely
colorable box-aug\-men\-ta\-tion we can write it as a finite sum-augmentation
of \wulpo's whose rank is bounded in terms of the ranks of the
connected  components of the factors
of the box.

\begin{lem} \label{lem:mu-function}
  Let $\Pf=(P, \leq)$ be a countable \wulpo which is a tamely colorable box
  augmentation of 
  $(\Pf_1, \dots, \Pf_n)$ via the bijection $\eta$. 
  Let $c\in\N$ such that 
  the image of antichains under $\eta$ does not contain a chain of
  length $c$ (exists by Lemma~\ref{lem:bound-length-chain-wulpo}). 
  Then there  is a map $\mu:P\to \{1, \dots, c\}$ such that
  for all chains $L\subseteq P$, and
  all $p,p'\in L$ we have $\mu(p) = \mu(p')$ if and only if 
  $\ConBox{p'}{p}$.
\end{lem}

\begin{proof}
  Fix an enumeration $m_1, m_2, \dots$ of the minimal elements of $\Pf$. 
  Let $\mu_0$ be the partial function with empty domain. We define
  $\mu$ as the limit of partial functions $\mu_i$ satisfying 
  the lemma (restricted to their domain). This limit $\mu$ is a  total
  function because we guarantee that
  \begin{equation*}
  	\dom(\mu_i)= \bigcup_{j=1}^i [m_j,\infty).
  \end{equation*}
  Assume that $\mu_{j}$ has already been defined.
  
  First we set $\mu_{j+1}(x) = \mu_{j}(x)$ for all $x\in\dom(\mu_{j})$. 
  By transfinite induction we extend 
  $\mu_{j+1}$ to all nodes $y \in [m_{j+1},\infty) \setminus \bigcup_{i=1}^j [m_i,\infty)$. 
  For this purpose assume that 
  $y \in [m_{j+1},\infty) \setminus \bigcup_{i=1}^j [m_i,\infty)$
  and assume that we have
  defined $\mu_{j+1}(x)$ for all $m_{j+1}\leq x < y$.
  Hence, the domain of the current $\mu_{j+1}$ is 
  $$ 
  D = \bigcup_{i=1}^j [m_i,\infty) \cup [m_{j+1},y) . 
  $$
  We assume that for all chains $L \subseteq D$, and
  all $p,p'\in L$ we have $\mu_{j+1}(p) = \mu_{j+1}(p')$ if and only if 
  $\ConBox{p'}{p}$. Let 
  \begin{equation*}
  	J = \{1,\ldots,c\} \setminus \mu_{j+1}(D \cap
        [m_{j+1},\infty)) .
  \end{equation*}
  Then $J \neq \emptyset$:  $D \cap [m_{j+1},\infty)$ is a chain.
  By Lemma~\ref{lemma:at-most-c-classes-on-chain} and the choice of
  $c$, each chain in $\Pf$ contains elements of at most $c-1$ many
  $\ConBoxNull$-classes. Hence, the elements from $D \cap [m_{j+1},\infty)$
  fall into at most $c-1$ many $\ConBoxNull$-classes. By our
  assumption on the current $\mu_{j+1}$, this mapping takes at most $c-1$ different
  values among the elements from $D \cap [m_{j+1},\infty)$. Hence, $J$ is not empty.

  Now we extend $\mu_{j+1}$ to $y$ as follows:
  \begin{align*}
  \mu_{j+1}(y):=
  \begin{cases}
    \mu_{j+1}(p) &\text{if $p \in D \cap [m_{j+1},\infty)$ and $\ConBox{p}{y}$,}  \\
    \min(J) & \text{otherwise.}
  \end{cases}  
  \end{align*}
  Each of the $\mu_j$ is a well-defined
  partial function and setting $\mu=\bigcup_{j\in\N} \mu_j$
  settles the claim. 
\end{proof}

\begin{cor}\label{cor:mu-connected-components}
  Let $\Pf$, $c$, and $\mu$ be defined as in Lemma~\ref{lem:mu-function}. 
  Let $1\leq i\leq c$ and  ${a,b\in \mu^{-1}(i)}$ such that $a$
  and $b$ are in a connected component of the suborder induced by
  ${\mu^{-1}(i)\subseteq \Pf}$. Then we have $\ConBox{a}{b}$. 
\end{cor}

\begin{proof}
  Let $\mu(a) = i = \mu(b)$ and assume that $a$
  and $b$ are in a connected component of the suborder induced by
  ${\mu^{-1}(i)\subseteq \Pf}$. Since this suborder is again a \wulpo,
  there exists $p \in \mu^{-1}(i)$ with $a \leq p$ and $b \leq p$.
  Lemma~\ref{lem:mu-function} together with $\mu(a)=\mu(p)=\mu(b)$
  implies $a \equiv p \equiv b$, i.e., $a \equiv b$ by 
  transitivity of $\ConBox{}{}$. 
\end{proof}

\begin{cor}\label{lem:wulpoboxsplitassums}
  Let $\Pf$ be a \wulpo that is box-augmentation of (\wulpo's)
  $(\Pf_1, \Pf_2, \dots, \Pf_m)$ such that there is some $n\in\N$ with
  $\rank(\Pf_j)<\omega^{\omega^n}$ for all $1\leq j\leq m$. 
  \begin{enumerate}[\em(1)]
  \item   For $\mu$ as in Lemma \ref{lem:mu-function}, 
    $\rank(\Pf{\restriction}_{\mu^{-1}(i)})\leq \omega^{\omega^n}$ for
    all $i$ in the range of $\mu$.
  \item Moreover, if $\Pf$ is a forest and there is $n\in\N$
    with $\rank(\Pf_j)<\omega^n$ for all $1\leq j\leq m$, then
    $\rank(\Pf{\restriction}_{\mu^{-1}(i)})\leq \omega^n$ for
    all $i$ in the range of $\mu$.
  \end{enumerate}
\end{cor}

\begin{proof}
  Let $\Cf_j$ be the set of connected components of $\Pf_j$ for all
  $1\leq j \leq m$. 
  Due to Corollary \ref{cor:mu-connected-components},  $\Pf{\restriction}_{\mu^{-1}(i)}$ is
  isomorphic to some suborder of the disjoint union
  \begin{align*}
    \bigsqcup_{(D_1,\dots, D_m)\in \Cf_1\times\dots\times \Cf_m} 
    \Pf{\restriction}_{\eta(D_1\times \dots \times D_m)}.
  \end{align*}
  Note that every $D_j \in \Cf_j$ is a substructure of $\Pf_j$, hence
  $\rank(D_j)<\omega^{\omega^n}$. Due to Lemma
  \ref{Cor:ProdConnectedWulpos},
  \begin{align*}
    \rank(\Pf{\restriction}_{\eta(D_1 \times \dots \times D_m)})\leq 
    \rank(D_1) \otimes \dots\otimes \rank(D_m)) \overset{\text{\eqref{Eqn-Otimes-Limit}}}{<} \omega^{\omega^n}.    
  \end{align*}
  Thus, 
  \begin{align*}
    \rank\left(\bigsqcup_{(D_1,\dots, D_m)\in \Cf_1\times\dots\times \Cf_m} 
    \Pf{\restriction}_{\eta(D_1\times \dots \times D_m)}\right) \leq
  \omega^{\omega^n}.
  \end{align*}
  Hence, the same holds for the substructure 
  $\Pf{\restriction}_{\mu^{-1}(i)}$. 

  If $\Pf$ is a forest, then without loss of generality
  $\Pf_2, \dots, \Pf_m$ are disjoint unions of finite ordinals (cf.\
  Remark \ref{rem-difference-forest-wulpos}). 
  Thus, 
  $\rank(\Pf{\restriction}_{\eta(D_1 \times \dots \times D_m)})\leq 
  \rank(D_1) \otimes c = \alpha \otimes c$ for some
  finite ordinal 
  $c<\omega$ and some $\alpha<\omega^n$. 
  We get
  $$\alpha \otimes c \leq \underbrace{\alpha \oplus \alpha \oplus \cdots
    \oplus \alpha}_{\text{$c$ many}} \overset{\text{\eqref{Eqn-Oplus-Limit}}}{<} \omega^n ,
  $$
  where the first inequality follows from the fact that every
  linearization of the direct product $\alpha \otimes c$ can be viewed
  as a linearization of the disjoint union of $c$ copies of $\alpha$
  (one can actually show that equality holds at the place of the first inequality).
  We can now conclude completely analogously to the \wulpo case. 
\end{proof}
We are now prepared to establish box-indecomposable intervals with
respect to the rank function on the domain of \wulpo's.
Recall that we denote by $\rank_{ul}$ the function $\rank$ restricted
to \wulpo's and by $\rank_{F}$ its restriction to well-founded forests. 

\begin{prop} \label{prop:rank-tame-box-indecomposable}
  Let $i\in\N$.
  \begin{enumerate}[\em(1)]
    \item \label{item:rank-tame-box-indecomposable-wulpo} The interval
      $[\omega^{\omega^i},\omega^{\omega^{i+1}}]$ is
      $\rank_{ul}$-tamely-colorable-box-in\-de\-com\-pos\-able.
    \item The interval $[\omega^i, \omega^{i+1}]$ is
      $\rank_{F}$-tamely-colorable-box-indecomposable.
  \end{enumerate}
\end{prop}

\begin{proof}
  For the first claim let $\Pf=(P,\leq)$ be a \wulpo of rank
  $\omega^{\omega^{i+1}}$. 
  Heading for a contradiction, 
  assume that it is a tamely colorable
  box-augmentation of \wulpo's $(\Pf_1, \dots, \Pf_n)$ such that
  $\rank(\Pf_j)< \omega^{\omega^i}$ for all $1\leq j \leq n$. 

  Take $c \in \N$ and the mapping $\mu: P \to \{1, \dots, c\}$
  from Lemma~\ref{lem:mu-function}. 
  The preimages $\mu^{-1}(i)$ of $\mu$ induce a sum-decomposition of $\Pf$ into 
  $c$ parts
  $\Pf{\restriction}_{\mu^{-1}(1)}, \dots, \Pf{\restriction}_{\mu^{-1}(c)}$. Since
  $\omega^{\omega^{i+1}}$ is $\rank$-sum-indecomposable by
  Proposition~\ref{prop:rank-sum-indecomposable}, 
  we can assume (perhaps after renaming of
  colors) that
  $\rank(\Pf{\restriction}_{\mu^{-1}(1)}) = \omega^{\omega^{i+1}}$. 
  But Corollary~\ref{lem:wulpoboxsplitassums} implies that 
  $\rank(\Pf{\restriction}_{\mu^{-1}(1)}) \leq \omega^{\omega^i}$ which is clearly a
  contradiction. 

  For the case of forests, we use exactly the same arguments, and the
  second part of Corollary~\ref{lem:wulpoboxsplitassums} yields a
  similar contradiction. 
\end{proof}
Using Propositions~\ref{Prop:Indecomposability},
\ref{prop:rank-sum-indecomposable} and~\ref{prop:rank-tame-box-indecomposable},
we obtain directly the desired bounds on the ranks.

\begin{thm} \label{thm:Forests-Rank}
  \hfill
  \begin{enumerate}[\em(1)]
  \item   Every tree-automatic \wulpo has rank strictly below $\omega^{\omega^\omega}$.
  \item \label{item:Forest-Rank-Forests}  Every tree-automatic
    well-founded forest has rank strictly below 
    $\omega^\omega$. 
  \end{enumerate}
\end{thm}
\begin{proof}
	Concerning the first claim,
	let $\alpha_i=\omega^{\omega^i}$ for each $i\in\N$.
	Due to Propositions~\ref{prop:rank-sum-indecomposable}
	and~\ref{prop:rank-tame-box-indecomposable} 
        (part \eqref{item:rank-tame-box-indecomposable-wulpo}),
	the sequence $\alpha_0<\alpha_1<\alpha_2<\dotsc$ satisfies
	the conditions of Proposition~\ref{Prop:Indecomposability}.
	For the sake of a contradiction to this latter proposition,
	assume that $\Pf$ is a tree-automatic \wulpo of rank
	at least~$\omega^{\omega^\omega}$.
	Let $\varphi$ be the formula ${\varphi(x,y)=x<y}$.
	By well-foundedness, for each $i\in\N$ there is a $p\in P$
	such that $\Pf{\restriction}_{\varphi,p}$
	has rank $\alpha_i$,
	contradicting Proposition~\ref{Prop:Indecomposability}.
	
	For \eqref{item:Forest-Rank-Forests},
        we proceed
	analogously with 
        $\alpha_i=\omega^i$. This time we use
	the second part of
	Proposition~\ref{prop:rank-tame-box-indecomposable}.
\end{proof}
Note that this result has an analogous counterpart in the world of
string-automatic structures: 
Khoussainov and
Minnes \cite{KhMi09} and Delhomm{\'e} \cite{Delhomme04} independently
proved that the ordinal ranks of string-automatic well-founded partial 
orders are the ordinals strictly below $\omega^\omega$ and this bound
is optimal because for each $\alpha<\omega^\omega$ the ordinal
$\alpha$ is string-automatic. 
In contrast, we \cite{KartzowLL12,KartzowLL12b-Version1} proved that
every string-automatic forest has rank strictly below $\omega^2$ and this bound is
also optimal. 
In the next section, we construct tree-automatic well-founded
trees of all ranks strictly below $\omega^\omega$. Moreover, since
ordinals are \wulpo's,  our bound on the ranks of \wulpo's is also
optimal. Thus, tree-automatic forests
realize much smaller ranks than tree-automatic partial orders 
(as it is also the case in the string-automatic setting). 
Note that there is no non-trivial bound on the ranks of arbitrary
well-founded tree-automatic partial orders. But our new result on the
subclass of all \wulpo's and the analogy to the string-automatic case
support the  following conjecture.

\begin{conj}
  Every tree-automatic well-founded partial order has rank below
  $\omega^{\omega^\omega}$. 
\end{conj}
Beside upwards linearity, well-founded forests 
also have the property that one cannot embed an infinite linear
order. One might wonder whether the class of tree-automatic partial
orders without infinite linear suborder already  satisfies that the
ranks are bounded by $\omega^\omega$.
The example below answers this question negatively be showing that for each
$\alpha<\omega^{\omega^\omega}$ there exists such a partial order of rank $\alpha$. 

\begin{exa}\label{exa:BoxDestroysRanksofWfPO}
Let $\alpha<\omega^{\omega^\omega}$ be an ordinal. We show that the
the direct product of $\alpha$ and $\omega^*$ (the reverse of
$\omega$) as \emph{strict} partial orders is a tree-automatic partial
order which has rank $\alpha$ and contains no infinite linear suborder. 

More formally, we consider the partial order $\Pf=(P,\preceq)$ on
\begin{equation*}
	P = \{ (\beta,k) \mid \beta<\alpha, k\in\N \}
\end{equation*}
whose induced strict partial order $\prec$ is given by
\begin{equation*}
	(\beta,k) \prec (\beta',k') \Longleftrightarrow
        \beta<\beta' 
        \text{ and }k>k'.
\end{equation*}
As a direct product of the tree-automatic structures $\alpha$ and $\omega^*$,
$\Pf$ is also tree-automatic.

Note that $\Pf$ does not contain an infinite linear suborder because
any such suborder could be projected to an infinite strictly descending
sequence in $\alpha$ or to an infinite strictly ascending sequence in 
$\omega^*$.  

We prove that $\rank(\Pf)=\rank(\alpha)$ by transfinite induction. 
For $\beta=0$ and $k\in\omega^*$  we have
$\rank((\beta,k),\Pf) = 0$ because $(\beta,k)$ is a minimal element.
For $\beta>0$, $k\in\omega^*$,
\begin{align*}
  \rank((\beta,k),\Pf)
  &= \sup\{ \rank((\beta',k'),\Pf) + 1 \mid (\beta',k') \prec (\beta,k) \in P \} \\
  &= \sup\{ \beta'+1 \mid \beta'<\beta \}
  = \beta .
\end{align*}
Thus, we conclude that
\begin{equation*}
	\rank(\Pf)
	= \sup\{ \rank((\beta,k),\Pf)+1 \mid (\beta,k)\in\Pf \}
	= \sup\{ \beta+1 \mid \beta < \alpha \}
	= \alpha.
\end{equation*}
The structure $\Pf$ also shows that Delhomm\'{e}'s decomposition
technique cannot be used to prove our conjecture that all
tree-automatic partial orders have ranks below
$\omega^{\omega^\omega}$. 
Note that for all $\alpha'\in \alpha$  the
substructure
$\Pf{\restriction}_{\{(\alpha',n)\mid n\in\omega^*\}}$ is isomorphic to
the countably infinite  antichain $(\omega^*, =)$. Similarly, for all
$n\in\omega^*$ the
substructure
$\Pf{\restriction}_{\{(\alpha',n)\mid \alpha'\in\alpha\}}$ is isomorphic to
the countably infinite  antichain $(\alpha, =)$. Thus, $\Pf$ is a
partial order of rank $\alpha$ that is a tamely colorable
box-augmentation of two antichains. Since infinite antichains have
rank $1$, we conclude that  applications of box-decomposition to
well-founded partial orders may destroy all information about the
rank. Thus, there 
is no hope that one could provide bounds on the ranks of
tree-automatic well-founded partial orders using  Delhomm\'{e}'s
decomposition technique.
\end{exa}

\section{Upper bound for the isomorphism
  problem for well-founded  trees} 
\label{sec:OmegaOmegaFormulas}

We first introduce hyperarithmetical sets. Afterwards we provide the
upper bound for the isomorphism problem of well-founded tree-automatic
trees. 

\subsection{Hyperarithmetical sets} \label{sec:hyperarith}

We use standard terminology concerning recursion theory;
see e.g. \cite{Rogers}.
We use the definition of the hyperarithmetical hierarchy from Ash and
Knight \cite{AshKNight00}, see also~\cite{HiWe02}. We first define
inductively a set
of {\em ordinal notations} $O\subseteq \N_{>0}$.
Simultaneously we define a mapping $a \mapsto |a|_O$ from $O$ into
ordinals and a strict partial order $<_O$ on $O$.
The set $O$ is the smallest
subset of $\N_{>0}$ satisfying the following conditions:
\begin{iteMize}{$\bullet$}
\item $1 \in O$ and $|1|_O=0$, i.e., $1$ is a notation for the ordinal
  $0$.
\item If $a \in O$, then also $2^a \in O$. We set $|2^a|_O = |a|_O+1$
and let $b <_O 2^a$ if and only if $b=a$ or $b <_O a$.
\item If $e \in \N$ is such that
  $\Phi_e$ (the $e^{th}$ partial computable
  function) is total, $\Phi_e(n) \in O$ for all $n \in \N$, and
  $\Phi_e(0) <_O \Phi_e(1) <_O  \Phi_e(2) <_O \cdots$, then also
  $3 \cdot 5^e \in O$. We set $|3 \cdot 5^e|_O = \sup\{|\Phi_e(n)|_O
  \mid n \in \N\}$ and let
  $b <_O 3 \cdot 5^e$ if and only if there exists $n \in \N$ with
  $b <_O \Phi_e(n)$.
\end{iteMize}
An ordinal $\alpha$ is {\em computable} if there exists $a \in O$
with $|a|_O = \alpha$. The smallest non-computable ordinal is
the Church-Kleene ordinal $\omega_1^{\mathsf{ck}}$.
If $a \in O$ then the restriction of the partial order $(O, <_O)$ to
$B = \{ b \in O \mid b <_O a\}$ is isomorphic to the ordinal $|a|_O$  \cite[Proposition~4.9]{AshKNight00}.
Moreover, the set $B$ is computably enumerable and an index for $B$
can be computed from $a$ \cite[Proposition~4.10]{AshKNight00}.

Next, we define the
{\em hyperarithmetical hierarchy} based on ordinal notations.
For this we define sets $H(a)$ for each $a \in O$ as follows:
\begin{iteMize}{$\bullet$}
\item $H(1) = \emptyset$,
\item $H(2^b) = H(b)'$ (the Turing jump of $H(b)$; see e.g. \cite{Rogers}),
\item $H(3 \cdot 5^e) = \{ \langle b,n\rangle \mid b <_O 3 \cdot 5^e,
  n \in H(b) \}$; here $\langle\cdot,\cdot\rangle$ denotes a
  computable pairing function.
\end{iteMize}
Spector has shown that $|a|_O = |b|_O$ implies that $H(a)$ and $H(b)$
are Turing equivalent.
The levels of the hyperarithmetical hierarchy can be defined as
follows, where $\alpha$ is a computable ordinal.
\begin{iteMize}{$\bullet$}
\item If $\alpha$ is infinite, then $\Sigma^0_\alpha$ is the set of all subsets $A \subseteq
  \N$ that are recursively enumerable in some
  $H(a)$ with $|a|_O = \alpha$ (by Spector's theorem, the concrete
  choice of $a$ is irrelevant). For $n > 0$ finite, one defines $\Sigma^0_n$
   as the set of all subsets $A \subseteq
  \N$ that are recursively enumerable in 
  $H(a)$ with $|a|_O = n-1$.\footnote{The distinction between finite and infinite ordinals is made in order to have 
  a correspondence between the arithmetical hierarchy and the finite part of the hyperarithmetical hierarchy, see also
  \cite[pp.~74 and 75]{AshKNight00}.}
\item $\Pi^0_\alpha$ is the set of all complements of
  $\Sigma^0_\alpha$-sets.
\item $\Delta^0_\alpha = \Sigma^0_\alpha \cap \Pi^0_\alpha$
\end{iteMize}
A relation $R \subseteq \N^k$ is $X^0_\alpha$
(with $X \in \{\Sigma,\Pi,\Delta\}$) if the set
$\{ \langle x_1, \ldots, x_k \rangle \mid (x_1,\ldots, x_k) \in R\}$
is $X^0_\alpha$, where $\langle \cdot, \ldots, \cdot\rangle$ denotes
a computable encoding of $k$-tuples.

For any two computable ordinals $\alpha,\beta$,
$\alpha<\beta$ implies $\Sigma_\alpha \cup \Pi_\alpha \subsetneq \Delta_\beta$.
The union of all classes
$\Sigma^0_\alpha$ where $\alpha< \omega_1^{\mathsf{ck}}$ yields the
class of all {\em hyperarithmetical sets}. By a
classical result of Kleene, the hyperarithmetical sets are exactly the
sets in $\Delta^1_1 = \Sigma^1_1\cap \Pi^1_1$, where $\Sigma^1_1$ is
the first existential level of the analytical hierarchy, and $\Pi^1_1$
is the set of all complements of $\Sigma^1_1$-sets.

For our purposes, it is convenient to present an alternative definition
of the hyperarithmetical hierarchy using computable infinitary
formulas, see \cite[Chapter~7]{AshKNight00}.
Fix a predicate $R(\overline{x}) \subseteq \N^k$ where $k\geq 1$.  If $R$
is computable, then a $\Sigma^0_0$ (resp. $\Pi^0_0$) {\em index} for
$R$ is a triple $(\Sigma,0, e)$ (resp. $(\Pi,0,e)$) where $e$ is an
index for $R$.  Next, let $0 < \alpha < \omega_1^{\mathsf{ck}}$ be a
computable ordinal. 
  \medskip
  \noindent
  {\em Case 1.} $\alpha=\beta+1$ is a successor ordinal. Then, a $\Sigma^0_\alpha$
  (resp. $\Pi^0_\alpha$) index for $R$ is a triple $(\Sigma, a, e)$
  (resp. $(\Pi, a,e)$) where $|a|_O = \alpha$ and $e$ is
  a $\Pi^0_\beta$ (resp. $\Sigma^0_\beta$) index for a predicate
  $P(\overline{x},y)\subseteq \N^{k+1}$ such that for all $\overline{x}\in
  \N^k$:
  \[
  R(\overline{x}) \Leftrightarrow \exists y: P(\overline{x},y) \qquad
  \Big( \text{resp. } R(\overline{x}) \Leftrightarrow \forall y:
  P(\overline{x},y)\Big).
  \]
  {\em Case 2.} $\alpha$ is a limit ordinal. Then, a $\Sigma^0_\alpha$
  (resp. $\Pi^0_\alpha$) index for $R$ is a triple $(\Sigma, a,e)$
  (resp. $(\Pi,a,e)$), where $|a|_O = \alpha$ and
  $\Phi_e$ is a total computable function such that the following holds:
  For all $n \in \N$, $\Phi_e(n)$ is a $\Pi^0_{\beta_n}$
  (resp. $\Sigma^0_{\beta_n}$) index for some
  predicate $P_n(\overline{x})\subseteq \N^k$,
  $\beta_0 <\beta_1 <\ldots < \alpha$ with $\sup\{\beta_n \mid n \in \N \}
   = \alpha$, and for all $\overline{x}\in \N^k$:
  \[
  R(\overline{x}) \Leftrightarrow \bigvee_{i\in \N} P_i(\overline{x})
  \qquad \Big( \text{resp. } R(\overline{x})\Leftrightarrow
  \bigwedge_{i\in \N} P_i(\overline{x}) \Big).
  \]
Essentially, we can view a $\Sigma^0_\alpha$ (resp. $\Pi^0_\alpha$)
index as a finite representation of a {\em computable infinitary formula} that defines
the corresponding $\Sigma^0_\alpha$ (resp. $\Pi^0_\alpha$) predicate.
For instance, the $\Sigma^0_\alpha$ index $(\Sigma, a,e)$,
where $|a|_O$ is a limit ordinal, represents
the computable infinitary formula $\bigvee_{i\in \N} \varphi_i$, where
$\varphi_i$ is the computable infinitary formula represented by the
index $\Phi_e(i)$. 
In this paper we use the notions of $\Sigma^0_\alpha$
(resp. $\Pi^0_\alpha$) predicates and indices interchangeably.
Formally, we do not allow negation in  computable infinitary formulas.
But if $\varphi(\overline{x})$ defines the $\Sigma^0_\alpha$ (resp. $\Pi^0_\alpha$)
set $A \subseteq \N^k$, then one can construct effectively a
$\Pi^0_\alpha$ (resp. $\Sigma^0_\alpha$) formula for $\N^k
\setminus A$; and we therefore may define this formula as
$\neg\varphi(\overline{x})$, see \cite[Theorem~7.1]{AshKNight00}.

\subsection{Hyperarithmetic Upper Bound}

It turns out that the $\rank$ for well-founded computable
trees yields an upper bound on the recursion-theoretic complexity of
the isomorphism 
problem. Recall that we defined forests as particular partial orders. For
the isomorphism problem, it is useful to assume that also the direct successor
relation is computable.
When speaking of a computable forest in the following theorem, we mean a
forest $\Ff = (F,\leq)$ such that $F$, $\leq$, and the direct successor relation
$E_\Ff$ are all computable.\footnote{On the other hand, if we would omit the
requirement of a computable direct successor relation in
Theorem~\ref{thm-comp-forest}, 
then we would only have to replace the constants in the theorem by
larger values.} 
Note that the direct successor relation of a tree-automatic forest is
even  tree-automatic because it is first-order definable.

\begin{thm} \label{thm-comp-forest}
Let $\alpha$ be a computable ordinal and assume that $\alpha = \lambda+ k$, where
$k \in \mathbb{N}$ and $\lambda$ is not a successor ordinal.
The isomorphism problem for well-founded computable trees of rank at
most $\alpha$  belongs to
level $\Pi^0_{\lambda+2k+1}$ of the hyperarithmetical hierarchy.
\end{thm}

\begin{proof}
Let us fix a well-founded forest $\Ff = (F,\leq)$.
We  define a computable infinitary $\Pi^0_{\lambda+2k+1}$ formula
expressing
$\Ff(x) \cong \Ff(y)$ for nodes $x$ and $y$ of $\Ff$ of rank at most
$\lambda + k$,
where $k \in \mathbb{N}$ and $\lambda = 0$ or $\lambda$ is a limit ordinal.
This suffices because the disjoint union of two computable trees is a
computable forest.

Let $E = E_\Ff$ be the direct successor relation of $\Ff$.
For every ordinal $\alpha$ we define a computable infinitary formula
$\iso_\alpha(x,y)$ over $\Ff$ as follows:
Let
\begin{equation*}
\iso_0(x,y)  =  \forall z ( \neg (x > z)
  \land \neg (y > z)).  
\end{equation*}
This is a $\Pi^0_1$ formula.
For a successor ordinal $\alpha+1$ let\footnote{We use $\exists^{\geq \ell} x : \varphi(x)$
as an abbreviation for $\exists x_1 \cdots \exists x_\ell : 
\bigwedge_{1 \leq i<j \leq n} x_i \neq x_j \wedge \bigwedge_{i=1}^n \varphi(x_i)$.
The quantifier $\exists^{\geq \ell} x$ can be encoded by an ordinary
single existential quantifier.}   $\iso_{\alpha+1}(x,y)$ be
\begin{equation*}
  \forall u \in  E(x) \cup E(y)\;  \forall \ell \geq 1 \left(
\begin{aligned}
  &\exists^{\geq \ell} v \in E(x)\;  \iso_{\alpha}(u,v)
\Longleftrightarrow  \exists^{\geq \ell} v \in E(y)\;
\iso_{\alpha}(u,v)) \\
&\land \iso_\alpha(u,u).
\end{aligned}
\right)
\end{equation*}
Finally, for a limit ordinal $\alpha$, we define  $\iso_{\alpha}(x,y)$ as
\begin{equation*}
  \bigwedge_{\beta<\alpha}
  \forall u \in E(x) \cup E(y)\; \forall \ell \geq 1 \left(
  \begin{aligned}
 &\exists^{\geq \ell} v \in E(x) \; \iso_{\beta}(u,v)
  \Longleftrightarrow  \exists^{\geq \ell} v \in E(y) \;
  \iso_{\beta}(u,v)) \\
  &\wedge \bigvee_{\beta<\alpha}\iso_\beta(u,u).
  \end{aligned}
  \right)
\end{equation*}
Let us argue by induction that $\iso_{\lambda+k}$ (where $k \in \N$ and 
$\lambda=0$ or $\lambda$ a limit ordinal)
is a  $\Pi^0_{\lambda+2k+1}$-formula.
The case $\lambda=k=0$ is clear. 
For the case $k=0$ and $\lambda>0$ note that the subformula
$\forall u \in E(x) \cup E(y) : \bigvee_{\beta<\lambda}\iso_\beta(u,u)$
is $\Pi^0_{\lambda+1}$. Moreover, 
the subformula 
$$
\forall u \in E(x) \cup E(y)\; \forall \ell \geq 1:
(\exists^{\geq \ell} v \in E(x) :
  \iso_{\beta}(u,v)
  \Longleftrightarrow  \exists^{\geq \ell} v \in E(y):
  \iso_{\beta}(u,v))
$$
is $\Pi^0_{\lambda'+2k'+3}$, where $\beta = \lambda'+k'<\alpha$
(with $k' \in \N$ and $\lambda'=0$ or $\lambda'$ a limit ordinal).
Since $\lambda$ is a limit ordinal, we have $\lambda'+2k'+3 <
\lambda$. This shows that $\iso_{\lambda}(x,y)$ is a conjunction
of a  $\Pi^0_{\lambda}$-formula and a $\Pi^0_{\lambda+1}$-formula, 
and hence also $\Pi^0_{\lambda+1}$.
Finally, for the case that $k>0$ and $\lambda>0$ note that
if $\iso_{\lambda+k}(x,y)$ is a \mbox{$\Pi^0_{\lambda+2k+1}$-formula},
then $\iso_{\lambda+k+1}(x,y)$ is a \mbox{$\Pi^0_{\lambda+2k+3}$-formula},
since two alternating quantifier blocks are added for successor ordinals.

It remains to show that if $\rank(x) \leq \alpha$ and
${\rank(y) \leq \alpha}$, then $\Ff
\models \iso_\alpha(x,y)$ if and only if $\Ff(x) \cong \Ff(y)$. This
is the content of Claim 2 below.
As an auxiliary step we show the following claim.

\medskip
\noindent
{\em Claim 1.} If
$\Ff \models \iso_\alpha(x,y)$ then
$\rank(x) \leq \alpha$ and  $\rank(y) \leq \alpha$.

\medskip
\noindent
We prove Claim 1 by induction on $\alpha$.
The case $\alpha=0$ is clear because trees of rank $0$ do only consist
of one node.
Next, consider an ordinal
$\alpha>0$ such that the claim holds for all $\beta<\alpha$.
Assume that
$\Ff \models \iso_{\alpha}(x,y)$.
Note that the second conjunct of $\iso_{\alpha}(x,y)$ and the induction
hypothesis imply that all children of $x$
have rank $<\alpha$.
Thus, $\rank(x) \leq \sup\{\beta+1\mid \beta<\alpha\} = \alpha$.
For $y$ we conclude analogously.

\medskip
\noindent
{\em Claim 2.} If $\rank(x) \leq \alpha$ and $\rank(y) \leq \alpha$,
then
$\Ff \models \iso_\alpha(x,y)$ if and only if $\Ff(x) \cong \Ff(y)$.

\medskip
\noindent
We prove Claim 2  by induction on $\alpha$.
Again, the case $\alpha=0$ is clear.
Next, consider some  ordinal
$\alpha>0$ and assume that $\rank(x) \leq \alpha$ and $\rank(y) \leq
\alpha$.

First, assume that $\Ff(x) \cong \Ff(y)$.
Fix $\beta<\alpha$, $u\in E(x)\cup E(y)$, and $\ell \geq 1$.
By assumption, we have $\rank(u) < \alpha$.
Furthermore, due to Claim 1, if $\beta<\rank(u)$, then there is no $v$
such that $\iso_\beta(u,v)$ holds. Thus, in this case
\begin{align*}
  \left(\exists^{\geq \ell} v \in E(x): \iso_{\beta}(u,v)
  \ \ \Longleftrightarrow \ \
  \exists^{\geq \ell} v \in E(y) : \iso_{\beta}(u,v) \right) .
  \end{align*}
 holds.
Moreover, if $\rank(u) \leq \beta < \alpha$, then Claim 1 and
the induction hypothesis imply that
\begin{align*}
  \exists^{\geq \ell} v \in E(x): \iso_{\beta}(u,v)
  \ \ \Longleftrightarrow \ \
  \exists^{\geq \ell} v \in E(y) : \iso_{\beta}(u,v)
\end{align*}
is equivalent to
\begin{align*}
  \exists^{\geq \ell} v \in E(x):\Ff(u) \cong \Ff(v)
  \Longleftrightarrow  \exists^{\geq \ell} v \in E(y)
  :\Ff(u) \cong \Ff(v)  .
\end{align*}
Since $\Ff(x) \cong \Ff(y)$ the latter is clearly satisfied.
Moreover, by induction hypothesis for each ${u\in E(x)\cup E(y)}$ we have
$\iso_{\rank(u)}(u,u)$ (note that $\rank(u)<\alpha$). Hence
$\bigvee_{\beta<\alpha} \iso_\beta(u,u)$ holds as well. Thus
$\iso_\alpha(x,y)$ is satisfied.

For the other direction, assume that $\Ff \models \iso_\alpha(x,y)$.
Then for each \mbox{$u\in E(x) \cup E(y)$}, there is some $\beta<\alpha$ such
that $\iso_\beta(u,u)$ holds. Due to Claim 1 this implies that
\mbox{$\rank(u) < \alpha$} for each child $u$ of $x$ or $y$.
Using the induction hypothesis and Claim $1$, we conclude that for 
all \mbox{$u\in E(x) \cup E(y)$} and all ordinals 
$\rank(u)\leq\beta < \alpha$,
\begin{equation} \label{eq:exists-at-least-rank}
  \exists^{\geq \ell} v \in E(x): \iso_{\beta}(u,v)
  \ \ \Longleftrightarrow \ \
  \exists^{\geq \ell} v \in E(y) : \iso_{\beta}(u,v)
\end{equation}
is equivalent to
\begin{equation} \label{eq:exists-at-least-rank'}
  \exists^{\geq \ell} v \in E(x):\Ff(u) \cong \Ff(v)
  \Longleftrightarrow  \exists^{\geq \ell} v \in E(y)
  :\Ff(u) \cong \Ff(v)  .
\end{equation}
Since $\iso_\alpha(x,y)$ holds,  for each $u\in E(x) \cup E(y)$ there
is a $\rank(u)\leq \beta< \alpha$ such that
\eqref{eq:exists-at-least-rank} holds. Hence, \eqref{eq:exists-at-least-rank'} holds
 for all children $u$ of $x$ or $y$, which implies that
$\Ff(x) \cong \Ff(y)$.  
\end{proof}

\begin{rem} \label{remark-Pi}
  If we are only interested in the isomorphism problem for computable trees of
  rank at most $\lambda$ for a limit ordinal $\lambda$ it suffices to consider
  the $\Pi^0_\lambda$-formula
  \begin{equation*}
      \bigwedge_{\beta<\alpha}
    \forall u \in E(x) \cup E(y)\; \forall \ell \geq 1: 
    (\exists^{\geq \ell} v \in E(x) :
    \iso_{\beta}(u,v)
    \Longleftrightarrow  \exists^{\geq \ell} v \in E(y):
    \iso_{\beta}(u,v)) 
  \end{equation*}
  (where the formulas $\iso_{\beta}$ are constructed as above).
  This is because if the rank of the root is bounded
  by $\lambda$, any child of the root has rank strictly below
  $\lambda$ and satisfies $\iso_\beta(u,u)$ for some $\beta<\lambda$.
\end{rem}

\begin{cor}
The isomorphism problem for well-founded tree-automatic trees belongs
to $\Delta^0_{\omega^\omega}$.
\end{cor}

\begin{proof}
The $\Sigma^0_{\omega^\omega}$ formula
$\bigvee_{\alpha<\omega^\omega}\iso_{\alpha}(x,y)$
expresses $\Ff(x) \cong \Ff(y)$ for all $x,y$ with $\rank(x),\allowbreak
\rank(y) < \omega^\omega$.

With Remark~\ref{remark-Pi}, we also obtain a $\Pi^0_{\omega^\omega}$ 
formula expressing
$\Ff(x) \cong \Ff(y)$ for all $x,y$ with $\rank(x),
\rank(y) < \omega^\omega$.

Due to Theorem \ref{thm:Forests-Rank}, the $\rank$ of every
tree-automatic well-founded tree is strictly
below $\omega^\omega$. Thus, the isomorphism problem for these trees
belongs to ${\Sigma^0_{\omega^\omega} \cap \Pi^0_{\omega^\omega} =
\Delta^0_{\omega^\omega}}$. 
\end{proof}

\section{Lower bound for the isomorphism problem for well-founded trees}

In this section, we prove hyperarithmetical lower bounds for the isomorphism
problem for well-founded tree-automatic trees. More precisely, we show
that for every ordinal $\omega^i$, there exists a well-founded tree
$\Vf$ such that the set of all tree-automatic copies of
$\Vf$ is hard 
for the class $\Pi^0_{\omega^i}$. Moreover, for the class of all  
well-founded  tree-automatic trees, we prove that the isomorphism
problem is $\Delta^0_{\omega^\omega}$-hard under Turing-reductions.

\subsection{Isomorphism for computable trees of rank $< \omega^\omega$}

Basically, our hardness proof is a reduction from computable
well-founded trees to tree-automatic well-founded trees.
For this, we make use of a construction from \cite{HiWe02}, which
works for all
computable ordinals. We use this construction only for ordinals
strictly below $\omega^\omega$. In this section,
a computable tree is a computable prefix-closed subset
$\Sf \subseteq \mathbb{N}_{>0}^*$.\footnote{For
technical reasons, it is useful to exclude $0$.
This makes the construction of tree-automatic copies easier.}
We identify $\Sf$ with the tree $(\Sf,\succeq)$. Recall that 
$\succeq$ is the inverse prefix relation. Hence, the empty 
word $\varepsilon$ is the root of $(\Sf,\succeq)$.

First, we have to fix the so called
fundamental sequence for every limit ordinal $< \omega^\omega$. Each
ordinal $\alpha < \omega^\omega$ can be written in its
Cantor normal form as
$$
\alpha = \omega^{e_i} \cdot n_i + \omega^{e_{i-1}} \cdot n_{i-1} + \cdots + \omega^{e_1} \cdot n_1,
$$
where $e_i  > e_{i-1} > \cdots > e_1 \geq 0$ 
and $n_j > 0$ are natural numbers for $1 \leq j \leq i$.
Assume that $e_1 > 0$ so that $\alpha$ is a limit ordinal.
Then
$$
\alpha = \sup \{\alpha_k \mid k \geq 1 \},
$$
where we define $\alpha_k$ as
\begin{equation}\label{def-alpha_k}
\alpha_k = \omega^{e_i} \cdot n_i + \omega^{e_{i-1}} \cdot n_{i-1} +
\cdots + \omega^{e_2} \cdot n_2 +
\omega^{e_1} \cdot (n_1-1) + \omega^{e_1-1} \cdot k + 1,
\end{equation}
for $k \geq 1$. Note that $\alpha_k$ is a successor ordinal and we
call $(\alpha_k)_{k\geq 1}$ the fundamental sequence. 

Next, we define for every ordinal $\alpha < \omega^\omega$ a
computable well-founded tree
$\Sf_\alpha \subseteq \mathbb{N}_{>0}^*$ by induction on $\alpha$.
Let $\Sf_0 = \{\varepsilon\}$ be the tree consisting of a single node.
If $\alpha = \beta+1 < \omega^\omega$ is a successor ordinal then
$$
\Sf_{\beta+1} = \{n u \mid n \in \mathbb{N}_{>0}, u \in \Sf_\beta\} \cup \{\varepsilon\} .
$$
Hence, $\Sf_{\beta+1}$ consists of $\aleph_0$ many copies of
$\Sf_\beta$ together with a new root.
Finally let $\alpha < \omega^\omega$ be a limit ordinal with the
fundamental sequence $(\alpha_k)_{k \geq 1}$ defined in
\eqref{def-alpha_k}. Then
$$
\Sf_\alpha = \{ k u \mid k \in \mathbb{N}_{>0}, u \in \Sf_{\alpha_k}\}  \cup \{\varepsilon\}.
$$
Thus, $\Sf_\alpha$ consists of all the trees $\Sf_{\alpha_k}$ ($k \in
\mathbb{N}_{>0}$) together with a new root.
By induction on $\alpha < \omega^\omega$, it is straightforward to
show that $\Sf_\alpha$ is well-founded and computable. Moreover,
also the set of leaves $\leaves(\Sf_\alpha) \subseteq \mathbb{N}_{>0}^*$
is computable.
Let $A \subseteq \leaves(\Sf_\alpha)$.
Then we denote the structure consisting
of the tree $\Sf_\alpha$ together with the additional unary predicate
$A$ by $(\Sf_\alpha, A)$.

The following result is implicitly shown in
\cite[Proposition~3.2]{HiWe02}, where it is stated
for all computable ordinals. But since we only defined fixed
fundamental sequences for limit ordinals below $\omega^\omega$, we
restrict to ordinals below $\omega^\omega$.

\begin{thm} \label{thm-hirschfeld-white}
Given a limit ordinal $\alpha < \omega^\omega$ and $k \in
\mathbb{N}_{>0} \cup \{\infty\}$ 
(resp., a successor ordinal ${\alpha < \omega^\omega}$),
one can compute indices of a computable
subset $L^\alpha_k \subseteq \leaves(\Sf_\alpha)$
(resp., computable subsets ${A_\alpha, E_\alpha \subseteq \leaves(\Sf_\alpha)}$) such that
the following holds: 

From a natural number $n \in \mathbb{N}$ and a
$\Pi^0_\alpha$ index $(\Pi,a,e)$ for a set $P \subseteq \mathbb{N}$
one can compute an index for a  computable
subset $T_{P,n} \subseteq \leaves(\Sf_\alpha)$ such that the following holds:
\begin{iteMize}{$\bullet$}
\item $(\Sf_\alpha, A_\alpha) \not\cong (\Sf_\alpha, E_\alpha)$ and
  $(\Sf_\alpha, L^\alpha_\infty) \not\cong
  (\Sf_\alpha, L^\alpha_k)$ for all $k\in\N_{>0}$, respectively.
\item
  If $\alpha$ is a successor ordinal, then
  \begin{align*} 
    (\Sf_\alpha, T_{P,n}) \cong \begin{cases}
      (\Sf_\alpha, A_\alpha) & \text{ if } n \in P \\
      (\Sf_\alpha, E_\alpha) & \text{ if } n \notin P .
    \end{cases}
  \end{align*}
\item
  If $\alpha$ is a limit ordinal, then
  \begin{align*}
    (\Sf_\alpha, T_{P,n}) \cong \begin{cases}
      (\Sf_\alpha, L^\alpha_\infty) & \text{ if } n \in P \\
      (\Sf_\alpha, L^\alpha_k) \text{ for some $k \in
        \mathbb{N}_{>0}$} & \text{ if } n \notin P.  
    \end{cases}
  \end{align*}
\end{iteMize}
\end{thm}

\begin{proof}
In Proposition~3.2 from \cite{HiWe02}, the computable subset $A_\alpha \subseteq \leaves(\Sf_\alpha)$
is replaced by a computable tree ${\A}_\alpha \subseteq \Sf_\alpha$ 
(and similarly for $E_\alpha, L^\alpha_k, T_{P,n}$).
For our purpose it is more convenient
to work with computable subsets of the leaves of the fixed tree
$\Sf_\alpha$. Nevertheless, the construction works analogously to the proof
of Proposition~3.2 from  \cite{HiWe02}, except for the induction base
$\alpha=1$. In \cite{HiWe02}, the tree ${\A}_1$ consists of a single
node and the tree ${\E}_1$ consists of a root with infinitely many children (i.e., ${\E}_1 \cong \Sf_1$).
For the construction of $T_{P,n}$, one chooses a computable
set $Q \subseteq \N\times\N$ such that
$P' = \{ n \mid \exists x  Q(x,n) \}$ is  the complement of the
$\Pi^0_1$-set $P$. Then the tree ${\T}_{P,n}$ consists of the root and  $x \in \mathbb{N}_{>0}$
is a child of the root if and only if there exists $y \leq x$ satisfying $Q(y,n)$.
Hence, ${\T}_{P,n} \cong {\E}_1$ if there exists $x$ with $Q(x,n)$ and
${\T}_{P,n} \cong {\A}_1$ if $\neg Q(x,n)$ for all $x$.

In our context, we define the subsets 
$A_1, E_1, T_{P,n} \subseteq \mathbb{N}_{>0} = \leaves(\Sf_1)$ as follows:
let ${A_1 = \emptyset}$, $E_1$ be the set of all non-zero even numbers, and
$T_{P,n} = \{ 2x \mid \exists y \leq x : Q(y,n) \}$.
Then, we have $(\Sf_1, T_{P,n}) \cong (\Sf_1, E_1)$ if there exists $x$ with $Q(x,n)$ and
$(\Sf_1, T_{P,n}) \cong (\Sf_1, A_1)$ if $\neg Q(x,n)$ for all $x$. 

Let us now discuss the induction step. As remarked above, it works analogously to the proof
of Proposition~3.2 from  \cite{HiWe02}, but we present the details for
the sake of completeness.  
Let  $\alpha < \omega^\omega$, $n \in \mathbb{N}$, and $(\Pi,a,e)$ be a
$\Pi^0_\alpha$ index  for a set $P \subseteq \mathbb{N}$. Hence, $|a|_O=\alpha$.
We distinguish three cases:
\begin{figure}[t]
    \centering
      \begin{tikzpicture}
     [ level distance=2.5 cm,
    sibling distance=1.8cm,
    node distance = 0.5cm,
    grow = down,
    child anchor = north,
    konf/.style={
      circle, draw,},
    tree/.style={isosceles triangle, shape border rotate= 90, font=\small,
      draw},
    edge from parent/.style={ ->,draw, pos=0.65, font=\scriptsize},
    edgelabel/.style={ above=-1mm,sloped},
    ]

    \begin{scope}
    \node[] (AAlpha) {$A_\alpha$};
    \node [konf, below of= AAlpha] (root) {}
    child { 
      node[tree] {$E_\beta$}}
    child { 
      node[tree] {$E_\beta$}}
    child { 
      node[tree] {$E_\beta$}}
      child[]{ node(dots){\dots}  edge from parent[draw=none]}
    ;
    \node[above right of= dots]  {$=$};
    \node[konf, node distance=7cm, right of=root, ]{}
    child { 
      node[tree] {$E_\beta$} edge from parent node[edgelabel]{$\aleph_0$}};
    \end{scope}
    \begin{scope}[yshift=-5cm]
    \node[] (EAlpha) {$E_\alpha$};
    \node [konf, below of= EAlpha] (Eroot) {}
    child { 
      node[tree] {$A_\beta$}}
    child { 
      node[tree] {$E_\beta$}}
    child { 
      node[tree] {$A_\beta$}}
    child { 
      node[tree] {$E_\beta$}}
      child[]{ node(Edots){\dots}  edge from parent[draw=none]}

    ;
    \node[above right of= Edots]  {$=$};
    \node[konf, node distance=7cm, right of=Eroot, ]{}
    child { 
      node[tree] {$E_\beta$} edge from parent node[edgelabel]{$\aleph_0$}}
    child { 
      node[tree] {$A_\beta$} edge from parent node[edgelabel]{$\aleph_0$}};
    \end{scope}
    \begin{scope}[yshift=-10cm, sibling distance = 2.2cm]
    \node[] (TPN) {$T_{P,n}$};
    \node [konf, below of= TPN,] (Troot) {}
    [level distance = 2.7cm] child{ 
      node[tree] {$E_\beta$} edge from parent
      node[edgelabel]{$\aleph_0$}}
    child[missing]{}
    child { 
      node[tree] {$T_{\bar Q, n, 0}$}edge from parent
      node[edgelabel]{$\aleph_0$}} 
    child { 
      node[tree] {$T_{\bar Q, n, 1}$}
      edge from parent node[edgelabel]{$\aleph_0$}}
    child { 
      node[tree] {$T_{\bar Q, n, 2}$}
      edge from parent node[edgelabel]{$\aleph_0$}}
    child[]{ node {\dots} edge from parent[draw=none]}
    ;
    \end{scope}

  \end{tikzpicture}

  \caption{Trees for Case 1 from the proof of Theorem~\ref{thm-hirschfeld-white} (here and in the following pictures we
    write an $\aleph_0$ labeled edge 
    from a node $d$ to a
    subtree $T$ if $d$ has edges to countably infinitely many subtrees
    isomorphic to $T$).}    
  \label{fig:Thm5.1Case1}
\end{figure}
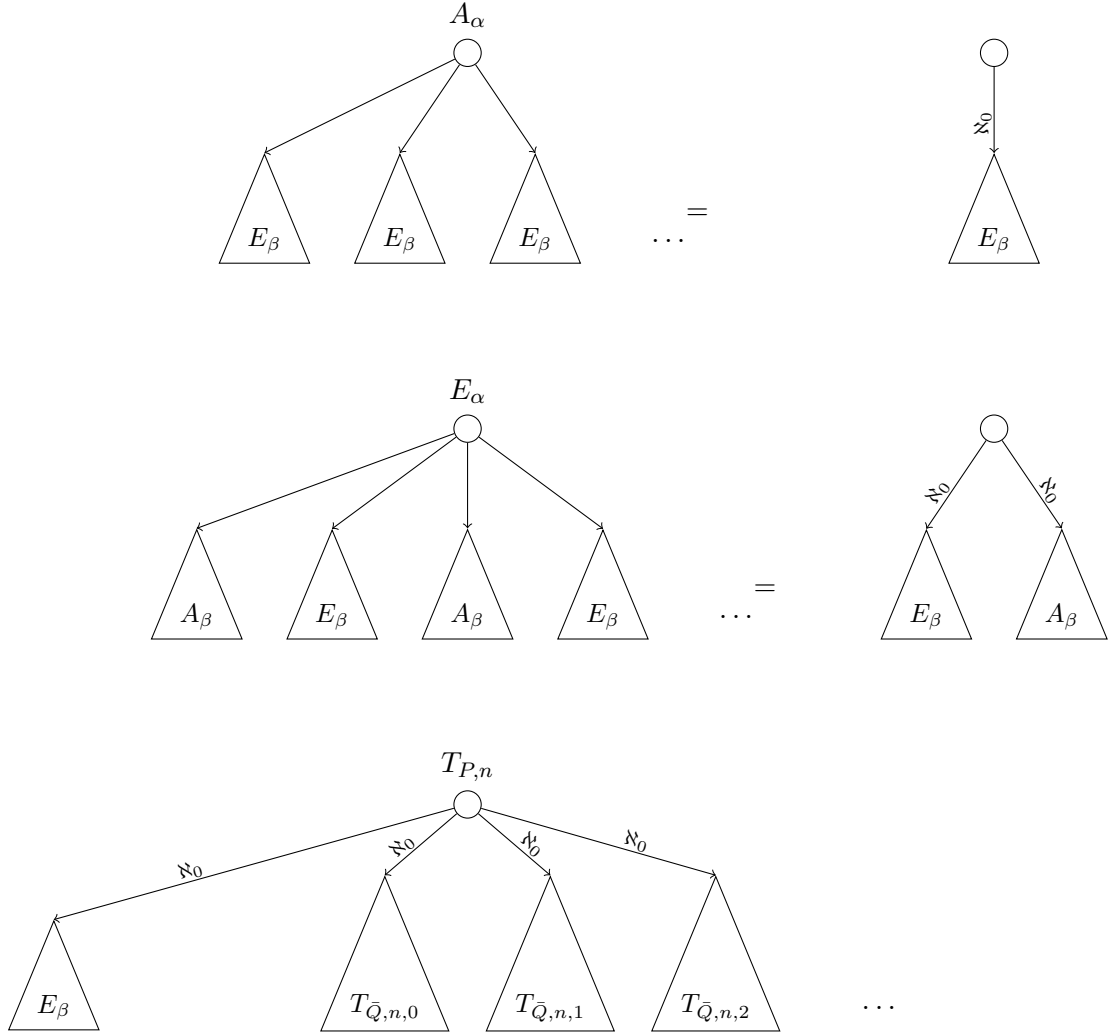

\medskip
\noindent
{\em Case 1.}  $\alpha = \beta+1$ for a successor ordinal $\beta < \omega^\omega$.
Note that we can compute indices for the sets 
\mbox{$A_\beta, E_\beta \subseteq \leaves(\Sf_\beta)$} from $\beta$. The following
constructions are illustrated in Figure \ref{fig:Thm5.1Case1}.
We define the sets $A_\alpha, E_\alpha \subseteq \leaves(\Sf_\alpha) = 
\{m u \mid m \in \mathbb{N}_{>0}, u \in \leaves(\Sf_\beta)\}$
as follows (here and in the following, we use square brackets to enclose single numbers in finite words over $\mathbb{N}$):
\begin{align*}
 &A_\alpha  =  \{m u \mid n \in \mathbb{N}_{>0}, u \in E_\beta \}, \\
 &E_\alpha  =  \{ [2m] u \mid n \in \mathbb{N}_{>0}, u \in E_\beta \}  \cup \{ [2m-1] u \mid m \in \mathbb{N}_{>0}, u \in A_\beta \}.
\end{align*}
Clearly, indices for these computable sets can be computed from the indices for $A_\beta$ and $E_\beta$.
Since $(\Sf_\beta, E_\beta) \not\cong (\Sf_\beta, A_\beta)$, we also have
$(\Sf_\alpha, E_\alpha) \not\cong (\Sf_\alpha, A_\alpha)$

The index $e$ from $(\Pi,a,e)$ is a $\Sigma^0_\beta$ index  for a set $Q \subseteq \mathbb{N} \times \mathbb{N}$ such that
$n \in P$ if and only if $(n,m) \in Q$ for all $m \in \mathbb{N}$. 
From $e$ one can compute a $\Pi^0_\beta$ index $f$ for the complement of $Q$ \cite[Theorem~7.1]{AshKNight00}.
By induction, one can compute from $f$ and $(n,m) \in \mathbb{N} \times \mathbb{N}$
an index of a computable set $T_{\overline{Q},n,m} \subseteq \leaves(\Sf_\beta)$ such that:\footnote{Formally, we have
to encode the pair $(n,m)$ into a single number in order to apply the induction hypothesis.}
\begin{align*} 
    (\Sf_\beta, T_{\overline{Q},n,m}) \cong \begin{cases}
      (\Sf_\beta, A_\beta) & \text{ if } (n,m) \notin Q, \\
      (\Sf_\beta, E_\beta) & \text{ if } (n,m) \in Q .
    \end{cases}
\end{align*}
We define the computable set  $T_{P,n}\subseteq \leaves(\Sf_\alpha)$
as follows
(where we write $2^{\mathbb{N}} 3^{\mathbb{N}}$  for 
\mbox{$\{n\in\N\mid \exists k,l\in\N\ n=2^k\cdot 3^l\}$)}:
\begin{align*}
T_{P,n} = \{ [2^m 3^x] u \mid x,m \in \mathbb{N}, u \in T_{\overline{Q},n,m} \} \cup 
                  \{ x u \mid  x \in \mathbb{N}_{>0} \setminus 2^{\mathbb{N}} 3^{\mathbb{N}}, u \in E_\beta \} . 
\end{align*}
Note that the leaf-labeled tree $(\Sf_\alpha, T_{P,n})$ consists of a root together with 
infinitely many copies of each of the trees $(\Sf_\beta, T_{\overline{Q},n,m})$ for $m \in \mathbb{N}$
together with infinitely many copies of the tree $(\Sf_\beta, E_\beta)$. 
Also note that each of the subtrees $(\Sf_\beta, T_{\overline{Q},n,m})$ is either isomorphic
to $(\Sf_\beta, A_\beta)$ or $(\Sf_\beta, E_\beta)$.

If there is an $m$ with $(n,m) \not\in Q$, then $(\Sf_\alpha, T_{P,n})$ consists of infinitely many copies of 
$(\Sf_\beta, A_\beta)$ and  infinitely many copies of  $(\Sf_\beta, E_\beta)$. Hence, we have 
$(\Sf_\alpha, T_{P,n}) \cong (\Sf_\alpha, E_\alpha)$. On the other hand, if $(n,m) \in Q$ for all $m \in \mathbb{N}$
then $(\Sf_\alpha, T_{P,n}) \cong (\Sf_\alpha, A_\alpha)$.

Finally, an index for the computable set $T_{P,n}$ can be computed from $(\Pi,a,e)$ and $n$.

\medskip
\noindent
{\em Case 2.}  $\alpha = \beta+1$ for a limit ordinal $\beta < \omega^\omega$.
Hence, from a given $k \in
\mathbb{N}_{>0} \cup \{\infty\}$  one can compute and index of a computable
subset $L^\alpha_k \subseteq \leaves(\Sf_\beta)$. We define the 
sets $A_\alpha, E_\alpha \subseteq \leaves(\Sf_\alpha) = 
\{m u \mid m \in \mathbb{N}_{>0}, u \in \leaves(\Sf_\beta)\}$
as follows (cf.~Figure \ref{fig:Thm5.1Case2}):
\begin{figure}
    \centering
      \begin{tikzpicture}
     [ level distance=2.5 cm,
    sibling distance=1.8cm,
    node distance = 0.5cm,
    grow = down,
    child anchor = north,
    konf/.style={
      circle, draw,},
    tree/.style={isosceles triangle, shape border rotate= 90, font=\small,
      draw},
    edge from parent/.style={ ->,draw, pos=0.7, font=\scriptsize},
    edgelabel/.style={ above=-1mm,sloped},
    ]
    \begin{scope}
      \node[] (AAlpha) {$A_\alpha$};
      \node [konf, below of= AAlpha] (root) {}
      child { 
        node[tree] {$L^\beta_1$} edge from parent node[edgelabel]{$\aleph_0$}}
      child { 
        node[tree] {$L^\beta_2$} edge from parent node[edgelabel]{$\aleph_0$}}
      child { 
        node[tree] {$L^\beta_3$} edge from parent node[edgelabel]{$\aleph_0$}}
      child[]{ node[left=0.3cm]{\dots} edge from parent[draw=none] }
      ;
      
      \node[node distance = 7.8cm, right of=AAlpha] (EAlpha) {$E_\alpha$};
      \node [konf, below of= EAlpha] (Eroot) {}
      child { 
        node[tree] {$L^\beta_\infty$} edge from parent
        node[edgelabel]{$\aleph_0$}} 
      child { 
        node[tree] {$L^\beta_1$} edge from parent node[edgelabel]{$\aleph_0$}}
      child { 
        node[tree] {$L^\beta_2$} edge from parent node[edgelabel]{$\aleph_0$}}
      child { 
        node[tree] {$L^\beta_3$} edge from parent node[edgelabel]{$\aleph_0$}}
      child[]{ node[left=0.3cm]{\dots} edge from parent[draw=none]}
      ;
    \end{scope}
    \begin{scope}[yshift = -5cm, xshift=4cm, sibling distance= 2cm,
      level distance =3cm]
      \node (TPN) {$T_{P,n}$};
      \node [konf, below of= TPN,] (Troot) {}
      child{ node{\dots} edge from parent[draw=none]}
      child{ 
        node[tree] {$L^\beta_3$} edge from parent node[edgelabel]{$\aleph_0$}}
      child{ 
        node[tree] {$L^\beta_2$} edge from parent node[edgelabel]{$\aleph_0$}}
      child{ 
      node[tree] {$L^\beta_1$} edge from parent node[edgelabel]{$\aleph_0$}}
    child { 
      node[tree] {$T_{\bar Q, n,0}$}edge from parent
      node[edgelabel]{$\aleph_0$}} 
    child { 
      node[tree] {$T_{\bar Q, n, 1}$}
      edge from parent node[edgelabel]{$\aleph_0$}}
    child { 
      node[tree] {$T_{\bar Q, n, 2}$}
      edge from parent node[edgelabel]{$\aleph_0$}}
    child{ node{\dots} edge from parent[draw=none]}
    ;
    \end{scope}

  \end{tikzpicture}

  \caption{Trees for Case 2 from the proof of Theorem~\ref{thm-hirschfeld-white}.}    
\label{fig:Thm5.1Case2}
\end{figure}
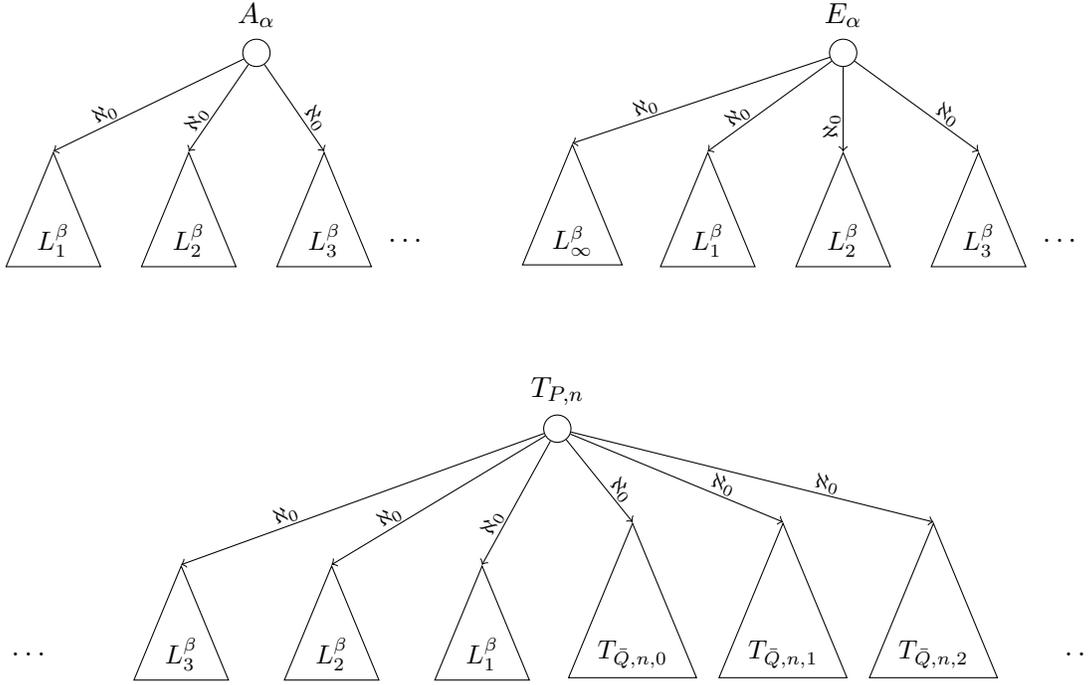

\begin{align*}
 &A_\alpha  =  \{ [2^{k-1} (2x-1)] u \mid k,x \in \mathbb{N}_{>0},  u \in L^\beta_k \}, \\
 &E_\alpha  =  \{ [2^k (2x-1)] u \mid k,x \in \mathbb{N}_{>0},  u \in L^\beta_k \}  \cup 
                             \{ [2x-1] u \mid x \in \mathbb{N}_{>0},  u \in L^\beta_\infty \}.
\end{align*}
Again, indices for these computable sets can be computed from $\alpha$.
Since $(\Sf_\beta, L^\beta_k) \not\cong (\Sf_\beta, L^\beta_\infty)$ for all $k \in \mathbb{N}_{>0}$, we have
$(\Sf_\alpha, E_\alpha) \not\cong (\Sf_\alpha, A_\alpha)$.

As in Case 1, the index $e$ from $(\Pi,a,e)$ is a $\Sigma^0_\beta$ index  for a set $Q \subseteq \mathbb{N} \times \mathbb{N}$ such that
$n \in P$ if and only if $(n,m) \in Q$ for all $m \in \mathbb{N}$. 
From $e$ one can compute a $\Pi^0_\beta$ index $f$ for the complement of $Q$.
By induction, one can compute from $f$ and $(n,m) \in \mathbb{N} \times \mathbb{N}$
an index of a computable set $T_{\overline{Q},n,m} \subseteq \leaves(\Sf_\beta)$ such that
\begin{align*} 
    (\Sf_\beta, T_{\overline{Q},n,m}) \cong \begin{cases}
      (\Sf_\beta, L^\beta_\infty) & \text{ if } (n,m) \notin Q, \\
      (\Sf_\beta, L^\beta_k) \text{ for some $k \in
        \mathbb{N}_{>0}$} & \text{ if } (n,m) \in Q .
    \end{cases}
\end{align*}
We define the computable set  $T_{P,n}\subseteq \leaves(\Sf_\alpha)$ as
the set
\begin{equation*}
   \{ [2^{2m} (2x+1)]u \mid m,x \in \mathbb{N},  u \in
   T_{\overline{Q},n,m} \}
   \cup \{  [2^{2k-1} (2x+1)] u \mid k \in \mathbb{N}_{>0},
   x\in\mathbb{N},  u \in L^\beta_k \}.       
\end{equation*}

The leaf-labeled tree $(\Sf_\alpha, T_{P,n})$ consists of a root together with 
infinitely many copies of each of the trees $(\Sf_\beta, T_{\overline{Q},n,m})$ for $m \in \mathbb{N}$
together with infinitely many copies of each of the trees $(\Sf_\beta, L^\beta_k)$ for $k \in \mathbb{N}_{>0}$. 
Also note that each of the subtrees $(\Sf_\beta, T_{\overline{Q},n,m})$ is isomorphic
to one of the trees $(\Sf_\beta, L^\beta_k)$ ($k \in \mathbb{N}_{>0} \cup \{\infty\}$).

If there is an $m$ with $(n,m) \not\in Q$, then $(\Sf_\alpha, T_{P,n})$ consists of infinitely many copies of 
each of the trees $(\Sf_\beta, L^\beta_k)$  for all $k \in \mathbb{N}_{>0} \cup \{\infty\}$.
Hence, we have 
$(\Sf_\alpha, T_{P,n}) \cong (\Sf_\alpha, E_\alpha)$. On the other hand, if $(n,m) \in Q$ for all $m \in \mathbb{N}$
then $(\Sf_\alpha, T_{P,n}) \cong (\Sf_\alpha, A_\alpha)$.

\medskip
\noindent
{\em Case 3.}  $\alpha$ is a limit ordinal. Let $(\alpha_k)_{k \geq 1}$
be the fundamental sequence of $\alpha$. Recall that $\alpha_k$ is a
successor ordinal.
\begin{figure}
    \centering
      \begin{tikzpicture}
    [ level distance=2.5 cm,
    sibling distance=1.8cm,
    node distance = 0.5cm,
    grow = down,
    child anchor = north,
    konf/.style={
      circle, draw,},
    tree/.style={isosceles triangle, shape border rotate= 90, font=\small,
      draw},
    edge from parent/.style={ ->,draw, pos=0.65, font=\tiny},
    edgelabel/.style={ above=-1mm,sloped},
    ]
    \begin{scope}
      \node[] (Linfty) {$L^\alpha_\infty$};
      \node [konf, below of= Linfty] (root) {}
      child { 
        node[tree] {$A_{\alpha_1}$} }
      child { 
        node[tree] {$A_{\alpha_2}$} }
      child { 
        node[tree] {$A_{\alpha_3}$} }
      child[]{ node{\dots}  edge from parent[draw=none]}
      ;
    \end{scope}
    \begin{scope}[yshift=-5cm]
      
      \node[] (Lk) {$L^\alpha_k$};
      \node [konf, below of= Lk] (Lkroot) {}   
      child { 
        node[tree] {$A_{\alpha_1}$} }
      child { 
        node[tree] {$A_{\alpha_2}$} }
      child{ node{\dots} edge from parent[draw=none]}
      child { 
        node[tree] {$A_{\alpha_{k-1}}$} }
      child { 
        node[tree] {$E_{\alpha_k}$} }
      child { 
        node[tree] {$E_{\alpha_{k+1}}$} }
      child { 
        node[tree] {$E_{\alpha_{k+2}}$} }
      child[]{ node{\dots} edge from parent[draw=none]}
      ;

    \end{scope}
    \begin{scope}[yshift = -10cm]
      \node (TPN) {$T_{P,n}$};
      \node [konf, below of= TPN,] (Troot) {}
      child { 
        node[tree] {$T_{P_1, n}$}} 
      child { 
        node[tree] {$T_{P_2, n}$}
      }
      child { 
        node[tree] {$T_{P_3, n}$}
      }
      child{ node{\dots} edge from parent[draw=none]}
      ;
    \end{scope}

  \end{tikzpicture}

  \caption{Trees for Case 3 from the proof of Theorem~\ref{thm-hirschfeld-white}.}    
  \label{fig:Thm5.1Case3}
\end{figure}
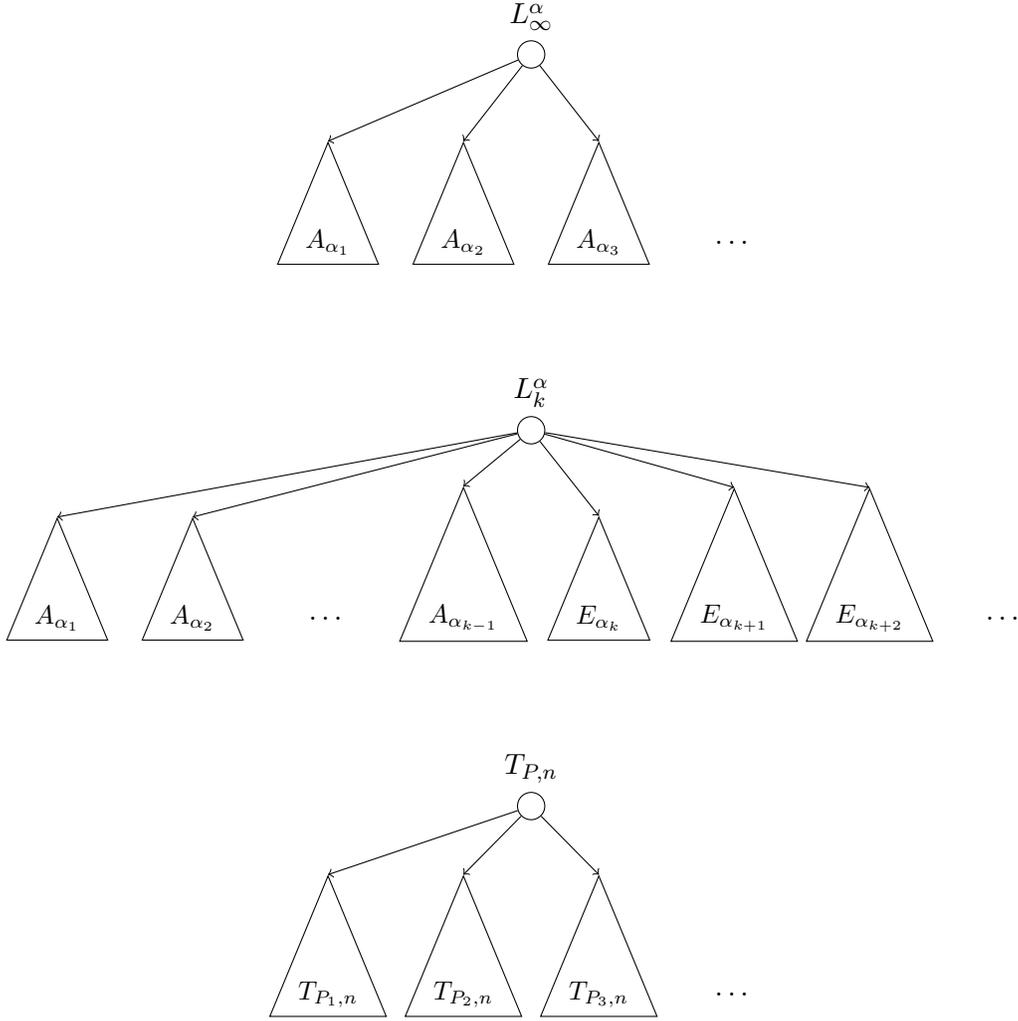
We define a computable subset $L^\alpha_\infty$ of the leaves 
$\leaves(\Sf_\alpha) =$ 
\mbox{$\{ k u \mid k \in \mathbb{N}_{>0}, u \in \leaves(\Sf_{\alpha_k})\}$} by
$$
L^\alpha_\infty = \{ ku \mid k \in \mathbb{N}_{>0},   u \in A_{\alpha_k}\}.
$$
For $k \in \mathbb{N}_{>0}$ we define 
the computable set $L^\alpha_k \subseteq \leaves(\Sf_\alpha)$ by
$$
L^\alpha_k = \{ xu \mid 0 < x < k,   u \in A_{\alpha_x}\} \cup \{ xu \mid x \geq k,   u \in E_{\alpha_x}\}.
$$
Recall that $(\Sf_{\alpha_k}, A_{\alpha_k}) \not\cong  (\Sf_{\alpha_k}, E_{\alpha_k})$ for all $k \geq 1$,
whence $(\Sf_\alpha,  L^\alpha_k) \not\cong (\Sf_\alpha,
L^\alpha_\infty)$ for all $k \geq 1$. 

For the $\Pi^0_\alpha$ index $(\Pi,a,e)$ we can assume (using a padding argument) that
$\Phi_e$ is a total computable function such that 
for all $k \geq 1$, $\Phi_e(k)$ is a $\Pi^0_{\alpha_k}$ index for some
set $Q_k \subseteq \N$ and $P = \bigcap_{k \geq 1} Q_k$.
Let us define for $k \in \mathbb{N}_{>0}$ the set  $P_k = \bigcap_{i=1}^k Q_k$.
Since $\Pi^0_{\alpha_k}$ sets are closed under finite intersections, $P_k$
is a $\Pi^0_{\alpha_k}$ set. Moreover, a $\Pi^0_{\alpha_k}$ index for $P_k$ can 
be computed from $k$.
By induction, from $k \in \mathbb{N}_{>0}$ and $n$ one can compute 
an index for a computable set $T_{P_k,n} \subseteq \leaves(\Sf_{\alpha_k})$
such that
\begin{align*} 
    (\Sf_{\alpha_k}, T_{P_k,n}) \cong \begin{cases}
      (\Sf_{\alpha_k}, A_{\alpha_k}) & \text{ if } n \in P_k, \\
      (\Sf_{\alpha_k}, E_{\alpha_k}) & \text{ if } n \notin P_k .
    \end{cases}
\end{align*}
Define the computable set
$$
T_{P,n} = \{ k u \mid k \in  \mathbb{N}_{>0}, u \in T_{P_k,n} \}.
$$
If $n \in P$, then $n \in P_k$ for all $k \in \mathbb{N}_{>0}$ and we get
$(\Sf_\alpha, T_{P,n}) \cong (\Sf_\alpha,  L^\alpha_\infty)$. On the other hand,
if $n \notin P$ then there is some $k \in \mathbb{N}_{>0}$ such that $n \in P_i$ if and only
if $i < k$. In this case we get $(\Sf_\alpha, T_{P,n}) \cong (\Sf_\alpha,  L^\alpha_k)$. 
\end{proof}

\subsection{Tree-automaticity of the trees $\Sf_{\omega^i}$}

In this section, we show that all trees $\Sf_{\omega^i}$ ($i \geq 1$)
from the previous section are tree-automatic.
For this, we need the following lemma.

\begin{lem}  \label{lemma-trees-for-omega^i}
Let $i \geq 1$, $n \geq 1$, $\alpha \leq \omega^i$. Then we have
\begin{equation} \label{def:Tf}
\Sf_{\alpha} \cup \{ uv \mid u \in \leaves(\Sf_{\alpha}), v \in \Sf_{\omega^i \cdot n}\}
= \Sf_{\omega^i \cdot n+\alpha} .
\end{equation}
\end{lem}
Note that the tree on the left-hand side of \eqref{def:Tf}
is the tree that results from the tree $\Sf_{\alpha}$
by replacing every leaf by a copy of the tree  $\Sf_{\omega^i \cdot n}$.

\begin{proof}
We prove the lemma by induction on $\alpha \leq \omega^i$. The case $\alpha = 0$ is clear.
Next, assume that $\alpha = \gamma+1$ is a successor ordinal. Then
\begin{equation} \label{def:Sf_alpha}
\Sf_{\alpha} = \{\varepsilon\} \cup \{n u \mid n \in \mathbb{N}_{>0}, u \in \Sf_\gamma\} .
\end{equation}
By induction hypothesis, we have
\begin{equation} \label{IH-for-gamma}
\Sf_{\gamma} \cup \{ uv \mid u \in \leaves(\Sf_{\gamma}), v \in \Sf_{\omega^i \cdot n}\}
= \Sf_{\omega^i \cdot n+\gamma} .
\end{equation}
Hence, we get
\begin{align*}
&\Sf_{\alpha} \cup  \{ uv \mid u \in \leaves(\Sf_{\alpha}), v \in \Sf_{\omega^i \cdot n}\}\\
 \overset{\eqref{def:Sf_alpha}}{=} &\{\varepsilon\} \cup
\{n u \mid n \in \mathbb{N}_{>0}, u \in \Sf_\gamma\}\cup 
 \{ nu'v \mid n \in \mathbb{N}_{>0}, u' \in \leaves(\Sf_{\gamma}), v \in \Sf_{\omega^i \cdot n}\} \\
\overset{\eqref{IH-for-gamma}}{=} &\{\varepsilon\} \cup 
\{ n w \mid n \in \mathbb{N}_{>0}, w \in \Sf_{\omega^i \cdot
  n+\gamma}\} = \Sf_{\omega^i \cdot n+\gamma+1} = \Sf_{\omega^i \cdot
  n+\alpha}. 
\end{align*}
Finally, assume that $\alpha \leq \omega^i$ is a limit ordinal with the fundamental sequence $(\alpha_k)_{k \geq 1}$.
Then $(\omega^i \cdot n + \alpha_k)_{k \geq 1}$ is our fundamental sequence for the ordinal $\omega^i \cdot n + \alpha$.
We have
\begin{equation} \label{def:Sf_alpha(2)}
\Sf_{\alpha} = \{\varepsilon\} \cup \{k u \mid k \in \mathbb{N}_{>0}, u \in \Sf_{\alpha_k}\} .
\end{equation}
By induction hypothesis, for every $k \geq 1$ we have
\begin{equation} \label{IH-for-alpha_k}
\Sf_{\alpha_k} \cup \{ uv \mid u \in \leaves(\Sf_{\alpha_k}), v \in \Sf_{\omega^i \cdot n}\}
= \Sf_{\omega^i \cdot n+\alpha_k} .
\end{equation}
Hence, we get
\begin{align*}
&\Sf_{\alpha} \cup \{ uv \mid u \in \leaves(\Sf_{\alpha}), v \in \Sf_{\omega^i \cdot n}\}\\
 \stackrel{\eqref{def:Sf_alpha(2)}}{=} &
\{\varepsilon\} \cup
\{k u \mid k \in \mathbb{N}_{>0}, u \in \Sf_{\alpha_k}\}  \cup 
 \{ ku'v \mid k \in \mathbb{N}_{>0}, u' \in \leaves(\Sf_{\alpha_k}), v \in \Sf_{\omega^i \cdot n}\} \\
 \stackrel{\eqref{IH-for-alpha_k}}{=} & \{\varepsilon\} \cup
\{ k w \mid k \in \mathbb{N}_{>0}, w \in \Sf_{\omega^i \cdot
  n+\alpha_k} \}  =  \Sf_{\omega^i \cdot n+\alpha} .
\end{align*}
\end{proof}
Now we can prove tree-automaticity of $\Sf_{\omega^i}$.

\begin{lem} \label{lemma-Sf_alpha-tree-auto}
For every $i \geq 1$, the tree $\Sf_{\omega^i} \subseteq \mathbb{N}_{>0}^*$ is
tree-automatic. Moreover, there is a unary tree-automatic copy
$(L,\leq)$ of $\Sf_{\omega^i}$ together with a computable isomorphism 
\mbox{$f : \Sf_{\omega^i} \to (L,\leq)$}.
\end{lem}

\begin{proof}
We prove the lemma by induction on $i$. Assume that we  have already constructed
a tree-automatic copy $(L, \leq)$ of the tree $\Sf_{\omega^i}$
over a unary alphabet (i.e., $L \subseteq \T^{\fin}_2$ is regular)
together with the computable isomorphism $f$.
In addition, we assume that the root of $(L, \leq)$ is the one-node tree $\{\varepsilon\}$;
this property is preserved by the construction.
We aim at constructing a unary tree-automatic copy of $\Sf_{\omega^{i+1}}$ with
root $\{\varepsilon\}$.
Let us first construct a tree-automatic copy of the computable forest $\bigsqcup_{n \geq 1} \Sf_{\omega^i \cdot n}$.
This forest is isomorphic to the inverse prefix relation on the domain 
\begin{equation*}
	\{ nu \mid n \geq 1, u \in \Sf_{\omega^i \cdot n} \} \subseteq
\mathbb{N}_{>0}^*.
\end{equation*}
Define well-founded trees $\Tf_n$ for $n \geq 1$ inductively as follows.
Let $\Tf_1 = \Sf_{\omega^i}$ and let
$\Tf_{n+1}$ result from the tree $\Sf_{\omega^i}$
by replacing every leaf by a copy of the tree  $\Tf_n$.
Formally, we define
$$
\Tf_n = \{ u_1 \cdots u_j v \mid 0 \leq j < n, u_1, \ldots, u_j \in
\leaves(\Sf_{\omega^i}), v \in \Sf_{\omega^i}\} .
$$
Lemma~\ref{lemma-trees-for-omega^i} implies $\Tf_n = \Sf_{\omega^i \cdot n}$ for $n \geq 1$.
We  construct a tree-automatic copy of $\bigsqcup_{n \geq 1}  \Tf_n$
using \mbox{$(L, \leq)$.} 
The universe of this copy is the set $L'$ of all trees  of the form
$\pref( \bigcup_{i=1}^n 0^i 1 t_i )$, where $n \geq 1$, $t_1, \ldots,
t_n \in L$ 
and there exists $1 \leq j \leq n$ such that
$t_j$ is a leaf of $(L, \leq)$ for all $j < i$ and
$t_j = \{\varepsilon\}$  for all $j > i$.
Since the set of leaves of $(L, \leq)$ is regular, the set
$L'$ is clearly regular.
We define a tree-automatic partial order $\leq'$ on the set $L'$ by
comparing the $t_i$ componentwise.
Let $t = \pref(\bigcup_{i=1}^m 0^i 1 t_i) \in L'$ and
$t' = \pref(\bigcup_{i=1}^n 0^i 1 t'_i) \in L'$.
Then $t \leq' t'$ if and only if $n = m$ and
$t_i \leq t'_i$ for
all $1 \leq i \leq n$.
From this construction, it follows easily that
\begin{equation*}
(L', \leq') \cong \bigsqcup_{n \geq 1}\Tf_n = \bigsqcup_{n \geq 1}\Sf_{\omega^i \cdot n}.  
\end{equation*}
The set of roots of this forest is $\{ \pref(\bigcup_{i=1}^n 0^i 1)
\mid n \geq 1 \}$.
Let us also define a computable isomorphism $f' :
\bigsqcup_{n \geq 1}\Tf_n \to (L', \leq')$.
Take an element $nw$ from $\bigsqcup_{n \geq 1}\Tf_n$, where
\mbox{$n \in \mathbb{N}_{>0}$} and $w \in \Tf_n$. There is a unique
factorization $w = u_1 \cdots u_j v$ with $0 \leq j < n$,
$u_1, \ldots, u_j \in\leaves(\Sf_{\omega^i})$, and $v \in \Sf_{\omega^i}$.
Since the set $\leaves(\Sf_{\omega^i})$ is computable, we can
compute this factorization. Next, using the computable isomorphism
$f : \Sf_{\omega^i} \to (L,\leq)$ define the trees $t_k = f(u_k)$
($1 \leq k \leq j$), $t_{j+1} = f(v)$, and $t_{j+2}, \ldots, t_n =
\{\varepsilon\}$. Then
set $f'(nw) = \pref(\bigcup_{i=1}^n 0^i 1 t_i)$. It is straightforward
to verify that this defines indeed an isomorphism from
$\bigsqcup_{n \geq 1}\Tf_n$ to $(L', \leq')$. 

We derive from $(L', \leq')$ a tree-automatic copy
of the computable forest $\bigsqcup_{n \geq 1} \Sf_{\omega^i \cdot n+1}$. This forest
is isomorphic to the inverse prefix relation on
$$
\{ nu \mid n \in \mathbb{N}_{>0}, u \in \Sf_{\omega^i \cdot n+1}\} =
\mathbb{N}_{>0} \cup \{ nmw \mid n,m \in \mathbb{N}_{>0}, w \in \Sf_{\omega^i \cdot n}\} .
$$
Note that in every node $u \in \T^{\fin}_2$ of the forest $(L', \leq')$, the root of $u$ has no
right child (i.e., $1 \notin u$).
Define the regular set of trees
$$
L''  =
 \{  \pref(0^n) \mid n \geq 1\} \cup
 \{  \pref( \{1^m\} \cup t) \mid t \in L', m \geq 1 \} .
$$
On the set $L''$ we define the partial order $\leq''$ as follows.
For $u,v \in L''$, let $u \leq'' v$ if and only if either
\begin{iteMize}{$\bullet$}
\item $u = \pref( \{1^m\} \cup s)$, $v = \pref( \{1^m\} \cup t)$ with
  $s,t \in L'$ and $s \leq' t$ or
\item $u =  \pref(\{1^m\} \cup t)$, $v = \pref(0^n)$ with $t\in L'$
of the form $\pref(\bigcup_{i=1}^n 0^i 1 t_i)$.
\end{iteMize}
This order relation is clearly tree-automatic.
Moreover, the construction implies that
$$
(L'', \leq'') \cong
( \N_{>0} \cup \{nmw \mid n,m \in \mathbb{N}_{>0}, w \in \Sf_{\omega^i
  \cdot n} \}, \succeq) = \bigsqcup_{n \geq 1} \Sf_{\omega^i \cdot n+1} .
$$
A computable isomorphism $f'' : \bigsqcup_{n \geq 1} \Sf_{\omega^i \cdot n+1} \to
(L'', \leq'')$ can be defined as follows.
For a root $n \in \mathbb{N}_{>0}$ let $f''(n) = \pref(0^n)$.
A node $nmu$ with $n,m \in \mathbb{N}_{>0}$ and $u \in \Sf_{\omega^i \cdot n}$
(hence, $mu \in \Sf_{\omega^i \cdot n+1}$)
is mapped to $f''(nmu) = \pref(\{1^m\} \cup t)$ with $f'(nu) = t$. 

Finally, we add to the forest $(L'', \leq'')$ the root $\{\varepsilon\}$; this gives us a tree-automatic
copy of the tree $\Sf_{\omega^{i+1}}$ with root $\{\varepsilon\}$. A computable isomorphism
is obtained by extending $f''$ by $f''(\varepsilon) = \{\varepsilon\}$.
\end{proof}

\subsection{Encoding $\Sigma^0_2$-sets of binary trees}

Theorem~\ref{thm-hirschfeld-white} and
Lemma~\ref{lemma-Sf_alpha-tree-auto}
show that the isomorphism problem for the
following class of computable structures is $\Pi^0_{\omega^i}$-hard
for every $i\in \N$: the class contains all structures of the
form $(V,\sqsubseteq, X)$ where $(V,\sqsubseteq)$ is the unary
tree-automatic copy of $\Sf_{\omega^i}$ from Lemma~\ref{lemma-Sf_alpha-tree-auto}
and $X$ is a computable unary predicate, which is moreover a
subset of ${\leaves(V,\sqsubseteq)}$. By
Lemma~\ref{lemma-unary-alphabet} we can moreover assume that
$V \subseteq \T_{\bin}$, i.e., $V$ consists of unlabeled full
binary trees.

Let us define the set
$$
\T_{\lef} = \{ \{\varepsilon\} \cup 0u \mid u \in \T_{\bin} \} .
$$
Thus, $\T_{\lef}$ contains all trees $t \in \T^{\fin}_2$, where
the root of $t$ has no right child ($1 \notin t$), the root has
a left child ($0 \in t$), and the subtree rooted at $0$ belongs to
$\T_{\bin}$, i.e., is a full binary tree.

In this section,
we  describe an encoding of $\Sigma^0_2$-subsets of $\T_{\lef}$
by sets of tree-automatic trees of height 3.
Actually, we  need this encoding only for computable subsets of
$\T_{\lef}$ (instead of $\Sigma^0_2$-sets), but the proof
of Lemma~\ref{lem:F_B} is not simpler for a computable set $B$.

\begin{lem}\label{lem:F_B}
  There are two trees $\Uf_0$ and $\Uf_1$ of height 3
  ($\Uf_0 \not\cong \Uf_1$) with the
  following property: From a given index of a $\Sigma^0_2$-set
  $B\subseteq \T_{\lef}$ one can
  effectively construct a tree-automatic forest $\Ff_B$ of height 3 such
  that:
  \begin{iteMize}{$\bullet$}
  \item The set of roots of $\Ff_B$ is $\T_{\lef}$.
  \item For every $t \in \T_{\lef}$, $\Ff_B(t) \cong \Uf_0$ if
    $t\in B$ and $\Ff_B(t)\cong \Uf_1$ if $t\notin B$.
  \end{iteMize}
\end{lem}
Restricting to trees from $\T_{\lef}$
makes our encoding technically a bit simpler, and this restriction can be easily enforced
later when we apply Lemma~\ref{lem:F_B}.

We prove Lemma~\ref{lem:F_B} using a similar statement for words from \cite{KuLiLo11}.
First, we have to introduce a notation from  \cite{KuLiLo11}.
Let $\A = (Q,\Sigma, \Delta, I, F)$ be a finite nondeterministic
\emph{string} automaton, where $Q$ is the set of states, $\Sigma$ is the input
alphabet, $\Delta \subseteq Q \times \Sigma \times Q$ is the set of transition triples,
$I$~is the set of initial states, and $F$ is the set of final states.
A \emph{successful run} of $\A$ on a non-empty word $w \in \Sigma^+$ is a word
$(q_0, a_1, q_1) (q_1, a_2, q_2) \cdots (q_{n-1}, a_n, q_n)$ over $\Delta$ such that
$q_0 \in I$, $q_n \in F$, and $w = a_1 a_2 \cdots a_n$.
The language $L(\A)$ accepted by $\A$ consists of all non-empty words
(for technical reasons, the empty word was excluded in \cite{KuLiLo11})
that have a successful run.
We define a forest $\forest(\A)$ in the following. Clearly,
$(L(\A), \succeq)$ (recall that $\preceq$ is the
prefix relation on words) is a (string-automatic) forest.  
The set of all
leaves of the forest $(L(\A), \succeq)$ is first-order definable
whence it is a regular language. It is the set of all words in $L(\A)$
that are not a proper prefix of another word in $(L(\A)$.
Let us denote the set of leaves of $(L(\A), \succeq)$
with $\leaves(\A)$.
We define the forest $\forest(\A)$ as follows:
\begin{iteMize}{$\bullet$}
\item The domain of $\forest(\A)$ is the regular set
$$
L(\A) \sqcup \{ r \in \Delta^+ \mid r
  \text{ is a successful run of $\A$ on some $v \in \leaves(\A)$} \}.
$$
\item The order relation $\leq$ of $\forest(\A)$ is defined by 
$u \leq v$ if
\begin{enumerate}[(1)]
\item either $u,v \in L(\A)$ and $v \preceq u$ or
\item $v \in L(\A)$ and $u \in \Delta^+$  is a successful
  run of $\A$ on some $w \in \leaves(\A)$ with $v \preceq w$.
\end{enumerate}
\end{iteMize}
Clearly, $\forest(\A)$ is string-automatic. Intuitively, we take the forest
resulting from the inverse prefix order on the regular language $L(\A)$ and
append to each leaf $v$ of $(L(\A), \succeq)$ all successful runs of $\A$
on $v$ as children. All these children are leaves in
$\forest(\A)$.  In \cite{KuLiLo11}, the following lemma was proved.

 \begin{lem} \label{lemma:U_0andU_1}
 There exist two trees $\Uf_0$ and
  $\Uf_1$ of height 3 ($\Uf_0 \not\cong \Uf_1$) with the following property:
  From a given index of a $\Sigma^0_2$-set $A \subseteq
  \{0,1\}^*1$ one can effectively construct a finite string automaton $\A$
  (over an alphabet $\Sigma$ with $0,1, \sharp \in \Sigma$) such that 
  $\forest(\A)$ is
  a forest of height 3 with the following properties.
  \begin{iteMize}{$\bullet$}
  \item The set of roots of $\forest(\A)$ is $\{0,1\}^*1\sharp$.
  \item For every $w \in \{0,1\}^*1$, 
    $\forest(\A)(w\sharp) \cong \Uf_0$     if $w \in  A$, 
    and $\forest(\A)(w\sharp) \cong \Uf_1$ if  $w \notin A$. 
  \end{iteMize}
\end{lem}
In order to prove Lemma~\ref{lem:F_B} using Lemma~\ref{lemma:U_0andU_1},
we have to encode trees from $\T_{\lef}$ by words.
A tree $t \in \T_{\bin}$ can be encoded by a
non-empty bracket expression, i.e., a word over the alphabet $\{ (, )\}$.
Here, we view such a  bracket expression as a binary string by
identifying $($ with $0$ and $)$ with $1$. Thus, we define
a nonempty word $\word(t) \in \{0,1\}^+$ as follows: Consider a
depth-first left-to-right traversal of $t$. Each time, we move from a
node $v$ to one of its children, we write down
$0$. Each time, we move from a node $vi$ to its
parent node $v$, we write down $1$. The resulting word is
$\word(t)$.
Formally, let $\word(\{\varepsilon\}) = \varepsilon$ and
for $t\in \T_{\bin} \setminus \{\{\varepsilon\}\}$
such  that $t=\{\varepsilon\} \cup 0t_1\cup 1t_2$
let
$$
\word(t) = 0\word(t_1)10\word(t_2)1 .
$$
This mapping $\word$ is clearly injective. Finally, for $t =
(\{\varepsilon\} \cup 0u)  \in \T_{\lef}$ with
$u \in \T_{\bin}$ let $\word(t) = 0 \word(u) 1$.
Also the mapping $\word :  \T_{\lef} \to  0\{0,1\}^*1$
is injective.

Let us now fix an alphabet $\Sigma$ such that $0,1, \sharp \in \Sigma$
(as in Lemma~\ref{lemma:U_0andU_1}).
Take a word $w = u \sharp v$ with $u \in \{0,1\}^*$, $v \in \Sigma^*$ and
$u = \word(t)$ for some tree $t \in  \T_{\lef}$.
Note that the root of $t$ has no right child in $t$.
For $1 \leq i \leq |v|$ let $a_i$ be the $i^{th}$ symbol of $v$.
We encode
$w$ by the $\Sigma$-labeled tree $\tree(w)  = (T,\lambda) \in
\T^{\fin}_{2,\Sigma}$, where
\begin{align*}
T  &= t \cup \{ 1^i \mid 1 \leq i \leq |v| \} \text{ and}\\
\lambda(x)  &= \begin{cases}
    a_i  & \text{ if $x = 1^{i}$ for some  $1 \leq i \leq |v|$}, \\
    \sharp & \text{ else.}
\end{cases}
\end{align*}
Note that $\tree( \word(t) \sharp) = t$, since we identify an unlabeled tree
with a tree where all nodes are labeled with $\sharp$.

\begin{lem} \label{lemma-string->tree}
From a given string automaton  $\A$ over $\Sigma$
one can construct effectively a tree automaton $\B$ over $\Sigma$
such that for every tree $t \in \T_{\lef}$ and every word $v \in \Sigma^*$
the following holds: The number of successful runs of $\A$ on the string
$\word(t)\sharp v$ equals the
number of successful runs of $\B$ on the tree $\tree(\word(t)\sharp v)$.
\end{lem}

\begin{proof}
Let $\A = (Q,\Sigma, \Delta, I, F)$.
Note that $\tree(\word(t)\sharp v)$ consists of the tree $t$ to which we add at the root
a right branch of length $|v|$. The $i^{th}$ node on this branch 
 is labeled with
the $i^{th}$ symbol of $\sharp v$ (we count from the root to the leaf).
Essentially, the tree automaton $\B$ simulates a
tree-walking automaton $\mathcal{W}$ (see \cite{BojaTWA:08} for a
survey on tree walking automata, but we do not need a formal definition of tree walking automata) that
walks over the tree $\tree(\word(t)\$v)$
in  depth-first left-to-right order. Thereby, $\mathcal{W}$ simulates
the string automaton $\A$. The automaton $\mathcal{W}$ starts in the root of the tree.
In a first phase (which is finished if $\mathcal{W}$ returns to the
root for the first time), $\mathcal{W}$
behaves as follows: each time, $\mathcal{W}$ moves down in the tree (towards the leaves),
it simulates a $0$-labeled transition of $\A$, and each time, $\mathcal{W}$ moves up in the tree (towards the root),
it simulates a $1$-labeled transition of $\A$.
After the first phase, the tree $t$ is fully traversed and
$\mathcal{W}$ goes into the right branch of  $\tree(\word(t)\sharp v)$
(which is labeled with the word $\sharp v$) and continues the simulation of $\A$.

Here is a formal definition of the tree automaton
$\B = (Q', \Delta', I', F')$, which simulates
a tree-walking automaton $\mathcal{W}$ with the above behavior.
Fix an arbitrary final state $ q_{\mathsf{f}} \in F$.
The state set of $\B$ is
$$
Q' = (Q \times Q \times Q) \cup (I \times Q) \cup Q ,
$$
the set of initial states is $I' = I \times Q$, and the set
of final states is $F$.
The set of transitions is $\Delta' = \Delta_1 \cup \Delta_2 \cup \Delta_3$, where
\begin{align*}
\Delta_1 =& \{ (  (p_1, p_2), \sharp, (q_1,q_2,q_3), r) \mid p_1 \in 
I, (p_1,0,q_1) \in \Delta, q_2\in Q, (q_3,1,p_2), (p_2,\sharp,r) \in \Delta \},
\\ 
\Delta_2 =& \{ ((p_1,p_2,p_3), \sharp, (q_1,q_2,q_3), (r_1,r_2,r_3))
\mid\\
&\quad  (p_1,0,q_1), q_2\in Q, (q_3,1,p_2), (p_2,0,r_1),r_2\in Q,
(r_3,1,p_3) \in \Delta\} 
\\ 
 &\cup   \{ ((p,p,p),\sharp,  q_{\mathsf{f}},q_{\mathsf{f}}) \mid p
 \in Q \},   \\ 
\Delta_3 =& \{ (p,a,q_{\mathsf{f}},q) \mid (p,a,q) \in \Delta \}. 
\end{align*}
With the transitions in $\Delta_1$ we split the simulation of the tree-walking automaton into its first and second
phase, i.e., $p_2$ in $\Delta_1$ is the state reached by the  tree-walking automaton after traversing the tree $t$.
The transitions in $\Delta_2$ simulate the traversal of $t$, whereas the transitions in $\Delta_3$ simulate the
the string automaton $\A$ on the right $\sharp v$-labeled branch.
\end{proof}
Let us now prove Lemma~\ref{lem:F_B}.

\begin{proof}
  [Proof of Lemma \ref{lem:F_B}]
 Fix a $\Sigma^0_2$-set $B\subseteq \T_{\lef}$.
 Then the set $A = \word(B)  \subseteq 0\{0,1\}^*1$ belongs to $\Sigma^0_2$ as well (the range of the $\word$-mapping is
 computable and on its range, the inverse of $\word$ is also computable).
 Therefore, we can apply Lemma~\ref{lemma:U_0andU_1}  to the set $\word(B)$.
 We obtain (effectively)  a finite string automaton $\A$
  (over an alphabet $\Sigma$ with $0,1, \sharp \in \Sigma$) such that 
  $\forest(\A)$ is
  a forest of height 3 with the following properties.
  \begin{iteMize}{$\bullet$}
  \item The set of roots of $\forest(\A)$ is $\{0,1\}^*1\sharp$.
  \item For every $w \in \{0,1\}^*1$, 
    $\forest(\A)(w\sharp) \cong \Uf_0$ if $w \in A$, and 
    $\forest(\A)(w\sharp) \cong \Uf_1$ if $w \notin A$
    for two
    nonisomorphic trees $\Uf_0$ and $\Uf_1$.
  \end{iteMize}
To the string automaton $\A$ we next apply Lemma~\ref{lemma-string->tree}.
We obtain (effectively) a tree automaton $\B$ over $\Sigma$
such that for every tree $t \in \T_{\lef}$ and every word $v \in \Sigma^*$
the following holds: The number of successful runs of $\A$ on the string $\word(t)\sharp v$ equals the
number of successful runs of $\B$ on the tree $\tree(\word(t)\sharp v)$.
Since $\A$ accepts every word from $\{0,1\}^*1\sharp$ (this set is the set of roots of
$\forest(\A)$), $\B$ accepts every tree $t \in \T_{\lef}$.

By taking the product with a deterministic tree automaton that accepts the set of
trees
\begin{equation*}
	\{ \tree(\word(t)\sharp v) \mid t \in \T_{\lef}, v \in \Sigma^*\}
\end{equation*}
(which is regular),
we can assume that
$$
L(\B) \subseteq \{ \tree(\word(t)\sharp v) \mid t \in \T_{\lef}, v \in \Sigma^*\}.
$$
For trees $t_1 = \tree(\word(t)\sharp v_1)$, $t_2 = \tree(\word(t)\sharp v_2)$ let us write
$t_1 \sqsubseteq t_2$ if $v_1$ is a prefix of $v_2$. Clearly, this is a tree-automatic relation.
Let
$$
\max(\B) = \{ t \in L(\B) \mid \text{there does not exist $t' \in
  L(\B)$ with $t \sqsubset t'$} \} ; 
$$
this set is regular as well.
We can now construct a tree-automatic forest $\Ff_B$ of height 3 as follows.
The set of nodes of $\Ff_B$ is
$$
L(\B) \cup \bigcup_{t \in \max(\B)} \Run(\B, t) .
$$
Since $\max(\B)$ is regular, this set is also regular.
The order relation of the forest $\Ff_B$ is the tree-automatic relation
$$
\sqsupseteq \;\cup\; \{ (\rho, t) \mid  t\in L(\B), \exists u \in
\max(\B) : t \sqsubseteq u,   \rho \in  \Run(\B, u)\} . 
$$
The set of roots of $\Ff_B$ is (as required) $\T_{\lef}$.
Moreover, for every tree $t \in \T_{\lef}$, the construction
directly implies that $\Ff_B(t) \cong \forest(\A)(\word(t)\sharp)$. Hence, for every tree
$t \in \T_{\lef}$ we have $\Ff_B(t) \cong \Uf_0$ if and only if 
$\forest(\A)(\word(t)\sharp) \cong \Uf_0$ if and only if
$\word(t) \in A = \word(B)$ if and only if $t \in B$, and similarly,
$\Ff_B(t) \cong \Uf_1$ if and only if $t \notin B$.
\end{proof}

\subsection{Hardness for the isomorphism problem}

Hardness of the isomorphism problem for well-founded tree-automatic
trees is established through the following theorem.

\begin{thm}\label{thm:hardness}
  From a given $i \in \N_{>0}$, one can compute a well-founded tree-automatic
  tree $\Vf_i$ such that the following holds: From a given
  $\Pi^0_{\omega^i}$-set $P\subseteq\mathbb{N}$ (represented by a
  $\Pi^0_{\omega^i}$ index) and $n \in \mathbb{N}$ one can compute a
  well-founded tree-automatic tree $\Wf_{P,n}$ such that
  $n \in P$ if and only if $\Vf_i \cong \Wf_{P,n}$
\end{thm}

\begin{proof}
Fix $i \geq 1$, an arbitrary $\Pi^0_{\omega^i}$-set $P\subseteq
\N$, and $n \in \N$.
According to  Theorem~\ref{thm-hirschfeld-white} there exists
effectively a computable subset $L_i = L^{\omega^i}_\infty \subseteq \leaves(\Sf_{\omega^i})$
such that the following holds. From an index for $P$ and $n$
one can compute an index for a computable
subset $T_{P,n} \subseteq \leaves(\Sf_{\omega^i})$ such that
$(\Sf_{\omega^i}, L_i) \cong (\Sf_{\omega^i}, T_{P,n})$  if and
only if $n \in P$. By Lemma~\ref{lemma-Sf_alpha-tree-auto}, the tree $\Sf_{\omega^i}$ is
tree-automatic and there exists a unary tree-automatic copy
$(S, \leq)$ of $\Sf_{\omega^i}$ together with a computable
isomorphism $f : \Sf_{\omega^i} \to (S, \leq)$. Moreover, by Lemma~\ref{lemma-unary-alphabet}
we can assume that $S \subseteq \T_{\bin}$. Finally,
by applying the computable bijection $t \mapsto \{\varepsilon\} \cup 0t$ from
$\T_{\bin}$ to $\T_{\lef}$, we can even assume that
$S \subseteq \T_{\lef}$. 
Since $f$ is computable and bijective and $L_i$ and $T_{P,n}$ 
are computable, the sets $L_i' = f(L_i) \subseteq S$ and $T'_{P,n} = f(T_{P,n}) \subseteq S$
are computable as well and in particular $\Sigma^0_2$.
We have $n \in P$ if and only if $(S, \leq, L_i') \cong (S, \leq, T'_{P,n})$.

Next, using Lemma~\ref{lem:F_B} we obtain two trees $\Uf_0$ and $\Uf_1$ of height 3
and from the indices of the computable sets $L'_i$ and $T'_{P,n}$, we
can compute two tree-automatic forests $\Gf_i$ and $\Hf_{P,n}$ of height 3
such that the following holds.
  \begin{iteMize}{$\bullet$}
  \item The set of roots of $\Gf_i$ (resp. $\Hf_{P,n}$) is $\T_{\lef}$.
  \item For every $t \in \T_{\lef}$,
    $\Gf_i(t)\cong \Uf_0$ if $t \in L_i'$ and $\Gf_i(t)\cong
    \Uf_1$ otherwise.
  \item For every $t \in  \T_{\lef}$,
    $\Hf_{P,n}(t)\cong \Uf_0$ if $t \in T'_{P,n}$ and $\Hf_{P,n}(t)\cong
    \Uf_1$ otherwise.
  \end{iteMize}
  Note that by Theorem~\ref{thm:tree-automatic}, $\leaves(S, \leq)$
  is a regular set of trees since it is first-order definable in the
  tree-automatic tree $(S,\leq)$. Let $\Gf_i'$ (resp. $\Hf'_{P,n}$) be the restriction of
  the forest $\Gf_i$ (resp. $\Hf_{P,n}$) to those trees with a root from
  $\leaves(S, \leq)$. It follows that $\Gf_i'$ and $\Hf'_{P,n}$ are again
  tree-automatic forests of height 3. Moreover, we can assume that the
  intersection of the domains of $\Gf_i'$ (resp. $\Hf'_{P,n}$) and $S$ equals
  $\leaves(S, \leq)$. Finally, let $\Vf_i$  (resp. $\Wf_{P,n}$)
  be the well-founded
  tree-automatic tree obtained from the union of $\Gf_i'$ (resp. $\Hf'_{P,n}$) and $(S,\leq)$.
  Hence, $\Vf_i$
  (resp. $\Wf_{P,n}$) results from the tree $(S,\leq)$
   by (i) replacing every leaf which belongs to
  $L_i'$ (resp. $T'_{P,n}$) by the tree-automatic
  height-3 tree $\Uf_0$ and by (ii) replacing every leaf which does not belong to
  $L_i'$ (resp. $T'_{P,n}$) by the tree-automatic
  height-3 tree $\Uf_1$. Since $\Uf_0 \not\cong \Uf_1$, we have
  $n \in P$ if and only if $(S, \leq, L_i') \cong (S, \leq, T'_{P,n})$
  if and only if $\Vf_i \cong \Wf_{P,n}$.
\end{proof}

\begin{thm}\label{thm:hardness2}
   The isomorphism problem for well-founded tree-automatic trees is
  $\Delta^0_{\omega^\omega}$-hard under Turing-reductions.
\end{thm}

\begin{proof}
  Let $\Phi_e$ be a total computable function that maps
  $i \in \N$ to an ordinal notation $a_i \in O$ with $|a_i|_O = \omega^i$.
  Hence $3 \cdot 5^e$ is an ordinal notation for $\omega^\omega$.
  Recall that $\Delta^0_{\omega^\omega}$ consists of all sets that are
  Turing-reducible to
  $$
  H(3 \cdot 5^e) = \{ \langle a,n\rangle \mid a <_O 3 \cdot 5^e, n \in
  H(a) \}.
  $$
  This set is Turing-reducible to
  $$
  A =  \{ \langle i, n\rangle \mid i \geq 1, n \in H(a_i)\} .
  $$
  To see this, take a pair $\langle a,n\rangle$.
  First, check whether $a <_O 3 \cdot 5^e$. Since the set
  $\{ b \in O \mid b <_O 3 \cdot 5^e\}$ is computably enumerable
  \cite[Prop.~4.10]{AshKNight00}, this is effectively possible using
  the halting problem as an oracle. Clearly, the halting problem
  is computable in $H(a_1) = H(\omega)$ and hence in $A$.
  If $a <_O 3 \cdot 5^e$, we can compute
  effectively $i \in \N$ with $a <_O a_i$: Simply enumerate
  all sets $B_i = \{ b \in O \mid b <_O a_i \}$ until $a$ is found.
  Having $i \in \N$ with $a <_O a_i$ we can finally compute
  $m$ such that $m \in H(a_i)$ if and only if $n \in H(a)$
  (more precisely, from $a,a_i \in O$ one can compute an index for a many-one
  reduction of $H(a)$ to $H(a_i)$ \cite[p.~437]{Rogers}).
  
  Finally, we reduce the set $A$ to the
  isomorphism problem for well-founded tree-automatic trees.
  Take a pair $\langle i, n\rangle$. From $i, n$ (and a
  $\Pi^0_{\omega^i}$ index for the $\Delta^0_{\omega^i}$-set $H(a_i)$, which can
  be computed from $i$) we can compute by Theorem~\ref{thm:hardness}
  two well-founded tree-automatic trees $\Vf$ and
  $\Wf$ such that $n \in H(a_i)$ if and only if
  $\Vf \cong \Wf$. Hence,
  $\langle i, n\rangle \in A$ if and only if
  $\Vf \cong \Wf$. This proves the theorem.
\end{proof}

\section*{Acknowledgment}

We thank Dietrich Kuske and the anonymous referees of this paper and its 
conference version for their useful comments. In particular, we thank
one of the referees for simplifying our arguments in Example
\ref{exa:BoxDestroysRanksofWfPO}.


\end{document}